\documentclass{lmcs} %%% last changed 2014-08-20


\usepackage{tikz-cd}
\usepackage{hyperref}
\usepackage{framed}
\usepackage{stmaryrd}
\usepackage{amssymb}
\usepackage{bm}
\usepackage[all]{xy}
\usepackage{todonotes}
\usepackage[normalem]{ulem}

\colorlet{darkgreen}{green!60!black}

\usepackage{xcolor}
\usepackage{xspace}

\newdimen\proofrulebreadth \proofrulebreadth=.05em
\newdimen\proofdotseparation \proofdotseparation=1.25ex
\newdimen\proofrulebaseline \proofrulebaseline=2ex
\newcount\proofdotnumber \proofdotnumber=3
\let\then\relax
\def\hfi{\hskip0pt plus.0001fil}
\mathchardef\squigto="3A3B
\newif\ifinsideprooftree\insideprooftreefalse
\newif\ifonleftofproofrule\onleftofproofrulefalse
\newif\ifproofdots\proofdotsfalse
\newif\ifdoubleproof\doubleprooffalse
\let\wereinproofbit\relax
\newdimen\shortenproofleft
\newdimen\shortenproofright
\newdimen\proofbelowshift
\newbox\proofabove
\newbox\proofbelow
\newbox\proofrulename
\def\shiftproofbelow{\let\next\relax\afterassignment\setshiftproofbelow\dimen0 }
\def\shiftproofbelowneg{\def\next{\multiply\dimen0 by-1 }%
\afterassignment\setshiftproofbelow\dimen0 }
\def\setshiftproofbelow{\next\proofbelowshift=\dimen0 }
\def\setproofrulebreadth{\proofrulebreadth}

\def\prooftree{% NESTED ZERO (\ifonleftofproofrule)
\ifnum  \lastpenalty=1
\then   \unpenalty
\else   \onleftofproofrulefalse
\fi
\ifonleftofproofrule
\else   \ifinsideprooftree
        \then   \hskip.5em plus1fil
        \fi
\fi
\bgroup% NESTED ONE (\proofbelow, \proofrulename, \proofabove,
\setbox\proofbelow=\hbox{}\setbox\proofrulename=\hbox{}%
\let\justifies\proofover\let\leadsto\proofoverdots\let\Justifies\proofoverdbl
\let\using\proofusing\let\[\prooftree
\ifinsideprooftree\let\]\endprooftree\fi
\proofdotsfalse\doubleprooffalse
\let\thickness\setproofrulebreadth
\let\shiftright\shiftproofbelow \let\shift\shiftproofbelow
\let\shiftleft\shiftproofbelowneg
\let\ifwasinsideprooftree\ifinsideprooftree
\insideprooftreetrue
\setbox\proofabove=\hbox\bgroup$\displaystyle % NESTED TWO
\let\wereinproofbit\prooftree
\shortenproofleft=0pt \shortenproofright=0pt \proofbelowshift=0pt
\onleftofproofruletrue\penalty1
}

\def\eproofbit{% NESTED TWO
\ifx    \wereinproofbit\prooftree
\then   \ifcase \lastpenalty
        \then   \shortenproofright=0pt  % 0: some other object, no indentation
        \or     \unpenalty\hfil         % 1: empty hypotheses, just glue
        \or     \unpenalty\unskip       % 2: just had a tree, remove glue
        \else   \shortenproofright=0pt  % eh?
        \fi
\fi
\global\dimen0=\shortenproofleft
\global\dimen1=\shortenproofright
\global\dimen2=\proofrulebreadth
\global\dimen3=\proofbelowshift
\global\dimen4=\proofdotseparation
\global\count255=\proofdotnumber
$\egroup  % NESTED ONE
\shortenproofleft=\dimen0
\shortenproofright=\dimen1
\proofrulebreadth=\dimen2
\proofbelowshift=\dimen3
\proofdotseparation=\dimen4
\proofdotnumber=\count255
}

\def\proofover{% NESTED TWO
\eproofbit % NESTED ONE
\setbox\proofbelow=\hbox\bgroup % NESTED TWO
\let\wereinproofbit\proofover
$\displaystyle
}%
\def\proofoverdbl{% NESTED TWO
\eproofbit % NESTED ONE
\doubleprooftrue
\setbox\proofbelow=\hbox\bgroup % NESTED TWO
\let\wereinproofbit\proofoverdbl
$\displaystyle
}%
\def\proofoverdots{% NESTED TWO
\eproofbit % NESTED ONE
\proofdotstrue
\setbox\proofbelow=\hbox\bgroup % NESTED TWO
\let\wereinproofbit\proofoverdots
$\displaystyle
}%
\def\proofusing{% NESTED TWO
\eproofbit % NESTED ONE
\setbox\proofrulename=\hbox\bgroup % NESTED TWO
\let\wereinproofbit\proofusing
\kern0.3em$
}

\def\endprooftree{% NESTED TWO
\eproofbit % NESTED ONE
  \dimen5 =0pt% spread of hypotheses
\dimen0=\wd\proofabove \advance\dimen0-\shortenproofleft
\advance\dimen0-\shortenproofright
\dimen1=.5\dimen0 \advance\dimen1-.5\wd\proofbelow
\dimen4=\dimen1
\advance\dimen1\proofbelowshift \advance\dimen4-\proofbelowshift
\ifdim  \dimen1<0pt
\then   \advance\shortenproofleft\dimen1
        \advance\dimen0-\dimen1
        \dimen1=0pt
        \ifdim  \shortenproofleft<0pt
        \then   \setbox\proofabove=\hbox{%
                        \kern-\shortenproofleft\unhbox\proofabove}%
                \shortenproofleft=0pt
        \fi
\fi
\ifdim  \dimen4<0pt
\then   \advance\shortenproofright\dimen4
        \advance\dimen0-\dimen4
        \dimen4=0pt
\fi
\ifdim  \shortenproofright<\wd\proofrulename
\then   \shortenproofright=\wd\proofrulename
\fi
\dimen2=\shortenproofleft \advance\dimen2 by\dimen1
\dimen3=\shortenproofright\advance\dimen3 by\dimen4
\ifproofdots
\then
        \dimen6=\shortenproofleft \advance\dimen6 .5\dimen0
        \setbox1=\vbox to\proofdotseparation{\vss\hbox{$\cdot$}\vss}%
        \setbox0=\hbox{%
                \advance\dimen6-.5\wd1
                \kern\dimen6
                $\vcenter to\proofdotnumber\proofdotseparation
                        {\leaders\box1\vfill}$%
                \unhbox\proofrulename}%
\else   \dimen6=\fontdimen22\the\textfont2 % height of maths axis
        \dimen7=\dimen6
        \advance\dimen6by.5\proofrulebreadth
        \advance\dimen7by-.5\proofrulebreadth
        \setbox0=\hbox{%
                \kern\shortenproofleft
                \ifdoubleproof
                \then   \hbox to\dimen0{%
                        $\mathsurround0pt\mathord=\mkern-6mu%
                        \cleaders\hbox{$\mkern-2mu=\mkern-2mu$}\hfill
                        \mkern-6mu\mathord=$}%
                \else   \vrule height\dimen6 depth-\dimen7 width\dimen0
                \fi
                \unhbox\proofrulename}%
        \ht0=\dimen6 \dp0=-\dimen7
\fi
\let\doll\relax
\ifwasinsideprooftree
\then   \let\VBOX\vbox
\else   \ifmmode\else$\let\doll=$\fi
        \let\VBOX\vcenter
\fi
\VBOX   {\baselineskip\proofrulebaseline \lineskip.2ex
        \expandafter\lineskiplimit\ifproofdots0ex\else-0.6ex\fi
        \hbox   spread\dimen5   {\hfi\unhbox\proofabove\hfi}%
        \hbox{\box0}%
        \hbox   {\kern\dimen2 \box\proofbelow}}\doll%
\global\dimen2=\dimen2
\global\dimen3=\dimen3
\egroup % NESTED ZERO
\ifonleftofproofrule
\then   \shortenproofleft=\dimen2
\fi
\shortenproofright=\dimen3
\onleftofproofrulefalse
\ifinsideprooftree
\then   \hskip.5em plus 1fil \penalty2
\fi
}

\newcommand{\emptyPremise}{\vphantom{\vdash}}
\newcommand{\indrulename}[1]{\textsf{\textup{#1}}}
\newcommand{\indrule}[3]{
\ensuremath{
\begin{array}{c}
  \prooftree #2
    \justifies #3
    \thickness=0.05em
    \using \indrulename{#1}
  \endprooftree
\end{array}}}

\definecolor{wenv_color}{rgb}{0.0, 0.0, 1.0}
\definecolor{tenv_color}{rgb}{1.0, 0.0, 0.0}
\definecolor{rewr_color}{rgb}{0.47, 0.27, 0.23}
\definecolor{term_color}{rgb}{0.0, 0.5, 0.5}
\definecolor{darkred}{rgb}{0.8, 0.0, 0.0}

\renewcommand{\theenumi}{\arabic{enumi}}
\renewcommand{\theenumii}{\arabic{enumii}}
\renewcommand{\theenumiii}{\arabic{enumiii}}

\makeatletter
\renewcommand\p@enumii{\theenumi.}
\renewcommand\p@enumiii{\theenumi.\theenumii.}
\renewcommand\p@enumiv{\theenumi.\theenumii.\theenumiii.}
\makeatother

\newtheorem{dummythm}{dummythm}

\newtheorem{assumption}[dummythm]{Assumption}

\newcommand{\llem}[1]{\label{lemma:#1}}
\newcommand{\rlem}[1]{Lemma~\ref{lemma:#1}}
\newcommand{\ldef}[1]{\label{def:#1}}
\newcommand{\rdef}[1]{Definition~\ref{def:#1}}
\newcommand{\lprop}[1]{\label{prop:#1}}
\newcommand{\rprop}[1]{Proposition~\ref{prop:#1}}
\newcommand{\lthm}[1]{\label{thm:#1}}
\newcommand{\rthm}[1]{Theorem~\ref{thm:#1}}
\newcommand{\lremark}[1]{\label{remark:#1}}
\newcommand{\rremark}[1]{Remark~\ref{remark:#1}}

\newcommand{\lsec}[1]{\label{section:#1}}
\newcommand{\rsec}[1]{Section~\ref{section:#1}}

\newcommand{\lexample}[1]{\label{example:#1}}
\newcommand{\rexample}[1]{Example~\ref{example:#1}}

\newcommand{\HS}{\hspace{.5cm}}

\newcommand{\ST}{\ |\ }
\newcommand{\ie}{{\em i.e.}\xspace}
\newcommand{\eg}{{\em e.g.}\xspace}
\newcommand{\cf}{{\em cf.}\xspace}

\newcommand{\ih}{IH\xspace}
\newcommand{\set}[1]{\{#1\}}
\newcommand{\eqdef}{\,\mathrel{\overset{\mathrm{def}}{=}}\,}

\newcommand{\subt}[2]{\textcolor{red}{{\bm\{}}#1{\backslash}#2\textcolor{red}{{\bm\}}}}
\newcommand{\subtr}[2]{\textcolor{darkgreen}{{\bm\{}}#1{\backslash}\!\!{\backslash}#2\textcolor{darkgreen}{{\bm\}}}}
\newcommand{\subrr}[2]{\textcolor{orange}{{\bm\{}}#1{\backslash}\!\!{\backslash}\!\!{\backslash}#2\textcolor{orange}{{\bm\}}}}
\newcommand{\subm}[2]{\textcolor{blue}{{\bm\{}}#1\backslash#2\textcolor{blue}{{\bm\}}}}

\newcommand{\fv}[1]{\mathsf{fv}(#1)}
\newcommand{\dom}[1]{\mathsf{dom}(#1)}

\newcommand{\lam}[2]{\lambda#1.#2}

\newcommand{\constantset}{\mathcal{C}}
\newcommand{\ruleset}{\mathcal{R}}

\newcommand{\var}{x}
\newcommand{\vartwo}{y}
\newcommand{\varthree}{z}
\newcommand{\varfour}{w}

\newcommand{\fapp}{\mathbf{app}}
\newcommand{\flam}{\mathbf{lam}}
\newcommand{\fdup}{\mathbf{d}}
\newcommand{\fun}{\mathbf{f}}
\newcommand{\funtwo}{\mathbf{g}}
\newcommand{\cons}{\mathbf{c}}
\newcommand{\constwo}{\mathbf{d}}
\newcommand{\consthree}{\mathbf{e}}
\newcommand{\consfour}{\mathbf{f}}

\newcommand{\rulewit}{\varrho}
\newcommand{\rulewittwo}{\vartheta}

\newcommand{\btyp}{\alpha}
\newcommand{\btyptwo}{\beta}

\newcommand{\tm}{s}
\newcommand{\tmtwo}{t}
\newcommand{\tmthree}{r}
\newcommand{\tmfour}{p}
\newcommand{\tmfive}{q}

\newcommand{\redseq}{\rho}
\newcommand{\redseqtwo}{\sigma}
\newcommand{\redseqthree}{\tau}
\newcommand{\redseqfour}{\upsilon}

\newcommand{\redseqn}{\hat{\redseq}}
\newcommand{\redseqntwo}{\hat{\redseqtwo}}

\newcommand{\mstep}{\mu}
\newcommand{\msteptwo}{\nu}
\newcommand{\mstepthree}{\xi}

\newcommand{\mstepB}{\hat{\mu}}
\newcommand{\msteptwoB}{\hat{\nu}}
\newcommand{\mstepthreeB}{\hat{\xi}}
\newcommand{\mstepC}{\dot{\mu}}
\newcommand{\msteptwoC}{\dot{\nu}}
\newcommand{\mstepthreeC}{\dot{\xi}}

\newcommand{\mstepn}{\hat{\mstep}}
\newcommand{\mstepntwo}{\hat{\msteptwo}}

\newcommand{\typ}{A}
\newcommand{\typtwo}{B}
\newcommand{\typthree}{C}

\newcommand{\noenv}{\cdot}
\newcommand{\tenv}{\Gamma}
\newcommand{\tenvtwo}{\Gamma'}

\newcommand{\imp}{\to}

\newcommand{\refl}[1]{\underline{#1}}
\newcommand{\seq}{\mathrel{{\bm;}}}

\newcommand{\termeq}{=_{\beta\eta}}
\newcommand{\judgTerm}[3]{
  \textcolor{tenv_color}{#1}
  \vdash \textcolor{term_color}{#2} : #3
}
\newcommand{\judgTermEq}[4]{
  \textcolor{tenv_color}{#1}
  \vdash \textcolor{term_color}{#2} \termeq \textcolor{term_color}{#3} : #4
}
\newcommand{\rewto}{\rightarrowtriangle}
\newcommand{\rewr}[4]{#1 : #2 \rewto #3 : #4}
\newcommand{\judgRewr}[5]{
  \textcolor{tenv_color}{#1}
  \vdash
  \textcolor{rewr_color}{#2}
  :
  \textcolor{term_color}{#3}
  \rightarrowtriangle
  \textcolor{term_color}{#4}
  : #5
}
\newcommand{\judgRewrLabelled}[6]{
  \textcolor{tenv_color}{#1}
  \vdash
  \textcolor{rewr_color}{#2}
  :
  \textcolor{term_color}{#3}
  \rightarrowtriangle_{#6}
  \textcolor{term_color}{#4}
  : #5
}

\newcommand{\permeq}{\mathrel{\approx}}
\newcommand{\permeqBy}[1]{\approx^{(\footnotesize{\indrulename{#1}})}}
\newcommand{\permeqRule}[1]{\indrulename{$\permeq$-#1}}

\newcommand{\judgPermR}[8]{
  \textcolor{tenv_color}{#1}
  \vdash
  (\textcolor{rewr_color}{#2}
   : \textcolor{term_color}{#3} \rightarrowtriangle \textcolor{term_color}{#4})
  \permeq
  (\textcolor{rewr_color}{#5}
   : \textcolor{term_color}{#6} \rightarrowtriangle \textcolor{term_color}{#7})
  :
  #8
}

\newcommand{\flateq}{\mathrel{\sim}}

\newcommand{\flateqRule}[1]{\indrulename{$\flateq$-#1}}

\newcommand{\judgFlatPerm}[8]{
  \textcolor{tenv_color}{#1}
  \vdash
  (\textcolor{rewr_color}{#2}
   : \textcolor{term_color}{#3} \rightarrowtriangle \textcolor{term_color}{#4})
  \flateq
  (\textcolor{rewr_color}{#5}
   : \textcolor{term_color}{#6} \rightarrowtriangle \textcolor{term_color}{#7})
  :
  #8
}

\newcommand{\etalong}{\overline{\eta}}
\newcommand{\toetaexp}{\mathrel{\,\to_{\etalong}\,}}

\newcommand{\ctxhole}{\Box}
\newcommand{\ctxof}[1]{\langle#1\rangle}

\newcommand{\cctx}{\mathtt{C}}
\newcommand{\cctxof}[1]{\cctx\ctxof{#1}}

\newcommand{\sctx}{\mathtt{S}}
\newcommand{\sctxof}[1]{\sctx\ctxof{#1}}

\newcommand{\kctx}{\mathtt{K}}
\newcommand{\kctxB}{\tilde{\mathtt{K}}}
\newcommand{\kctxof}[1]{\kctx\ctxof{#1}}

\newcommand{\flatteningSystem}{\ensuremath{\mathcal{F}}}
\newcommand{\flatsymbol}{\flatteningSystem}
\newcommand{\tof}{\mathrel{\overset{\flatsymbol}{\mapsto}}}

\newcommand{\tofs}{\mathrel{\overset{\flatsymbol}{\mapsto}{\hspace{-.2em}}^{*}}}
\newcommand{\tofsSub}[1]{\mathrel{\overset{\flatsymbol}{\mapsto}{\hspace{-.2em}}^{*}_{#1}}}

\newcommand{\tofeqSub}[1]{\mathrel{\overset{\flatsymbol}{\mapsto}{\hspace{-.2em}}^{=}_{#1}}}

\newcommand{\tofsinv}{\mathrel{\overset{\flatsymbol}{\mapsfrom}{\hspace{-.2em}}^{*}}}
\newcommand{\tofnoeta}{\mathrel{\overset{\circ}{\mapsto}}}
\newcommand{\flatRule}[1]{\indrulename{$\mathcal{F}$-#1}}
\newcommand{\flatRuleAnon}{\mathtt{x}}
\newcommand{\flatRuleAnonTwo}{\mathtt{y}}
\newcommand{\flatten}[1]{#1^{\flatsymbol}}
\newcommand{\flattennoeta}[1]{#1^{\circ}}

\newcommand{\rsrc}[1]{#1^{\mathsf{src}}}
\newcommand{\rtgt}[1]{#1^{\mathsf{tgt}}}

\newcommand{\rheavy}[1]{\#_{\mathsf{h}}(#1)}
\newcommand{\rweight}[1]{\#_{\mathsf{w}}(#1)}

\newcommand{\altsrc}[1]{#1^{\mathsf{<\etalong}}}

\newcommand{\fsrc}[1]{#1^{\blacktriangleleft}}
\newcommand{\ftgt}[1]{#1^{\blacktriangleright}}
\newcommand{\fsrcb}[1]{\flatten{(\rsrc{#1})}}
\newcommand{\ftgtb}[1]{\flatten{(\rtgt{#1})}}

\newcommand{\judgSplit}[3]{#1 \Leftrightarrow #2 \seq #3}
\newcommand{\judgSplitF}[3]{#1 \overset{\flatsymbol}{\Leftrightarrow} #2 \seq #3}
\newcommand{\judgSplitFTable}[3]{
  \begin{array}{rrl}
    #1
    & \overset{\flatsymbol}{\Leftrightarrow}
    & #2
  \\
    & \seq
    & #3
  \end{array}
}

\newcommand{\judgProj}[3]{#1/\!\!/\!\!/#2\Rightarrow#3}
\newcommand{\judgCompat}[2]{#1 \uparrow #2}

\newcommand{\projmr}{\mathrel{/^1}}
\newcommand{\projrm}{\mathrel{/^2}}
\newcommand{\projrr}{\mathrel{/^3}}

\newcommand{\sz}[1]{|#1|}

\newcommand{\streq}{\mathrel{\asymp}}

\newcommand{\tostd}{\mathrel{\rhd}}
\newcommand{\tostdinv}{\mathrel{\lhd}}
\newcommand{\tostdBy}[1]{\mathrel{\rhd}^{(\footnotesize{\indrulename{#1}})}}
\newcommand{\stdrule}[1]{\indrulename{\ensuremath{\rhd}-{#1}}}
\newcommand{\len}[1]{\ell(#1)}

\newcommand{\join}{\sqcup}

\newcommand{\depth}[1]{\mathbf{D}(#1)}
\newcommand{\SR}[1]{\mathbf{SR}(#1)}
\newcommand{\srgt}{\succ}
\newcommand{\stdform}[1]{\mathsf{std}(\redseq)}

\newcommand{\resultName}[1]{\textsf{\textup{#1}}}

\newcommand{\hrsone}{\mathcal{R}}
\newcommand{\tohrs}[1]{\mathrel{\,\to_{#1}\,}}
\newcommand{\consof}[1]{\mathbf{#1}}

\newcommand{\varseq}[1]{\vec{#1}}

\newcommand{\NotNow}[1]{}
\newcommand{\proja}{{/\!\!/}}

\theoremstyle{plain}
\def\eg{{\em e.g.}}
\def\cf{{\em cf.}}

\begin{document}

\title[Reductions in Higher-Order Rewriting and Their Equivalence]{Reductions in Higher-Order Rewriting and Their Equivalence}
%\titlecomment{{\lsuper*}OPTIONAL comment concerning the title, \eg,
%  if a variant or an extended abstract of the paper has appeared elsewhere.}
\thanks{Partially supported by project ECOS-Sud A17C01}	%optional

% affiliations are numbered automatically with a, b, c (see below)
% use the optional argument to indicate the affiliation(s) of each author
% omit the argument if there is only one author, or only one affiliation
\author[P.~Barenbaum]{Pablo Barenbaum}[a,b]
\author[E.~Bonelli]{Eduardo Bonelli}[c]
%\author[J.~Name3]{Josiah S.~Carberry\lmcsorcid{0000-0002-1825-0097}}[a]

% affiliation 1 (automatically numbered a)
\address{Universidad Nacional de Quilmes (CONICET), Argentina} %optional
\email{pbarenbaum@dc.uba.ar} 
\address{Universidad de Buenos Aires, Argentina} %optional
%\url{http://foones.github.io}	
% write emails for all authors having that affiliation
%optional

% affiliation 2 (automatically numbered b)
\address{Stevens Institute of Technology, USA}	%optional
\email{ebonelli@stevens.edu}  %optional
% \url{https://ebonelli.github.io/}
%% etc.

%% required for running head on odd and even pages, use suitable
%% abbreviations in case of long titles and many authors:

%%%%%%%%%%%%%%%%%%%%%%%%%%%%%%%%%%%%%%%%%%%%%%%%%%%%%%%%%%%%%%%%%%%%%%%%%%%

% \ccsdesc[500]{Theory of computation~Equational logic and rewriting}
% \ccsdesc[500]{Theory of computation~Type theory}

%  \relatedversiondetails[cite=techReport]{Extended Version}{https://arxiv.org/abs/2210.15654}

\begin{abstract}
  Proof terms are syntactic expressions that represent computations in
  term rewriting.  They were introduced by Meseguer and exploited by van Oostrom
  and de Vrijer
  to study {\em equivalence of reductions} in (left-linear)
  first-order term rewriting systems.  We study the problem of
  extending the notion of proof term to {\em
    higher-order rewriting}, which generalizes the first-order setting
  by allowing terms with binders and higher-order substitution.  In
  previous works that devise proof terms for higher-order rewriting,
  such as Bruggink's, it has been noted that the challenge lies in
  reconciling composition of proof terms and higher-order substitution
  ($\beta$-equivalence).  This led Bruggink to reject ``nested''
  composition, other than at the outermost level.  In this paper, we
  propose a notion of higher-order proof term we dub \emph{rewrites} that supports nested
  composition. We then define {\em two} notions of equivalence on rewrites, namely  
  {\em permutation equivalence} and {\em projection equivalence}, and show that they coincide.
  We also propose a {\em standardization} procedure, that computes a canonical
  representative of the permutation equivalence class of a rewrite.
\end{abstract}

\maketitle

\section{Introduction}

Term rewriting systems model computation as sequences of steps between terms,
{\em reduction sequences}, where steps are instances of term rewriting
rules~\cite{Terese03}. It is natural to consider reduction sequences up to
swapping of orthogonal steps since such reductions perform the ``same work''.
The ensuing notion of equivalence is called {\em permutation equivalence} and
was first studied by L\'evy~\cite{Levy1978} in the setting of the
$\lambda$-calculus but has appeared in other guises connected with
concurrency~\cite[Rem.8.1.1]{Terese03}. 
As an example, consider the rewrite rule
$\cons(\var,\consfour(\vartwo)) \rewto \constwo(\var,\var)$
 and the following reduction sequence 
where, in each step, the contracted redex is underlined:
  \begin{equation}
  \underline{\cons(\cons(\varthree,\consfour(\varthree)),\consfour(\varthree))}
  \rewto
  \constwo(\underline{\cons(\varthree,\consfour(\varthree))},
           \cons(\varthree,\consfour(\varthree)))
  \rewto
  \constwo(\constwo(\varthree,\varthree),
           \underline{\cons(\varthree,\consfour(\varthree))})
  \rewto
  \constwo(\constwo(\varthree,\varthree),\constwo(\varthree,\varthree))
  \label{eq:red}
  \end{equation}
Performing the innermost redex first, rather than the outermost one, leads to:
  \begin{equation}
  \cons(\underline{\cons(\varthree,\consfour(\varthree))},\consfour(\varthree))
  \rewto
  \underline{\cons(\constwo(\varthree,\varthree),
           \consfour(\varthree))}
  \rewto
  \constwo(\constwo(\varthree,\varthree),\constwo(\varthree,\varthree))
  \label{eq:red:ii}
  \end{equation}
  The first step in~(\ref{eq:red}) makes two copies of the innermost redex. It is the two steps contracting these copies that are swapped with the first one in~(\ref{eq:red}) to produce~(\ref{eq:red:ii}).
  Consider also the following reduction sequence:
  \begin{equation}
  \underline{\cons(\varthree,\consfour(\cons(\varthree,\consfour(\varthree))))}
  \rewto
  \constwo(\varthree,\varthree)
  \label{eq:red:iii}
  \end{equation}
Performing the innermost redex first, rather than the outermost one, leads to:
  \begin{equation}
  \cons(\varthree,\consfour(\underline{\cons(\varthree,\consfour(\varthree))}))
  \rewto
  \underline{\cons(\varthree,
           \consfour(\constwo(\varthree,\varthree)))}
  \rewto
  \constwo(\varthree,\varthree)
  \label{eq:red:iv}
  \end{equation}
  The step in~(\ref{eq:red:iii}) erases the innermost redex. The two steps of~(\ref{eq:red:iv}) can be swapped to produce~(\ref{eq:red:iii}).
  Such duplication and erasure contribute most of the complications behind permutation equivalence, both in its formulation and the study of its properties.

\paragraph*{Proof Terms}
One natural way to represent computations is by means of so-called \emph{proof terms}.
They were introduced by Meseguer as a means of representing proofs in Rewriting
Logic~\cite{DBLP:journals/tcs/Meseguer92} and exploited by van Oostrom and
de Vrijer in the setting of first-order left-linear rewriting systems,
to study equivalence of reductions
in~\cite{DBLP:journals/entcs/OostromV02} and~\cite[Chapter 9]{Terese03}.
Each rewrite rules is assigned a {\em rule symbol} denoting the application of that rule. Proof terms are expressions built using
function symbols, rule symbols, and a binary operator ``$\seq$''
denoting sequential composition of proof terms.  
Assuming
that we choose the symbol ``$\rulewit$'' for our rewrite rule,
in such a way that $\rulewit(\var,\vartwo)$ witnesses the step
$\cons(\var,\consfour(\vartwo)) \rewto \constwo(\var,\var)$,
% Furthermore, there is 
% in such a way that
% if $\redseq_1$ denotes a computation $\tm_0 \rewto \tm_1$
% and $\redseq_2$ denotes a computation $\tm_1 \rewto \tm_2$,
% then $(\redseq_1\seq\redseq_2)$ denotes their composition,
% which is a computation $\tm_0 \rewto \tm_2$.
reduction~(\ref{eq:red}) may be represented
as the proof term:
$
       \rulewit(\cons(\varthree,\consfour(\varthree)),\varthree)
  \seq \constwo(\rulewit(\varthree,\varthree),\cons(\varthree,\consfour(\varthree)))
  \seq \constwo(\constwo(\varthree,\varthree),\rulewit(\varthree,\varthree))
$
and reduction~(\ref{eq:red:ii}) as the proof term:
$
\cons(\rulewit(\varthree,\varthree),\consfour(\varthree))
\seq\rulewit(\constwo(\varthree,\varthree),\varthree)
$.
One notable feature of proof terms is that they support parallel steps. For instance, both proof terms above are permutation equivalent
to
$\rulewit(\cons(\varthree,\consfour(\varthree)),\varthree)
  \seq \constwo(\rulewit(\varthree,\varthree),\rulewit(\varthree,\varthree))$,
which performs the two last steps in parallel,
as well as to
$\rulewit(\rulewit(\varthree,\varthree),\varthree)$,
which performs all steps simultaneously.
Using this representation, permutation equivalence
can be studied in terms of equational theories on proof terms.

\paragraph*{Equivalence of Reductions via Proof Terms for First-Order Rewriting.}
In~\cite{DBLP:journals/entcs/OostromV02},
van Oostrom and de Vrijer characterize permutation equivalence
of proof terms in four alternative ways.
First, they formulate an equational theory
of permutation equivalence $\redseq \permeq \redseqtwo$
between proof terms, such that for example
$
\rulewit(\cons(\varthree,\consfour(\varthree)),\varthree)
\seq\constwo(\rulewit(\varthree,\varthree),\rulewit(\varthree,\varthree))
\permeq \rulewit(\rulewit(\varthree,\varthree),\varthree)$ holds.
These equations account for the behavior of proof term composition,
which has a monoidal structure, in the sense that composition is associative
and {\em empty} steps act as identities. 
Second, they define an operation of {\em projection} $\redseq/\redseqtwo$,
denoting the computational work that is left of $\redseq$ after
$\redseqtwo$.
For example,
$\cons(\rulewit(\varthree,\varthree),\consfour(\varthree))
/\rulewit(\cons(\varthree,\consfour(\varthree)),\varthree)
= \constwo(\rulewit(\varthree,\varthree),\rulewit(\varthree,\varthree))$.
This induces a notion of {\em projection equivalence}
between proof terms $\redseq$ and $\redseqtwo$,
declared to hold when both $\redseq/\redseqtwo$ and $\redseqtwo/\redseq$
are empty, \ie they contain no rule symbols.
Third, they define a {\em standardization procedure} to reorder the
steps of a reduction in outside-in order, mapping each proof term $\redseq$
to a proof term $\redseq^*$ in {\em standard form}.
For example, the (parallel) standard form of 
$\cons(\rulewit(\varthree,\varthree),\consfour(\varthree))
\seq\rulewit(\constwo(\varthree,\varthree),\varthree)$
is
$\rulewit(\cons(\varthree,\consfour(\varthree)),\varthree)
  \seq \constwo(\rulewit(\varthree,\varthree),\rulewit(\varthree,\varthree))$.
This induces a notion of {\em standardization equivalence}
between proof terms $\redseq$ and $\redseqtwo$,
declared to hold when $\redseq^* = \redseqtwo^*$.
Fourth, they define a notion of {\em labelling equivalence},
based on lifting computational steps to labelled terms.
Although these notions of equivalence were known prior to~\cite{DBLP:journals/entcs/OostromV02},
the main 
contribution of that paper is to show that they can be systematically studied using proof terms
 and, moreover, shown to coincide.

\paragraph*{Higher-Order Rewriting.}
Higher-order term rewriting (HOR) generalizes first-order term rewriting
by allowing binders.
Function symbols
are generalized to constants of any given simple type,
and first-order terms
are generalized to simply-typed $\lambda$-terms,
including constants and up to $\beta\eta$-equivalence.
The paradigmatic example of a higher-order rewriting system is the
$\lambda$-calculus.
It includes a base type $\iota$
and two constants $\fapp : \iota \to \iota \to \iota$
and $\flam : (\iota \to \iota) \to \iota$; $\beta$-reduction
may be expressed as
the higher-order rewrite rule
$
  \fapp\,(\flam\,(\lam{\varthree}{\var\, \varthree}))\,\vartwo
          \rewto
          \var\,\vartwo
$. 
A sample reduction sequence is:
{\small
\begin{equation}
  \mathbf{lam}(\lambda
  v.\mathbf{app}(\mathbf{lam}(\lambda
  x. x), \underline{\mathbf{app}(\mathbf{lam}(\lambda
  w. w),v)}))
  \rewto
  \mathbf{lam}(\lambda
  v. \underline{\mathbf{app}(\mathbf{lam}(\lambda x. x),v)})
  \rewto
  \mathbf{lam}(\lambda v.v)
  \label{eq:red:ho}
\end{equation}}%
Generalizing proof terms to the setting of higher-order rewriting is a natural
goal. Just like in the first-order case, we assign rule symbols to rewrite rules. One would then expect to obtain proof terms by adding these rule symbols and the ``$\seq$'' composition operator to the simply typed $\lambda$-calculus.  If we assume the following rule symbol for our rewrite rule $
  \rulewit\, \var\,\vartwo: \fapp\,(\flam\,(\lam{\varthree}{\var\, \varthree}))\,\vartwo
          \rewto
          \var\,\vartwo
$, then an example of a higher-order proof term for (\ref{eq:red:ho}) is:
\[
  {\small
    \mathbf{lam}\bigg(\lambda v.\big(\mathbf{app}(\mathbf{lam}(\lambda x.x), \rulewit\, (\lambda w.w)\,v)\ ;\ \rulewit\, (\lambda u.u)\,v\big)\bigg)    
  }
\]

However, higher-order substitution and proof term composition seem not to be in
consonance, an issue already observed by Bruggink~\cite{thesis:bruggink:08}.
Consider a variable $\var$. This variable itself
denotes an empty computation $\var \rewto \var$,
so the composition $(\var\seq\var)$
also denotes an empty computation $\var \rewto \var$.
If $\redseqtwo$ is an arbitrary proof term representing a computation $\tm \rewto \tmtwo$,
the proof term $(\lam{\var}{(\var\seq\var)})\,\redseqtwo$
should, in principle, represent a computation
$(\lam{\var}{\var})\,\tm \rewto (\lam{\var}{\var})\,\tmtwo$.
This is the same as $\tm \rewto \tmtwo$,
because terms are regarded up to $\beta\eta$-equivalence.
The challenge lies in lifting $\beta\eta$-equivalence to the level of
proof terms:
if $\beta$-reduction is naively extended to operate on proof terms,
the well-formed proof term $(\lam{\var}{(\var\seq\var)})\,\redseqtwo$
becomes equal to $(\redseqtwo\seq\redseqtwo)$,
which is ill-formed because $\redseqtwo$ is not composable with itself
if $\tm\not=_{\beta\eta}\tmtwo$.
Rather than simply disallowing the use of ``$\seq$'' under applications and
abstractions (the route taken in~\cite{thesis:bruggink:08}), our aim is to integrate it with
$\beta\eta$-reduction.

\paragraph*{Contribution.}
We propose a syntax for higher-order proof terms,
  called \emph{rewrites}, that includes  $\beta\eta$-equivalence
  and allows them to be freely composed.
  We then define a relation $\redseq \permeq \redseqtwo$
  of \emph{permutation equivalence} between rewrites, the central notion of our work.
  The issue mentioned above is avoided by {\em disallowing}
  the ill-behaved substitution
  of a rewrite in a rewrite ``$\redseq\{\var\backslash\redseqtwo\}$'',
  and by only allowing notions of substitution
  of a term in a rewrite $\redseq\subt{\var}{\tm}$,
  and of a rewrite in a term $\tm\subtr{\var}{\redseq}$.
  From these, a well-behaved
  notion of substitution of a rewrite in a rewrite
  $\redseq\subrr{\var}{\redseqtwo}$
  can be shown to be \emph{derivable}.
  We also define a notion of \emph{projection} $\redseq\proja\redseqtwo$.
  The induced notion of
  \emph{projection equivalence coincides with permutation equivalence},
  in the sense that
  $\redseq \permeq \redseqtwo$ iff
   $\redseq\proja\redseqtwo \permeq \rtgt{\redseqtwo}$
    and 
    $\redseqtwo\proja\redseq \permeq \rtgt{\redseq}$,
    where $\rtgt{\redseq}$ stands for the {\em target} term of $\redseq$.
  The equivalence is established by means of \emph{flattening},
  a method to convert an arbitrary rewrite $\redseq$
  into a ({\em flat}) representative $\flatten{\redseq}$
  that only uses the composition operator ``$\seq$'' at the top level, and
  a notion of \emph{flat permutation equivalence} $\redseq \flateq \redseqtwo$.
  Flattening is achieved by means of a rewriting system
  whose objects are themselves rewrites.
  This system is shown to be confluent and strongly normalizing. We also show   that
  \emph{permutation equivalence is sound and complete
  with respect to flat permutation equivalence}
  in the sense that $\redseq \permeq \redseqtwo$
  if and only if $\flatten{\redseq} \flateq \flatten{\redseqtwo}$.
  Finally, we present a notion of parallel standardization
  that can be used to effectively
  find a canonical representative $\redseq^*$ of the permutation equivalence class
  of a given rewrite $\redseq$, under appropriate conditions.
  
\paragraph*{Structure of the Paper.}
  In~\rsec{hor}
  we review Nipkow's Higher-Order Rewriting Systems.
  \rsec{rewrites} proposes our notion of rewrite and \rsec{permutation_equivalence}
  introduces permutation equivalence for them.
  Flattening is presented in~\rsec{flat_permutation_equivalence}.
  In this section, we also formulate an equational theory
  defining the relation $\redseq \flateq \redseqtwo$
  of flat permutation equivalence between flat rewrites.
  It relies crucially on a ternary relation
  between {\em multisteps}, called {\em splitting}
  and written $\judgSplit{\mstep}{\mstep_1}{\mstep_2}$,
  meaning that $\mstep$ and $\mstep_1\seq\mstep_2$
  perform the same computational work.
  In~\rsec{projection}
  we first define a projection operator for flat rewrites $\redseq/\redseqtwo$,
  and we lift it to a projection operator for arbitrary rewrites
  $\redseq\proja\redseqtwo \eqdef \flatten{\redseq}/\flatten{\redseqtwo}$.
  Then we show that
  the induced notion of projection equivalence coincides with permutation equivalence.
  \rsec{standardization} presents some preliminary results on parallel standardization
  of rewrites.
  Finally, we conclude and discuss related and future work.
  % \pablo{[[Entiendo que lo siguiente no va más]]:}
  % \pabloso{Some details have been relegated to an accompanying appendix.}

\paragraph*{Note.} An extended abstract of this paper was published as~\cite{DBLP:conf/csl/BarenbaumB23}.
This report includes full proofs and a new section on Standardization (\rsec{standardization}).

%%% Local Variables:
%%% mode: latex
%%% TeX-master: "main"
%%% End:

\section{Higher-Order Rewriting}
\lsec{hor}

\newcommand{\etanf}[1]{\ensuremath{{#1}\downarrow^{\eta}}}
\newcommand{\etaexp}[1]{\ensuremath{{#1}\uparrow^{\eta}}}
\newcommand{\betanf}[1]{\ensuremath{{#1}\downarrow^{\beta}}}
\newcommand{\betaetanf}[1]{\ensuremath{{#1}\updownarrow^{\eta}_{\beta}}}
\newcommand{\seqt}[1]{\overline{#1}}
\newcommand{\lamseq}[2]{\ensuremath{\lambda \seqt{#1}.{#2}}}
\newcommand{\subst}{\theta}

There are various approaches to HOR in the literature, 
including Klop's Combinatory Reduction Systems (CRSs)~\cite{Klop1980}
and Nipkow's Higher-Order Rewriting Systems
(HRSs)~\cite{nipkow1991higher,MayrNipkow98}. We consider HRSs in this paper.
Their use of the simply-typed lambda calculus for representing terms and
substitution provides a suitable starting point for modeling our rewrites.

Assume given a denumerably infinite set of
{\em variables} ($\var,\vartwo,\hdots$),
{\em base types} ($\btyp,\btyptwo,\hdots$), and
{\em constant symbols} ($\cons,\constwo,\hdots$).
The sets of {\em terms} ($\tm,\tmtwo,\hdots$) and
{\em types} ($\typ,\typtwo,\hdots$) 
are given by:
\[
  \begin{array}{lrl@{\hspace{1cm}}lrl}
    \tm & ::=  & \var \,\mid\,\cons\,\mid\,\lam{\var}{\tm}\,\mid\,\tm\,\tm
  &
    \typ & ::=  & \btyp \,\mid\,\typ \imp \typ
  \\
  \end{array}
\]
A term can either be a variable, a constant, an abstraction or an
application. A type can either be a base type or an arrow type. We write $\fv{\tm}$ for the free variables of $\tm$.
We use $\seqt{X_n}$, or sometimes just $\seqt{X}$ if $n$ is clear
from the context,
to denote a sequence $X_1,\ldots,X_n$.
Following standard conventions, $\tm\,\seqt{\tmtwo_n}$ stands for the iterated application
$\tm\,\tmtwo_1\hdots\tmtwo_n$,
and $\seqt{\typ_n}\to\typtwo$ for the type $\typ_1\to\hdots\typ_n\to\typtwo$.
We write $\tm\subt{\var}{\tmtwo}$ for the capture-avoiding
substitution of all free occurrences of $\var$ in $\tm$ with $\tmtwo$ and call it a \emph{term/term substitution}.
As usual, terms are defined modulo $\alpha$-conversion, \ie we identify terms that differ only in the names of their bound variables.   A {\em typing context} ($\tenv,\tenvtwo,\hdots$)  is a partial function from variables to types. We write $\dom{\tenv}$ for the {\em domain} of $\tenv$. Given a typing context $\tenv$ and $\var\notin\dom{\tenv}$, we write $\tenv,\var:\typ$ for the typing context such that $(\tenv,\var:\typ)(\var)= \typ$, and $(\tenv,\var:\typ)(\vartwo)=\tenv(\vartwo)$ whenever $\vartwo\neq \var$. We write a centered dot 
$\noenv$ for the empty typing context and
$\var \in \tenv$ if $\var\in\dom{\tenv}$.
A {\em signature} of a HRS is a set $\constantset$ of typed constants $\cons:\typ$. A sample signature is $\constantset=\{\fapp : \iota \to \iota \to \iota,
\flam : (\iota \to \iota) \to \iota\}$ for $\iota$ a base type.

\begin{defi}[Type system for terms]
\ldef{typing_terms}
Terms are typed using the usual typing rules of the simply-typed $\lambda$-calculus:
\[
  \indrule{Var}{
    (\var:\typ) \in \tenv
  }{
    \judgTerm{\tenv}{\var}{\typ}
  }
  \indrule{Con}{
    (\cons:\typ) \in \constantset
  }{
    \judgTerm{\tenv}{\cons}{\typ}
  }
  \indrule{Abs}{
    \judgTerm{\tenv, \var:\typ}{\tm}{\typtwo}
  }{
    \judgTerm{\tenv}{\lam{\var}{\tm}}{\typ \imp \typtwo}
  }
  \indrule{App}{
    \judgTerm{\tenv}{\tm}{\typ \imp \typtwo}
    \HS
    \judgTerm{\tenv}{\tmtwo}{\typ}
  }{
    \judgTerm{\tenv}{\tm\,\tmtwo}{\typtwo}
  }
\]
Given any $\tenv$ and $\typ$ such that $\judgTerm{\tenv}{\tm}{\typ}$ can be proved
using these rules, we say $\tm$ is a \emph{typed term} over $\constantset$.
We typically assume $\constantset$ to be globally fixed, leaving it as an implicit parameter of the type system.
\end{defi}

The notions
of $\beta$ and $\eta$-reduction between terms are defined as usual.
Recall that $\beta$-reduction and $\eta$-reduction are confluent and terminating on typed terms.
We write $\betanf{\tm}$ (resp. $\etanf{\tm}$) for the unique $\beta$-normal form (resp. $\eta$-normal form) of $\tm$. The $\beta$-normal form of a term $\tm$ has the form $\lamseq{x_k}{a\,\tmtwo_1\ldots\tmtwo_m}$, for $a$ either a constant or a variable. The $\eta$-expanded form of $\tm$ is defined as:
\begin{center}
  $\begin{array}{rcl}
   \etaexp{\tm} & \eqdef & \lamseq{x_{n+k}}{a\,(\etaexp{\seqt{\tmtwo_m}})\,(\etaexp{\var_{n+1}})\ldots(\etaexp{\var_{n+k}})}
   \end{array}$
   \end{center}
where $\tm$ is assumed to have type $\seqt{\typ_{n+k}}\imp \typtwo$ and the $x_{n+1},\ldots,x_{n+k}$ are fresh. We use $\betaetanf{\tm}$ to denote the term $\etaexp{\betanf{\tm}}$ and call it the $\beta\etalong$-normal form of $\tm$.

A {\em substitution} $\subst$ is a function mapping variables to typed terms
such that $\subst(\var) \neq \var$ only for finitely many $\var$.
The {\em domain} of a substitution is defined as
$\dom{\subst}= \set{\var \ST \subst(\var) \neq \var}$.
The application of a substitution $\subst = \{\var_1\mapsto\tm_1,\ldots,\var_n\mapsto\tm_n\}$
to a term $\tmtwo$ is defined as $\subst\,\tmtwo \eqdef \betaetanf{((\lamseq{\var_n}{\tmtwo})\,\seqt{\tm_n})}$.

\begin{defi}
A {\em pattern} is a typed term  in $\beta$-normal form
such that all free occurrences of a variable $\var_i$ are in a subterm of the 
  form $\var_i\,\tmtwo_1\hdots\tmtwo_k$
  with $\tmtwo_1,\hdots,\tmtwo_k$ $\eta$-equivalent to 
  distinct bound variables.
A \emph{rewriting rule} is a pair
$\langle \ell,r\rangle$
of typed terms in $\beta\etalong$-normal form of the same base type with $\ell$
a pattern not $\eta$-equivalent to a variable and $\fv{r}\subseteq\fv{\ell}$. An HRS is a pair consisting of a signature and a set of rewriting rules over that signature. We typically omit the signature.
\end{defi}

\begin{exa}
\lexample{hrs:mu_plus_f_rules}
Consider a base type $\iota$ and typed constants
$\consof{mu}:(\iota\imp\iota)\imp\iota$ and $\consof{f}:\iota\imp\iota$. Two sample rewriting rules are:
$
      \langle \consof{mu}(\lam{\vartwo}{\var\,\vartwo}),
      \var\,(\consof{mu}(\lam{\vartwo}{\var\,\vartwo}))
      \rangle
$
and
$
\langle
\consof{f}\,\var,
\consof{g}\,\var
\rangle
      $. All four terms have base type $\iota$.   
    \end{exa}
    
\begin{defi}
  The rewrite relation $\tohrs{\hrsone}$ for an HRS $\hrsone$ is the relation over typed terms in $\beta\etalong$-normal form defined as follows:
  \[
  \indrule{Root}{
    \langle \ell,r\rangle \in\hrsone
  }{
    \subst\,\ell \tohrs{\hrsone} \subst\,r
  }
  \indrule{App}{
        \tm \tohrs{\hrsone} \tmtwo
  }{
    a\,\seqt{\tmthree_m}\,\tm\,\seqt{\tmfour_n} \tohrs{\hrsone} a\,\seqt{\tmthree_m}\,\tmtwo\,\seqt{\tmfour_n}
  }
  \indrule{Abs}{
    \tm \tohrs{\hrsone} \tmtwo
  }{
    \lam{\var}{\tm} \tohrs{\hrsone} \lam{\var}{\tmtwo}
  }
\]
where $a$ is either a constant or a variable of type $\seqt{\typ_{m+1+n}}\imp \typtwo$.
We write $\stackrel{*}{\rightarrow_{\hrsone}}$ (resp. $ \stackrel{*}{\leftrightarrow_{\hrsone}}$) for the reflexive, transitive  (resp.  reflexive, symmetric and transitive) closure of $\tohrs{\hrsone}$.
\end{defi}

\begin{exa}
\lexample{hrs:mu_plus_f}
An example of a sequence of rewrite steps in the HRS of \rexample{hrs:mu_plus_f_rules} is
      $\consof{mu}\,(\lam{\var}{\consof{f}\,\var})
      \tohrs{\hrsone}
      \consof{f}\,(\consof{mu}\,(\lam{\var}{\consof{f}\,\var}))
      \tohrs{\hrsone}
      \consof{f}\,(\consof{mu}\,(\lam{\var}{\consof{g}\,\var}))
      \tohrs{\hrsone}
      \consof{g}\,(\consof{mu}\,(\lam{\var}{\consof{g}\,\var}))$.      
\end{exa}

Recall that an HRS is {\em orthogonal} if:
\begin{enumerate}
\item %\inlineItem{1.}
  The rules are {\em left-linear}, \ie if the left-hand side
  $\ell$ has $\fv{\ell}=\{\var_1,\ldots,\var_n\}$,
  then 
  there is {\em exactly} one free occurrence of $\var_i$ in $\ell$,  for each $1 \leq i \leq n$.
\item %\inlineItem{2.}
  There are {\em no critical pairs}, as defined for example in~\cite[Def.~4.1]{nipkow1991higher}.
  \end{enumerate}
Orthogonal HRSs are deterministic in the sense that their rewrite relation is
confluent~\cite{nipkow1991higher}.
All of the examples of HRSs presented above are orthogonal. In the
sequel of this paper, we assume given a fixed, orthogonal HRS $\hrsone$.

%%% Local Variables:
%%% mode: latex
%%% TeX-master: "main"
%%% End:

\section{Rewrites}
\lsec{rewrites}
\newcommand{\aneqn}[2]{{#1}\doteq {#2}}
\newcommand{\eqt}[3]{{#1}\doteq_{#3} {#2}}
\newcommand{\judgRewrE}[6]{
  \textcolor{tenv_color}{#1}
  \vdash
  %\textcolor{rewr_color}{#2}
  %:
  \textcolor{term_color}{#3}
  % \rightarrowtriangle
  \doteq_{#6}
  \textcolor{term_color}{#4}
  : #5
}

\newcommand{\judgRewrL}[5]{
  \textcolor{tenv_color}{#1}
  \vdash
  \textcolor{term_color}{#3}
  \rightarrowtriangle
  \textcolor{term_color}{#4}
  : #5
}

\newcommand{\closed}[1]{{#1}^c}

In this section we propose a syntax for higher-order proof terms,
called \emph{rewrites}\footnote{Our notion of rewrite is unrelated to that of Definition~2.4 in~\cite{DBLP:journals/tcs/Meseguer92}; it corresponds to ``proof terms'' as introduced in Section~3.1  in~\cite{DBLP:journals/tcs/Meseguer92}.}. Rewrites for an HRS $\hrsone$ are a means for denoting proofs in Higher-Order Rewriting Logic (HORL, \cf~\rdef{body:typing_terms_and_rewrites}) which, in turn, correspond to reduction sequences in $\hrsone$ (\cf~Theorem.~\ref{thm:rewrites:and:HORL}). As in the first-order case~\cite{DBLP:journals/tcs/Meseguer92}, HORL is simply the equational theory that results from an HRS but disregarding symmetry.
Given an HRS $\hrsone$, let $\closed{\ruleset}$ denote the set of pairs $\langle\lam{\seqt{\var_n}}{\ell}, \lam{\seqt{\var_n}}{r}\rangle$ such that $\langle \ell,r \rangle\in\ruleset$ and $\{\var_1,\ldots,\var_n\}=\fv{\ell}$. We begin by recalling the definition of equational logic (\cf~\rdef{body:horeql}), the equational theory induced by an HRS. It is essentially that of~\cite[Def.~3.11]{MayrNipkow98}, except that in the inference rule $\indrulename{ERule}$ we use $\closed{\ruleset}$ rather than $\ruleset$. This equivalent formulation will be convenient when introducing rewrites since free variables in the LHS of a rewrite rule will be reflected in the rewrite too.

\begin{defi}[Equational Logic]
\ldef{body:horeql}
An HRS $\hrsone$ induces a congruence
$\eqt{}{}{\hrsone}$ on terms defined by the following rules:
{\small
\[
  \begin{array}{c}
  \indrule{EBeta}{
    \judgTerm{\tenv,\var:\typ}{\tm}{\typtwo}
    \HS
    \judgTerm{\tenv}{\tmtwo}{\typ}
  }{
    \judgRewrE{\tenv}{}{(\lam{\var}{\tm})\,\tmtwo}{\tm\subt{\var}{\tmtwo}}{\typtwo}{\hrsone}
  }
  \indrule{EEta}{
    \judgTerm{\tenv,\var:\typ}{\tm}{\typtwo}
    \HS
    \var\notin\fv{\tm}
  }{
    \judgRewrE{\tenv}{}{\lam{\var}{\tm\,\var}}{\tm}{\typtwo}{\hrsone}
  }
  \\
  \\
  \indrule{EVar}{
    (\var:\typ) \in \tenv
  }{
    \judgRewrE{\tenv}{\var}{\var}{\var}{\typ}{\hrsone}
  }
  \indrule{ECon}{
    (\cons:\typ) \in \constantset
  }{
    \judgRewrE{\tenv}{\cons}{\cons}{\cons}{\typ}{\hrsone}
  }
  \\
  \\
  \indrule{ESymm}{
    \judgRewrE{\tenv}{
      \redseq
    }{
      \tm_0
    }{
      \tm_1
    }{
      \typ
    }{\hrsone}
  }{
    \judgRewrE{\tenv}{
      \redseq\seq\redseqtwo
    }{
      \tm_1
    }{
      \tm_0
    }{
      \typ
    }{\hrsone}
     }
  \HS
  \indrule{ETrans}{
    \judgRewrE{\tenv}{
      \redseq
    }{
      \tm_0
    }{
      \tm_1
    }{
      \typ
    }{\hrsone}
    \HS
    \judgRewrE{\tenv}{
      \redseqtwo
    }{
      \tm_1
    }{
      \tm_2
    }{
      \typ
    }{\hrsone}
  }{
    \judgRewrE{\tenv}{
      \redseq\seq\redseqtwo
    }{
      \tm_0
    }{
      \tm_2
    }{
      \typ
    }{\hrsone}
  }
  \\
  \\
  \indrule{EAbs}{
    \judgRewrE{\tenv,\var:\typ}{
      \redseq
    }{
      \tm_0
    }{
      \tm_1
    }{\typtwo}{\hrsone}
  }{
    \judgRewrE{\tenv}{
      \lam{\var}{\redseq}
    }{
      \lam{\var}{\tm_0}
    }{
      \lam{\var}{\tm_1}
    }{\typ \imp \typtwo}{\hrsone}
     }
  \indrule{EApp}{
    \judgRewrE{\tenv}{
      \redseq
    }{
      \tm_0
    }{
      \tm_1
    }{\typ \imp \typtwo}{\hrsone}
    \HS
    \judgRewrE{\tenv}{
      \redseqtwo
    }{
      \tmtwo_0
    }{
      \tmtwo_1
    }{\typ}{\hrsone}
  }{
    \judgRewrE{\tenv}{
      \redseq\,\redseqtwo
    }{
      \tm_0\,\tmtwo_0
    }{
      \tm_1\,\tmtwo_1
    }{\typtwo}{\hrsone}
  }
     \\
     \\
  \indrule{ERule}{
    \langle s,t\rangle\in\closed{\ruleset}
    \HS
    \judgTerm{\noenv}{s}{\typ}
    \HS
    \judgTerm{\noenv}{t}{\typ}
  }{
    \judgRewrE{\tenv}{\rulewit}{s}{t}{\typ}{\hrsone}
  }
   \end{array}
\]
}
\end{defi}
There is no explicit reflexivity rule saying that
$\tm \eqt{}{}{\hrsone} \tm$. But note that this is easy to prove
by induction on $\tm$.

\begin{thm}[Thm.~3.12 in~\cite{MayrNipkow98}]
   \label{thm:HOLandHOR}
   $\judgRewrE{\tenv}{
    }{
      \tm
    }{
      \tmtwo
    }{
      \typ
    }{\hrsone}$ iff $\betaetanf{\tm} \stackrel{*}{\leftrightarrow_{\hrsone}}\betaetanf{\tmtwo}$.
   \end{thm}
   The $(\Leftarrow)$ direction follows from observing that $\rightarrow_{\beta,\overline{\eta}}$ and $ \stackrel{*}{\leftrightarrow_{\hrsone}}$ are all included in $\doteq_{\hrsone}$. The $(\Rightarrow)$ direction is by induction on the derivation  of $\judgRewrE{\tenv}{
    }{
      \tm
    }{
      \tmtwo
    }{
      \typ
    }{\hrsone}$.
   
   Higher-Order Rewriting Logic results from removing the $\indrulename{ESymm}$ rule
in~\rdef{body:horeql} and adding a proof witness. Its judgments take the form
$\judgRewrLabelled{\tenv}{\redseq}{\tm}{\tmtwo}{\typ}{\hrsone}$ where
$\redseq$ is an expression that witnesses
the proof of the judgment, called a {\em rewrite}. Given a set of {\em rule symbols}
($\rulewit,\rulewittwo,\hdots$), the set of {\em rewrites}
($\redseq,\redseqtwo,\hdots$) is defined by the grammar:
\begin{center}
$\begin{array}{rcl}
   \redseq ::= \var \,\mid\,\cons \,\mid\,\rulewit \,\mid\,\lam{\var}{\redseq}\,\mid\,\redseq\,\redseq \,\mid\,\redseq \seq \redseq
 \end{array}$
 \end{center}
A rewrite can either be a variable, a constant, a rule symbol,
an abstraction congruence, an application congruence, or a composition.
Note that composition may occur anywhere inside a rewrite\footnote{This is
one of the key differences between our approach and that of~\cite{thesis:bruggink:08}.}.
For the sake of clarity we present the full system for Higher-Order Rewriting Logic next. We assume given an HRS $\hrsone$ such that each rewrite rule $\langle \ell,r\rangle\in\ruleset$ has been assigned a unique rule symbol  $\rulewit$ and shall write $\langle \rulewit, \ell,r\rangle\in\ruleset$ and also use the same notation for $\closed{\ruleset}$. HORL consists of two forms of typing judgments: 
\begin{enumerate}
\item %\inlineItem{1.}~
   $\judgTermEq{\tenv}{\tm}{\tmtwo}{\typ}$,
  meaning that
  $\tm$ and $\tmtwo$ are $\beta\eta$-equivalent
  terms of type $\typ$ under $\tenv$;
and
\item %\inlineItem{2.}~
  $\judgRewrLabelled{\tenv}{\redseq}{\tm}{\tmtwo}{\typ}{\hrsone}$,
  meaning that $\redseq$ is a rewrite
  with source $\tm$ and target $\tmtwo$,
  which are terms of type $\typ$ under $\tenv$.
\end{enumerate}

\begin{defi}[Higher-Order Rewriting Logic]
\ldef{body:typing_terms_and_rewrites}
Term equivalence is defined as follows:
\[
{\small
  \begin{array}{c}
  \indrule{EqBeta}{
    \judgTerm{\tenv,\var:\typ}{\tm}{\typtwo}
    \HS
    \judgTerm{\tenv}{\tmtwo}{\typ}
  }{
    \judgTermEq{\tenv}{(\lam{\var}{\tm})\,\tmtwo}{\tm\subt{\var}{\tmtwo}}{\typtwo}
  }
  \indrule{EqEta}{
    \judgTerm{\tenv,\var:\typ}{\tm}{\typtwo}
    \HS
    \var\notin\fv{\tm}
  }{
    \judgTermEq{\tenv}{\lam{\var}{\tm\,\var}}{\tm}{\typtwo}
    }
    \\
    \\
 \indrule{EqRefl}{
    \judgTerm{\tenv}{\tm}{\typ}
  }{
    \judgTermEq{\tenv}{\tm}{\tm}{\typ}
  }
  \indrule{EqSym}{
    \judgTermEq{\tenv}{\tm}{\tmtwo}{\typ}
  }{
    \judgTermEq{\tenv}{\tmtwo}{\tm}{\typ}
  }
  \indrule{EqTrans}{
    \judgTermEq{\tenv}{\tm}{\tmtwo}{\typ}
    \HS
    \judgTermEq{\tenv}{\tmtwo}{\tmthree}{\typ}
  }{
    \judgTermEq{\tenv}{\tm}{\tmthree}{\typ}
  }
    \\
    \\
    \indrule{EqCongAbs}{
    \judgTermEq{\tenv,\var:\typ}{\tm_0}{\tm_1}{\typtwo}
  }{
    \judgTermEq{\tenv}{\lam{\var}{\tm_0}}{\lam{\var}{\tm_1}}{\typ \imp \typtwo}
  }
  \indrule{EqCongApp}{
    \judgTermEq{\tenv}{\tm_0}{\tm_1}{\typ \imp \typtwo}
    \HS
    \judgTermEq{\tenv}{\tmtwo_0}{\tmtwo_1}{\typ}
  }{
    \judgTermEq{\tenv}{\tm_0\,\tmtwo_0}{\tm_1\,\tmtwo_1}{\typ}
  }
  \end{array}
}
\]
Typing rules for rewrites are as follows:
\[
{\small
  \begin{array}{c}
  \indrule{RVar}{
    (\var:\typ) \in \tenv
  }{
    \judgRewrLabelled{\tenv}{\var}{\var}{\var}{\typ}{\hrsone}
  }
  \indrule{RCon}{
    (\cons:\typ) \in \constantset
  }{
    \judgRewrLabelled{\tenv}{\cons}{\cons}{\cons}{\typ}{\hrsone}
    }
    \\
    \\
  \indrule{RAbs}{
    \judgRewrLabelled{\tenv,\var:\typ}{
      \redseq
    }{
      \tm_0
    }{
      \tm_1
    }{\typtwo}{\hrsone}
  }{
    \judgRewrLabelled{\tenv}{
      \lam{\var}{\redseq}
    }{
      \lam{\var}{\tm_0}
    }{
      \lam{\var}{\tm_1}
    }{\typ \imp \typtwo}{\hrsone}
  }
  \\
  \\
  \indrule{RApp}{
    \judgRewrLabelled{\tenv}{
      \redseq
    }{
      \tm_0
    }{
      \tm_1
    }{\typ \imp \typtwo}{\hrsone}
    \HS
    \judgRewrLabelled{\tenv}{
      \redseqtwo
    }{
      \tmtwo_0
    }{
      \tmtwo_1
    }{\typ}{\hrsone}
  }{
    \judgRewrLabelled{\tenv}{
      \redseq\,\redseqtwo
    }{
      \tm_0\,\tmtwo_0
    }{
      \tm_1\,\tmtwo_1
    }{\typtwo}{\hrsone}
  }
  \\
  \\
  \indrule{RRule}{
    \langle \rulewit, s, t\rangle\in\closed{\ruleset}
    \HS
    \judgTerm{\noenv}{s}{\typ}
    \HS
    \judgTerm{\noenv}{t}{\typ}
  }{
    \judgRewrLabelled{\tenv}{\rulewit}{s}{t}{\typ}{\hrsone}
  }
  \indrule{RTrans}{
    \judgRewrLabelled{\tenv}{
      \redseq
    }{
      \tm_0
    }{
      \tm_1
    }{
      \typ
    }
    {\hrsone}
    \HS
    \judgRewrLabelled{\tenv}{
      \redseqtwo
    }{
      \tm_1
    }{
      \tm_2
    }{
      \typ
    }
    {\hrsone}
  }{
    \judgRewrLabelled{\tenv}{
      \redseq\seq\redseqtwo
    }{
      \tm_0
    }{
      \tm_2
    }{
      \typ
    }
    {\hrsone}
  }
  \\
  \\
  \indrule{RConv}{
    \judgRewrLabelled{\tenv}{
      \redseq
    }{
      \tm'
    }{
      \tmtwo'
    }{
      \typ
    }
    {\hrsone}
    \HS
    \judgTermEq{\tenv}{\tm}{\tm'}{\typ}
    \HS
    \judgTermEq{\tenv}{\tmtwo'}{\tmtwo}{\typ}
  }{
    \judgRewrLabelled{\tenv}{
      \redseq
    }{
      \tm
    }{
      \tmtwo
    }{
      \typ
    }
    {\hrsone}
  }
  \end{array}
}
\]
\end{defi}
Term equivalence rules (prefixed by ``\indrulename{Eq}'') define
the least congruence generated by $\beta$~and~$\eta$-conversion
over simply typed $\lambda$-terms.
Typing rules for rewrites (prefixed by~``\indrulename{R}'') include
the \indrulename{RVar} and \indrulename{RCon} rules,
which express that variables and constants represent identity rewrites.
The \indrulename{RAbs} and \indrulename{RApp} rules
express congruence below abstraction and application.
The \indrulename{RRule} rule allows to use a rule symbol to stand for
a rewrite between its source and its target,
which must be closed terms of the same type.
The \indrulename{RTrans} rule is the typing rule for rewrite composition,
which can be understood as expressing that the rewriting relation is closed by
transitivity.
The \indrulename{RConv} rule states that the source and the target
of a rewrite are regarded up to $\beta\eta$-equivalence.
Note that there are no rules equating rewrites; such rules are the purpose of~\rsec{permutation_equivalence} which introduces permutation equivalence.
\begin{exa}
\lexample{rewrite}
Suppose we assign the following rule symbols to the rewriting rules of~\rexample{hrs:mu_plus_f}:
  $
  \langle \rulewit,
  \consof{mu}(\lam{\vartwo}{\var\,\vartwo}),
      \var\,(\consof{mu}(\lam{\vartwo}{\var\,\vartwo}))
      \rangle
$
and
$
\langle
\rulewittwo,
\consof{f}\,\var,
\consof{g}\,\var
\rangle$.
Recall that $\constantset\eqdef\{\consof{mu}:(\iota\imp\iota)\imp\iota,\consof{f}:\iota\imp\iota\}$.
The reduction sequence of~\rexample{hrs:mu_plus_f} can be represented as a rewrite:
\begin{center}
$\judgRewrLabelled{\noenv}{
   \rulewit\,(\lam{\var}{\consof{f}\,\var})\seq \consof{f}\, (\consof{mu}\,(\lam{\var}{\rulewittwo\,\var}))\seq \rulewittwo \,(\consof{mu}\,(\lam{\var}{\consof{g}\,\var})) 
 }{
   \consof{mu}\,(\lam{\var}{\consof{f}\,\var})
 }{
   \consof{g}\,(\consof{mu}\,(\lam{\var}{\consof{g}\,\var})
 }{
   \iota
 }
 {\hrsone}$
\end{center}
\end{exa}

Inspection of the proof of Theorem~\ref{thm:HOLandHOR} in~\cite{MayrNipkow98} reveals that $\beta$ and $\eta$ are only needed for substitutions in rewrite rules. As a consequence:
 \begin{thm}
   \label{thm:rewrites:and:HORL}
There is a rewrite $\redseq$ such that
$\judgRewrLabelled{\tenv}{\redseq}{\tm}{\tmtwo}{\typ}{\hrsone}$ if and only if $\betaetanf{\tm} \stackrel{*}{\rightarrow_{\hrsone}}\betaetanf{\tmtwo}$.
   \end{thm}
Now that we know that rewrites over an HRS $\hrsone$ are sound and complete with respect to reduction sequences in $\hrsone$, we review some basic properties of rewrites and then focus, in the remaining sections, on equivalences between rewrites. In the sequel we will omit $\hrsone$ in $\judgRewrLabelled{\tenv}{\redseq}{\tm}{\tmtwo}{\typ}{\hrsone}$ and write $\judgRewr{\tenv}{\redseq}{\tm}{\tmtwo}{\typ}$.

\subsection{Properties of Rewrites}

We begin by introducing the notion of source and target of a rewrite and then move on to a few important syntactic properties of terms and rewrites. 
% \SeeAppendix{(detailed statements and proofs can be found in
% \rsec{appendix:rewrites} of~\cite{techReport})}.
% \SeeAppendix{(detailed statements and proofs can be found in Section A of~\cite{techReport})}.
  \begin{defi}[Source and target of a rewrite]
For each rewrite $\redseq$
we define the {\em source} $\rsrc{\redseq}$
and the {\em target} $\rtgt{\redseq}$
as the following terms:
\[
  \begin{array}{rcll}
    \rsrc{\var}
    & \eqdef &
    \var
    \\
    \rsrc{\cons}
    & \eqdef &
    \cons
    \\
    \rsrc{\rulewit}
    & \eqdef &
    \tm
      \HS\text{if $(\rewr{\rulewit}{\tm}{\tmtwo}{\typ}) \in \ruleset$}
    \\
    \rsrc{(\lam{\var}{\redseq})}
    & \eqdef &
    \lam{\var}{\rsrc{\redseq}}
    \\
    \rsrc{(\redseq\,\redseqtwo)}
    & \eqdef &
    \rsrc{\redseq}\,\rsrc{\redseqtwo}
    \\
    \rsrc{(\redseq\seq\redseqtwo)}
    & \eqdef &
    \rsrc{\redseq}
  \end{array}
  \begin{array}{rcll}
    \rtgt{\var}
    & \eqdef &
    \var
    \\
    \rtgt{\cons}
    & \eqdef &
    \cons
    \\
    \rtgt{\rulewit}
    & \eqdef &
    \tmtwo
      \HS\text{if $(\rewr{\rulewit}{\tm}{\tmtwo}{\typ}) \in \ruleset$}
    \\
    \rtgt{(\lam{\var}{\redseq})}
    & \eqdef &
    \lam{\var}{\rtgt{\redseq}}
    \\
    \rtgt{(\redseq\,\redseqtwo)}
    & \eqdef &
    \rtgt{\redseq}\,\rtgt{\redseqtwo}
    \\
    \rtgt{(\redseq\seq\redseqtwo)}
    & \eqdef &
    \rtgt{\redseq}
  \end{array}
\]
\end{defi}

The free variables of
an expression $X$ (which may be a term or a rewrite) are written $\fv{X}$,
and defined as expected, with lambdas binding variables
in their bodies.
We write $\redseq\subt{\var}{\tmtwo}$
for the capture-avoiding substitution of
the variable $\var$ in the rewrite $\redseq$ by $\tmtwo$.
This operation is called {\em rewrite/term substitution}.
Together with term/term substitution $\tm\subt{\var}{\tmtwo}$ (\cf \rsec{hor})
this allows us to write $X\subt{\var}{\tm}$ when $X$ is either a term or a rewrite.

As for the properties, we begin with the \emph{Substitution Lemma}, which holds for any term or rewrite $X$
as long as $\var \neq \vartwo$ and $\var \notin \fv{\tmtwo}$ and is proved by induction on $X$:

\begin{lem}%[Substitution Lemma]
\llem{substitution_lemma}
The equation
$
  X\subt{\var}{\tm}\subt{\vartwo}{\tmtwo}
  =
  X\subt{\vartwo}{\tmtwo}\subt{\var}{\tm\subt{\vartwo}{\tmtwo}}
$
holds.
\end{lem}

\emph{Weakening} holds too, proved by induction on the derivation of the target judgment:

\begin{lem}%[Weakening]
\llem{weakening}
Let $\var \notin \tenv$. Then:
\begin{enumerate}
\item
  If $\judgTerm{\tenv}{\tm}{\typ}$
  then $\judgTerm{\tenv,\var:\typtwo}{\tm}{\typ}$.
\item
  If $\judgTermEq{\tenv}{\tm}{\tmtwo}{\typ}$
  then $\judgTermEq{\tenv,\var:\typtwo}{\tm}{\tmtwo}{\typ}$.
\item
  If $\judgRewr{\tenv}{\redseq}{\tm}{\tmtwo}{\typ}$
  then $\judgRewr{\tenv,\var:\typtwo}{\redseq}{\tm}{\tmtwo}{\typ}$.
\end{enumerate}
\end{lem}
%\begin{proof}
%Straightfoward by induction on the derivation of the target judgment.
%\end{proof}

%(\eg if $\judgRewr{\tenv}{\redseq}{\tm}{\tmtwo}{\typ}$
%then $\judgRewr{\tenv,\var:\typtwo}{\redseq}{\tm}{\tmtwo}{\typ}$)
Substitution of a variable by a term \emph{commutes with the source and target operators}:

\begin{lem}%[Source and target of rewrite/term substitution]
\llem{source_and_target_of_subrt}
If $\judgRewr{\tenv,\var:\typ}{\redseq}{\tm}{\tmtwo}{\typtwo}$
and $\judgTerm{\tenv}{\tmthree}{\typ}$
then
$\rsrc{\redseq\subt{\var}{\tmthree}} = \rsrc{\redseq}\subt{\var}{\tmthree}$
and $\rtgt{\redseq\subt{\var}{\tmthree}} = \rtgt{\redseq}\subt{\var}{\tmthree}$.

% \end{lem}
% \begin{proof}

% \end{proof}

\end{lem}
\begin{proof}
By induction on $\redseq$.
\end{proof}

Substitution also \emph{preserves typing}.
In particular, the third item of the following lemma
expresses that rewrite/terms substitution can
be understood, intuitively,
as the morphism component of a functor $\cdot\subt{\var}{\tmthree}$:
\[
  \indrule{}{
    \tm \overset{\redseq}{\xrightarrow{\hspace*{1cm}}} \tmtwo
  }{
    \tm\subt{\var}{\tmthree} \overset{\redseq\subt{\var}{\tmthree}}{\xrightarrow{\hspace*{1cm}}} \tmtwo\subt{\var}{\tmthree}
  }
\]

\begin{lem}%[Substitution]
\llem{substitution_term_variables}
Let $\judgTerm{\tenv}{\tmthree}{\typ}$. Then:
\begin{enumerate}
\item
  If $\judgTerm{\tenv,\var:\typ}{\tm}{\typtwo}$
  then $\judgTerm{\tenv}{\tm\subt{\var}{\tmthree}}{\typtwo}$.
\item
  If $\judgTermEq{\tenv,\var:\typ}{\tm}{\tmtwo}{\typtwo}$
  then $\judgTermEq{\tenv}{\tm\subt{\var}{\tmthree}}{\tmtwo\subt{\var}{\tmthree}}{\typtwo}$.
\item
  If $\judgRewr{\tenv,\var:\typ}{\redseq}{\tm}{\tmtwo}{\typtwo}$
  then $\judgRewr{\tenv}{\redseq\subt{\var}{\tmthree}}{\tm\subt{\var}{\tmthree}}{\tmtwo\subt{\var}{\tmthree}}{\typtwo}$.
\end{enumerate}
\end{lem}
\begin{proof}
  Each item is by induction on the derivation of the target judgment.
\end{proof}
Terms appearing in valid equality and rewriting judgments
are  typable:
% that is,
% if either $\judgTermEq{\tenv}{\tm}{\tmtwo}{\typ}$
% or $\judgRewr{\tenv}{\redseq}{\tm}{\tmtwo}{\typ}$,
% then $\judgTerm{\tenv}{\tm}{\typ}$ and $\judgTerm{\tenv}{\tmtwo}{\typ}$.

\begin{lem}%[Equal terms are typable]
\llem{equal_terms_are_typable}
\begin{enumerate}
\item  If $\judgTermEq{\tenv}{\tm}{\tmtwo}{\typ}$
then $\judgTerm{\tenv}{\tm}{\typ}$
and $\judgTerm{\tenv}{\tmtwo}{\typ}$.
%\end{lem}
%\begin{proof}
%By induction on the derivation of $\judgTermEq{\tenv}{\tm}{\tmtwo}{\typ}$.
%\end{proof}

%\begin{lem}[Source and target inversion]
%\llem{source_target_inversion}
\item If $\judgRewr{\tenv}{\redseq}{\tm}{\tmtwo}{\typ}$
then
$\judgTerm{\tenv}{\tm}{\typ}$
and
$\judgTerm{\tenv}{\tmtwo}{\typ}$.
\end{enumerate}
\end{lem}
\begin{proof}
The first item is by induction on the derivation of $\judgTermEq{\tenv}{\tm}{\tmtwo}{\typ}$ and the second by induction on the derivation of 
$\judgRewr{\tenv}{\redseq}{\tm}{\tmtwo}{\typ}$.
\end{proof}

Given a typable rewrite,
$\judgRewr{\tenv}{\redseq}{\tm}{\tmtwo}{\typ}$,
the source of $\redseq$ and $\tm$ are not necessarily equal,
but they are interconvertible, and similarly for the target:
\begin{lem}%[Endpoint coherence]
\llem{endpoint:coherence}
If $\judgRewr{\tenv}{\redseq}{\tm}{\tmtwo}{\typ}$
then $\judgTermEq{\tenv}{\tm}{\rsrc{\redseq}}{\typ}$
and $\judgTermEq{\tenv}{\tmtwo}{\rtgt{\redseq}}{\typ}$.
\end{lem}
\begin{proof}
By induction on the derivation of
$\judgRewr{\tenv}{\redseq}{\tm}{\tmtwo}{\typ}$.
\end{proof}
For example,
if $\rewr{\rulewit}{\lam{\var}{\cons\,\var}}{\lam{\var}{\constwo}}{\typ\to\typ}$
then it can be shown that
$\judgRewr{}{\rulewit\,\constwo}{\cons\,\constwo}{\constwo}{\typ}$,
and indeed
$\cons\,\constwo \termeq (\lam{\var}{\cons\,\var})\,\constwo = \rsrc{(\rulewit\,\constwo)}$.

Any typable term $\tm$ can be understood as an empty or \emph {unit}
rewrite $\refl{\tm}$,
without occurrences of rule symbols, between $\tm$ and itself. This allows us to coerce terms to rewrites implicitly
if there is little danger of confusion.
% if $\judgTerm{\tenv}{\tm}{\typ}$
% then $\judgRewr{\tenv}{\refl{\tm}}{\tm}{\tm}{\typ}$.

\begin{lem}%[Reflexivity]
\llem{reflexivity}
If $\judgTerm{\tenv}{\tm}{\typ}$
then $\judgRewr{\tenv}{\refl{\tm}}{\tm}{\tm}{\typ}$.
\end{lem}
\begin{proof}
Straightforward by induction on the derivation of the judgment.
\end{proof}

{\em Term/rewrite substitution}, defined below, generalizes
term/term substitution $\tm\subt{\var}{\tmtwo}$
when $\tmtwo$ is a rewrite, \ie $\tm\subtr{\var}{\redseq}$.
Sometimes we also call this notion {\em lifting substitution},
as $\tm\subtr{\var}{\redseq}$ ``lifts''
the expression $\tm$ from the level of terms to the level of rewrites.

\begin{defi}[Term/rewrite substitution]
  \ldef{term_rewrite_substitution}
\[
  \begin{array}{rcl@{\HS}rcl}
    \vartwo\subtr{\var}{\redseq}
    & \eqdef &
    \begin{cases}
      \redseq & \text{if $\var = \vartwo$} \\
      \vartwo & \text{if $\var \neq \vartwo$}
    \end{cases}
  &
    \cons\subtr{\var}{\redseq}
    & \eqdef &
    \cons
  \\
    (\lam{\vartwo}{\tm})\subtr{\var}{\redseq}
    & \eqdef &
    \lam{\vartwo}{\tm\subtr{\var}{\redseq}}
    \HS\text{if $\var \neq \vartwo$}
  &
    (\tm\,\tmtwo)\subtr{\var}{\redseq}
    & \eqdef &
    \tm\subtr{\var}{\redseq}\,
    \tmtwo\subtr{\var}{\redseq}
  \end{array}
\]
\end{defi}

We mention some important properties of term/rewrite substitution. First, rewrite/term substitution commutes with source and target operations:
\begin{lem}%[Source and target of rewrite/term substitution]
\llem{source_and_target_of_subtr}
If $\judgTerm{\tenv,\var:\typ}{\tmthree}{\typtwo}$
and $\judgRewr{\tenv}{\redseq}{\tm}{\tmtwo}{\typ}$
then
$\rsrc{\tmthree\subtr{\var}{\redseq}} = \tmthree\subtr{\var}{\rsrc{\redseq}}$
and $\rtgt{\tmthree\subtr{\var}{\redseq}} = \tmthree\subtr{\var}{\rtgt{\redseq}}$.
\end{lem}
\begin{proof}
By induction on $\tmthree$.
\end{proof}

Second, term/rewrite substitution can
be understood, intuitively,
as the morphism component of a functor $\tmthree\subt{\var}{\cdot}$:
\[
  \indrule{}{
    \tm\overset{\redseq}{\xrightarrow{\hspace*{1cm}}}\tmtwo
  }{
    \tmthree\subt{\var}{\tm} \overset{\tmthree\subtr{\var}{\redseq}}{\xrightarrow{\hspace*{1cm}}} \tmthree\subt{\var}{\tmtwo}
  }
\]
More precisely:

\begin{lem}%[Typing rule for term/rewrite substitution]
\llem{fundamental_property_of_term_lifting}
If $\judgTerm{\tenv,\var:\typ}{\tmthree}{\typtwo}$
and $\judgRewr{\tenv}{\redseq}{\tm}{\tmtwo}{\typ}$
then
$\judgRewr{\tenv}{
    \tmthree\subtr{\var}{\redseq}
  }{
    \tmthree\subt{\var}{\tm}
  }{
    \tmthree\subt{\var}{\tmtwo}
  }{\typtwo}$.
\end{lem}
\begin{proof}
By induction on the derivation of
$\judgTerm{\tenv,\var:\typ}{\tmthree}{\typtwo}$.
\end{proof}

Third, term/rewrite and rewrite/term substitution commute as follows:
% according
% to the equation
% $
%   \tm
%     \subtr{\var}{\redseq}
%     \subt{\vartwo}{\tmtwo}
%   =
%   \tm
%     \subt{\vartwo}{\tmtwo}
%     \subtr{\var}{\redseq\subt{\vartwo}{\tmtwo}}
% $,
% assuming that
% $\judgTerm{\tenv,\var:\typ,\vartwo:\typtwo}{\tm}{\typthree}$
% and $\judgRewr{\tenv,\vartwo:\typtwo}{\redseq}{\tmthree}{\tmthree'}{\typ}$
% and $\judgTerm{\tenv}{\tmtwo}{\typtwo}$
%(where, by convention, $\var\notin\fv{\tmtwo}$).
\begin{lem}%[Commutation of lifting and term substitution (I)]
\llem{subtr_subt_commutation_I}
If
$\judgTerm{\tenv,\var:\typ,\vartwo:\typtwo}{\tm}{\typthree}$
and
$\judgRewr{\tenv,\vartwo:\typtwo}{\redseqthree}{\tmfour_0}{\tmfour_1}{\typ}$
and
$\judgTerm{\tenv}{\tmfive}{\typtwo}$
then:
\[
  \tm
    \subtr{\var}{\redseqthree}
    \subt{\vartwo}{\tmfive}
  =
  \tm
    \subt{\vartwo}{\tmfive}
    \subtr{\var}{\redseqthree\subt{\vartwo}{\tmfive}}
\]
In particular, if $\vartwo$ does not occur free in $\redseqthree$,
then
$
\tm\subtr{\var}{\redseqthree}\subt{\vartwo}{\tmfive}
=
\tm\subt{\vartwo}{\tmfive}\subtr{\var}{\redseqthree}
$.
\end{lem}
\begin{proof}
By induction on the derivation of
$\judgTerm{\tenv,\var:\typ,\vartwo:\typtwo}{\tm}{\typthree}$.
\end{proof}

% Fourth, term/rewrite substitution commutes with reflexivity:
% \begin{lem}%[Lifting reflexivity]
% \llem{lifting_reflexivity}
% Let $\judgTerm{\tenv,\var:\typ}{\tm}{\typtwo}$
% and $\judgTerm{\tenv}{\tmtwo}{\typ}$.
% Then $\refl{\tm\subt{\var}{\tmtwo}} = \tm\subtr{\var}{\refl{\tmtwo}}$.
% \end{lem}
% \begin{proof}
% By induction on the derivation of
% $\judgTerm{\tenv,\var:\typ}{\tm}{\typtwo}$.
% \end{proof}
%\pablo{[[Lema eliminado]]}

It is also easy to show that term/rewrite substitution commutes with reflexivity
in the sense that $\refl{\tm\subt{\var}{\tmtwo}} = \tm\subtr{\var}{\refl{\tmtwo}}$,
and with source and target operators,
in the sense that
$\rsrc{\tm\subtr{\var}{\redseq}} = \tm\subt{\var}{\rsrc{\redseq}}$
and $\rtgt{\tm\subtr{\var}{\redseq}} = \tm\subt{\var}{\rtgt{\redseq}}$.
% hold whenever $\judgTerm{\tenv,\var:\typ}{\tm}{\typtwo}$
% and $\judgRewr{\tenv}{\redseq}{\tmtwo}{\tmtwo'}{\typ}$.

%%% Local Variables:
%%% mode: latex
%%% TeX-master: "main"
%%% End:

\section{Permutation equivalence}
\lsec{permutation_equivalence}
This section presents \emph{permutation equivalence}~(\rdef{body:permutation_equivalence}),
an equivalence relation over (typed) rewrites $\redseq\permeq\redseqtwo$ that identifies any
two rewrites $\redseq$ and $\redseqtwo$
---in a given orthogonal HRS~$\hrsone$---
when they denote computations that, intuitively speaking,
perform the ``same computational work''.
In~\rsec{flat_permutation_equivalence}, we will show
that, as the name suggests, permutation equivalence essentially amounts
to identifying computations up to permutation of steps.

%\subparagraph*{Towards Permutation Equivalence for Rewrites.} 

\begin{defi}[Permutation equivalence]
\ldef{body:permutation_equivalence}
Suppose $\judgRewr{\tenv}{\redseq}{\tm}{\tmtwo}{\typ}$ and $\judgRewr{\tenv}{\redseq'}{\tm'}{\tmtwo'}{\typ}$ are derivable. Permutation equivalence, written $\judgPermR{\tenv}{\redseq}{\tm}{\tmtwo}{\redseq'}{\tm'}{\tmtwo'}{\typ}$ (or simply $\redseq \permeq \redseq'$
if $\tenv,\tm,\tmtwo,\tm',\tmtwo',\typ$ are clear from the context), 
is defined as the reflexive, symmetric, transitive, and contextual closure
of the following axioms:
\[
  \begin{array}{rclrl}
    \refl{\rsrc{\redseq}}\seq\redseq
    & \permeq &
    \redseq
    &
    & \permeqRule{IdL}
  \\
    \redseq\seq\refl{\rtgt{\redseq}}
    & \permeq &
    \redseq
    &
    & \permeqRule{IdR}
  \\
    (\redseq\seq\redseqtwo)\seq\redseqthree
    & \permeq &
    \redseq\seq(\redseqtwo\seq\redseqthree)
    &
    & \permeqRule{Assoc}
  \\
    (\lam{\var}{\redseq})\seq(\lam{\var}{\redseqtwo})
    & \permeq &
    \lam{\var}{(\redseq\seq\redseqtwo)}
    &
    & \permeqRule{Abs}
  \\
    (\redseq_1\redseq_2)\seq(\redseqtwo_1\redseqtwo_2)
    & \permeq &
    (\redseq_1\seq\redseqtwo_1)(\redseq_2\seq\redseqtwo_2)
    &
    & \permeqRule{App}
  \\
    (\lam{\var}{\refl{\tm}})\,\redseq
    & \permeq &
    \tm\subtr{\var}{\redseq}
    &
    & \permeqRule{BetaTR}
  \\
    (\lam{\var}{\redseq})\,\refl{\tm}
    & \permeq &
    \redseq\subt{\var}{\tm}
    &
    & \permeqRule{BetaRT}
  \\
    \lam{\var}{\redseq\,\var}
    & \permeq &
    \redseq
    &
    \text{if $\var \notin \fv{\redseq}$}
    & \permeqRule{Eta}
  \end{array}
\]
\end{defi}
Rules $\permeqRule{IdL}$, $\permeqRule{IdR}$ and $\permeqRule{Assoc}$, state
that rewrites together with rewrite composition have a monoidal structure.
Recall from~\rsec{rewrites} that  $\rsrc{\redseq}$ is a term and
$\refl{\rsrc{\redseq}}$ is its corresponding rewrite. Rules $\permeqRule{Abs}$
and $\permeqRule{App}$ state that rewrite composition
commutes
with abstraction and application.
An important thing to be wary of is that rules
may be applied only if both the left and the right-hand sides are well-typed.
In particular, the right-hand side of the \permeqRule{App} rule may
not be well-typed even if the left-hand side is; for example
given rule symbols
$\cons:\typ\to\typtwo$ and $\constwo:\typ$,
the expression $((\lam{\var}{\var})(\cons\,\constwo))\seq(\cons\,\constwo)$
is well-typed, with source and target $\cons\,\constwo$,
while ${((\lam{\var}{\var})\seq\cons)\,((\cons\,\constwo)\seq\constwo)}$
is not well-typed.
Finally, rules $\permeqRule{BetaTR}$,
$\permeqRule{BetaRT}$ and $\permeqRule{Eta}$ introduce $\beta\eta$-equivalence
for rewrites.  Note that  $\permeqRule{BetaTR}$ and $\permeqRule{BetaRT}$
restrict either the body of the abstraction or the argument to a unit rewrite,
thus avoiding the issue mentioned in the introduction where a naive combination
of composition and $\beta\eta$-equivalence can lead to invalid rewrites.
Permutation equivalent rewrites are coinitial and cofinal modulo $\termeq$, a property which may easily be proved by induction on the derivation of $\redseq \permeq \redseqtwo$:
\begin{lem}%[Equivalence of endpoints of permutation equivalent rewrites]
\llem{permeq_endpoints_are_termeq}
Let $\judgRewr{\tenv}{\redseq}{\tmfour_0}{\tmfour_1}{\redseqtwo}$
and $\judgRewr{\tenv}{\redseqtwo}{\tmfive_0}{\tmfive_1}{\redseqtwo}$
be such that $\redseq \permeq \redseqtwo$.
Then $\rsrc{\redseq} \termeq \rsrc{\redseqtwo}$
and $\rtgt{\redseq} \termeq \rtgt{\redseqtwo}$.
\end{lem}

Note that there are no explicit sequencing equations such as the I/O equations\footnote{
$I:\varrho(\sigma_1,...,\sigma_n) \permeq l(\sigma_1,...,\sigma_n)\cdot \varrho(t_1,...,t_n)$ and $
O:\varrho(\sigma_1,...,\sigma_n) \permeq \varrho (s_1,...,s_n)\cdot r(\sigma_1,...,\sigma_n)$
} defining permutation equivalence in the first-order case~\cite{Terese03} and the corresponding equations flat-l and flat-r of~\cite{thesis:bruggink:08} for the higher-order case. Nonetheless, we can derive the following
coherence equation whose proof is by induction on the derivation of
$\judgRewr{\tenv,\var:\typ}{
    \redseq}{\tmfour_0}{\tmfour_1}{\typtwo}$ and which will motivate our upcoming notion of substitution of rewrites (\rdef{rewrite_rewrite_substitution}):
% (\SeeAppendix{see \rlem{coherence} in~\cite{techReport} for
% the proof}):
% \eduso{(\SeeAppendix{see Lem. 63 and Lem. 64 in~\cite{techReport} for
% the proof}):}
% \begin{center}
%   $
%   \begin{array}{rcll}
%     \eduso{\redseq\subt{\var}{\tm'}\seq\tmtwo\subtr{\var}{\redseqtwo}}
%   & \permeq &
%   \eduso{\tm\subtr{\var}{\redseqtwo}\seq\redseq\subt{\var}{\tmtwo'}}
%     & (\permeqRule{Perm})
%     \end{array}
%     $
%     \end{center}
% \eduso{where
% $\judgRewr{\tenv,\var:\typ}{
%     \redseq
%   }{\tm}{\tmtwo}{\typtwo}$
% and
% $\judgRewr{\tenv}{\redseqtwo}{\tm'}{\tmtwo'}{\typ}$.}
\begin{lem}%[Coherence of term variable substitution]
\llem{coherence}
Let
$\judgRewr{\tenv,\var:\typ}{
    \redseq
  }{\tmfour_0}{\tmfour_1}{\typtwo}$
and
$\judgRewr{\tenv}{\redseqtwo}{\tmfive_0}{\tmfive_1}{\typ}$.
Then:
\begin{center}
  $
  \begin{array}{rcll}
    \redseq\subt{\var}{\tmfive_0}\seq\tmfour_1\subtr{\var}{\redseqtwo}
  & \permeq &
  \tmfour_0\subtr{\var}{\redseqtwo}\seq\redseq\subt{\var}{\tmfive_1}
    & (\permeqRule{Perm})
    \end{array}
    $
  \end{center}
Graphically, the following diagram commutes:
\[
  \xymatrix@C=4cm{
    \tm\subt{\var}{\tm'}
      \ar^-{\redseq\subt{\var}{\tm'}}[r]
      \ar_-{\tm\subtr{\var}{\redseqtwo}}[d]
  &
    \tmtwo\subt{\var}{\tm'}
      \ar^-{\tmtwo\subtr{\var}{\redseqtwo}}[d]
  \\
    \tm\subt{\var}{\tmtwo'}
      \ar_-{\redseq\subt{\var}{\tmtwo'}}[r]
  &
    \tmtwo\subt{\var}{\tmtwo'}
  }
\]
\end{lem}

\begin{exa}
\lexample{mu_permeq}
Consider the two following reduction sequences,
where ($R_2$) is the same as in \rexample{hrs:mu_plus_f}:
  \[
    \begin{array}{lll}
      R_1 & : & \consof{mu}\,(\lam{\var}{\consof{f}\,\var}) \rewto \consof{mu}\,(\lam{\var}{\consof{g}\,\var}) \rewto \consof{g}\,(\consof{mu}\,(\lam{\var}{\consof{g}\,\var}))\\
      R_2 & : & \consof{mu}\,(\lam{\var}{\consof{f}\,\var}) \rewto \consof{f}\,(\consof{mu}\,(\lam{\var}{\consof{f}\,\var}))\rewto \consof{f}\,(\consof{mu}\,(\lam{\var}{\consof{g}\,\var})) \rewto \consof{g}\,(\consof{mu}\,(\lam{\var}{\consof{g}\,\var}))
\end{array}
\]
Reduction sequence $R_1$ can be encoded as the rewrite $\consof{mu}\,(\lam{\var}{\rulewittwo\,\var}) ; \rulewit\,(\lam{\var}{\consof{g}\,\var})$ and $R_2$ as $\rulewit\,(\lam{\var}{\consof{f}\,\var})\seq \consof{f}\, (\consof{mu}\,(\lam{\var}{\rulewittwo\,\var}))\seq \rulewittwo \,(\consof{mu}\,(\lam{\var}{\consof{g}\,\var})) $. These two rewrites are permutation equivalent:
  \[
  \begin{array}{lllll}
    & \consof{mu}\,(\lam{\var}{\rulewittwo\,\var}) \seq \rulewit\,(\lam{\var}{\consof{g}\,\var})
    \\
    \permeqBy{Eta} & \consof{mu}\,\rulewittwo \seq \rulewit\,\consof{g}
    \\
    = &   ( \consof{mu}\,\vartwo)\subtr{\vartwo}{\rulewittwo}\seq( \rulewit\,\vartwo)\subt{\vartwo}{\consof{g}}
    \\
    \permeqBy{Perm} & (\rulewit\,\vartwo)\subt{\vartwo}{\consof{f}}\seq(\vartwo\,(\consof{mu}\,\vartwo))\subtr{\vartwo}{\rulewittwo}
    \\
    = & \rulewit\,\consof{f}\seq \rulewittwo\,(\consof{mu}\,\rulewittwo)
    \\
    \permeqBy{IdL} & \rulewit\,\consof{f}\seq (\consof{f}\seq\rulewittwo)\,(\consof{mu}\,\rulewittwo)
    \\
    \permeqBy{IdR} & \rulewit\,\consof{f}\seq (\consof{f}\seq\rulewittwo)\,((\consof{mu}\,\rulewittwo)\seq (\consof{mu}\,\consof{g}))
    \\
    \permeqBy{App} & \rulewit\,\consof{f}\seq \consof{f}\, (\consof{mu}\,\rulewittwo)\seq \rulewittwo \,(\consof{mu}\,\consof{g})
    \\
    \permeqBy{Eta} & \rulewit\,(\lam{\var}{\consof{f}\,\var})\seq \consof{f}\, (\consof{mu}\,(\lam{\var}{\rulewittwo\,\var}))\seq \rulewittwo \,(\consof{mu}\,(\lam{\var}{\consof{g}\,\var}))
  \end{array} 
\]
\end{exa}

\subsection{Substitution of Rewrites in Rewrites.}
The $\permeqRule{Perm}$ rule of \rlem{coherence}
motivates the following definition of substitution of a rewrite inside another rewrite:
\begin{defi}[Rewrite/Rewrite Substitution]
\ldef{rewrite_rewrite_substitution}
Let $\judgRewr{\tenv,\var:\typ}{\redseq}{\tm}{\tmtwo}{\typtwo}$
and let $\judgRewr{\tenv}{\redseqtwo}{\tm'}{\tmtwo'}{\typ}$. Then
$
  \redseq\subrr{\var}{\redseqtwo} \eqdef
  \redseq\subt{\var}{\tm'}\seq\tmtwo\subtr{\var}{\redseqtwo}
$.
Note that, by the $\permeqRule{Perm}$ rule~(\rlem{coherence}),
$
  \redseq\subrr{\var}{\redseqtwo} \permeq
  \tm\subtr{\var}{\redseqtwo}\seq\redseq\subt{\var}{\tmtwo'}
$,
which could be taken as an alternative definition.
Graphically:
\[
  \xymatrix@C=4cm{
    \tm\subt{\var}{\tm'}
      \ar@{.>}^-{\redseq\subt{\var}{\tm'}}[r]
      \ar@{.>}_-{\tm\subtr{\var}{\redseqtwo}}[d]
      \ar^-{\redseq\subrr{\var}{\redseqtwo}}[dr]
  &
    \tmtwo\subt{\var}{\tm'}
      \ar@{.>}^-{\tmtwo\subtr{\var}{\redseqtwo}}[d]
  \\
    \tm\subt{\var}{\tmtwo'}
      \ar@{.>}_-{\redseq\subt{\var}{\tmtwo'}}[r]
  &
    \tmtwo\subt{\var}{\tmtwo'}
  }
\]
\end{defi}
One subtle remark is that $\redseq\subrr{\var}{\redseqtwo}$
depends on $\tmtwo$ and $\tm'$,
and hence on the particular typing derivations
for $\redseq$ and $\redseqtwo$.
\rlem{congruence_permeq_subs}(3) below
ensures that the value of $\redseq\subrr{\var}{\redseqtwo}$
does not depend, up to permutation equivalence, on those typing
derivations.
It relies on a number of preliminary results,
all of which express that equivalence.
The proof of \rlem{congruence_termeq_subs}(1) is by induction on the derivation of $\redseq$
and that of \rlem{congruence_termeq_subs}(2) by induction on the derivation of
$\judgTermEq{\tenv,\var:\typ}{\tm}{\tm'}{\typtwo}$.
\begin{lem}[Congruence of term equivalence below substitution]
\llem{congruence_termeq_subs}
%\llem{congruence_termeq_subrt}
\quad
\begin{enumerate}
\item
  Let $\judgRewr{\tenv,\var:\typ}{\redseq}{\tm}{\tmtwo}{\typtwo}$
  and $\judgTermEq{\tenv}{\tmthree}{\tmthree'}{\typ}$.
  Then $\redseq\subt{\var}{\tmthree} \permeq \redseq\subt{\var}{\tmthree'}$.
% \begin{lem}[$\termeq$ under term/rewrite substitution]
% \llem{congruence_termeq_subtr}
\item
  Let $\judgTermEq{\tenv,\var:\typ}{\tmthree}{\tmthree'}{\typtwo}$
  and $\judgRewr{\tenv}{\redseq}{\tm}{\tmtwo}{\typ}$.
  Then $\tmthree\subtr{\var}{\redseq} \permeq \tmthree'\subtr{\var}{\redseq}$.
\end{enumerate}
\end{lem}
% \begin{proof}
% By induction on the derivation of the judgment
% $\judgTermEq{\tenv,\var:\typ}{\tm}{\tm'}{\typtwo}$.
% Reflexivity, symmetry, transitivity, and congruence under term
% constructors are immediate by \ih. The interesting cases are:
% \begin{enumerate}
% \item \indrulename{EqBeta}:
%   let $\tm = (\lam{\var}{\tmtwo})\,\tmthree$
%   and $\tm' = \tmtwo\subt{\var}{\tmthree}$.
%   Then:
%   \[
%     \begin{array}{rcll}
%       ((\lam{\var}{\tmtwo})\,\tmthree)\subtr{\var}{\redseq}
%     & = &
%       (\lam{\var}{\tmtwo\subtr{\var}{\redseq}})\,\tmthree\subtr{\var}{\redseq}
%     \\
%     & \permeq &
%       \tmtwo\subtr{\var}{\redseq}\subrr{\vartwo}{\tmthree\subtr{\var}{\redseq}}
%       & \text{by \permeqRule{BetaRR}~(\rlem{betaRR})}
%     \\
%     & \permeq &
%       \tmtwo\subt{\vartwo}{\tmthree}\subtr{\var}{\redseq}
%       & \text{by \rlem{subtr_subt_commutation_II}}
%     \end{array}
%   \]
% \item \indrulename{EqEta}:
%   let $\tm = \lam{\vartwo}{\tm'\,\var}$ with $\vartwo \notin \fv{\tm'}$.
%   Then:
%   \[
%     \begin{array}{rcll}
%       (\lam{\vartwo}{\tm'\,\vartwo})\subtr{\var}{\redseq}
%     & = &
%       \lam{\vartwo}{\tm'\subtr{\var}{\redseq}\,\vartwo}
%     \\
%     & \permeq &
%       \tm'
%       & \text{by \permeqRule{Eta}, as $\vartwo \notin \fv{\tm'\subtr{\var}{\redseq}}$}
%     \end{array}
%   \]
%   In the last step, note that we may assume that
%   $\vartwo \notin \fv{\redseq}$ by Barendregt's convention.
% \end{enumerate}
% \end{proof}

\begin{lem}[Congruence of rewrite equivalence below substitution]
  \llem{congruence_permeq_subs}
%  \llem{congruence_permeq_subrt}
\quad
\begin{enumerate}
\item Let $\judgRewr{\tenv,\var:\typ}{\redseq}{\tm}{\tmtwo}{\typtwo}$
and $\judgRewr{\tenv,\var:\typ}{\redseq'}{\tm'}{\tmtwo'}{\typtwo}$
be such that $\redseq \permeq \redseq'$,
and let $\judgTerm{\tenv}{\tmthree}{\typ}$.
Then $\redseq\subt{\var}{\tmthree} \permeq \redseq'\subt{\var}{\tmthree}$.

\item Let $\judgTerm{\tenv,\var:\typ}{\tmthree}{\typtwo}$
and let $\judgRewr{\tenv}{\redseq}{\tm}{\tmtwo}{\typ}$
and $\judgRewr{\tenv}{\redseq'}{\tm'}{\tmtwo'}{\typ}$
such that $\redseq \permeq \redseq'$.
Then $\tmthree\subtr{\var}{\redseq} \permeq \tmthree\subtr{\var}{\redseq'}$.

\item Let $\judgRewr{\tenv,\var:\typ}{\redseq}{\tm}{\tmtwo}{\typtwo}$
and $\judgRewr{\tenv,\var:\typ}{\redseq'}{\tm'}{\tmtwo'}{\typtwo}$
be such that $\redseq \permeq \redseq'$,
and
let $\judgRewr{\tenv}{\redseqtwo}{\tmthree}{\tmfour}{\typ}$
and $\judgRewr{\tenv}{\redseqtwo'}{\tmthree'}{\tmfour'}{\typ}$
be such that $\redseqtwo \permeq \redseqtwo'$.
Then
$\redseq\subrr{\var}{\redseqtwo} \permeq \redseq'\subrr{\var}{\redseqtwo'}$.
\end{enumerate}
\end{lem}
\begin{proof}
The first item is by induction on the derivation of $\redseq\permeq \redseq'$ and the second on that of 
  $\judgTerm{\tenv,\var:\typ}{\tmthree}{\typtwo}$. For the third we reason as follows:
\[
 \begin{array}{rcll}
   \redseq\subrr{\var}{\redseqtwo}
 & = &
   \redseq\subt{\var}{\tmthree}
   \seq
   \tmtwo\subtr{\var}{\redseqtwo}
 \\
 & \permeq &
   \redseq'\subt{\var}{\tmthree}
   \seq
   \tmtwo\subtr{\var}{\redseqtwo}
   & \text{by \rlem{congruence_permeq_subs}(1)}
 \\
 & \permeq &
   \redseq'\subt{\var}{\tmthree'}
   \seq
   \tmtwo\subtr{\var}{\redseqtwo}
   & \text{by \rlem{congruence_termeq_subs}(1)}
 \\
 & \permeq &
   \redseq'\subt{\var}{\tmthree'}
   \seq
   \tmtwo'\subtr{\var}{\redseqtwo}
   & \text{by \rlem{congruence_termeq_subs}(2)}
 \\
 & \permeq &
   \redseq'\subt{\var}{\tmthree'}
   \seq
   \tmtwo'\subtr{\var}{\redseqtwo'}
   & \text{by \rlem{congruence_permeq_subs}(2)}
 \\
 & = &
   \redseq'\subrr{\var}{\redseqtwo'}
 \end{array}
\]
\end{proof}

Rewrite/rewrite substitution generalizes
rewrite/term and term/rewrite substitution, in the sense that
$\redseq\subt{\var}{\tmtwo} \permeq \redseq\subrr{\var}{\refl{\tmtwo}}$
and $\tm\subtr{\var}{\redseq} \permeq \refl{\tm}\subrr{\var}{\redseq}$.

Other important facts involving rewrite/rewrite substitution
are the following.
First, it commutes with abstraction, application, and composition:
% that is
% $(\lam{\vartwo}{\redseq})\subrr{\var}{\redseqtwo} \permeq
%   \lam{\vartwo}{\redseq\subrr{\var}{\redseqtwo}}$,
% $(\redseq_1\,\redseq_2)\subrr{\var}{\redseqtwo}
%  \permeq
%  \redseq_1\subrr{\var}{\redseqtwo}\,\redseq_2\subrr{\var}{\redseqtwo}$,
% and
% $(\redseq_1\seq\redseq_2)\subrr{\var}{\redseqtwo_1\seq\redseqtwo_2}
%   \permeq
%   \redseq_1\subrr{\var}{\redseqtwo_1}\seq\redseq_2\subrr{\var}{\redseqtwo_2}$.

  \begin{lem}%[Recursive equations for rewrite/rewrite substitution]
\llem{subrr_recursion}
Let $\judgRewr{\tenv}{\redseqtwo}{\tmfive_0}{\tmfive_1}{\typ}$.
The following recursive equations hold for rewrite/rewrite substitution:
\begin{enumerate}
\item If $\judgRewr{\tenv,\var:\typ,\vartwo:\typtwo}{\redseq}{\tm_0}{\tm_1}{\typthree}$, then
      $
        (\lam{\vartwo}{\redseq})\subrr{\var}{\redseqtwo}
        \permeq
        \lam{\vartwo}{\redseq\subrr{\vartwo}{\redseqtwo}}
      $.
\item If $\judgRewr{\tenv,\var:\typ}{\redseq_1}{\tm_0}{\tm_1}{\typtwo \imp \typthree}$
      and $\judgRewr{\tenv,\var:\typ}{\redseq_2}{\tmtwo_0}{\tmtwo_1}{\typtwo}$,
      then
      $
        (\redseq_1\,\redseq_2)\subrr{\var}{\redseqtwo}
        \permeq
        \redseq_1\subrr{\var}{\redseqtwo}\,\redseq_2\subrr{\var}{\redseqtwo}
      $.
\item Suppose that:
\begin{itemize}
\item
  $\judgRewr{\tenv,\var:\typ}{\redseq_1}{\tm_0}{\tm_1}{\typtwo}$
  and
  $\judgRewr{\tenv,\var:\typ}{\redseq_2}{\tm_1}{\tm_2}{\typtwo}$
\item
  $\judgRewr{\tenv}{\redseqtwo_1}{\tmtwo_0}{\tmtwo_1}{\typtwo}$
  and
  $\judgRewr{\tenv}{\redseqtwo_2}{\tmtwo_1}{\tmtwo_2}{\typtwo}$
\end{itemize}
Then
$
  (\redseq_1\seq\redseq_2)\subrr{\var}{\redseqtwo_1\seq\redseqtwo_2}
  \permeq
  \redseq_1\subrr{\var}{\redseqtwo_1}\seq\redseq_2\subrr{\var}{\redseqtwo_2}
$.
\end{enumerate}
\end{lem}
\begin{proof}
We check each item separately:
\begin{enumerate}
\item Abstraction:
  \[
    \begin{array}{rcll}
          (\lam{\vartwo}{\redseq})\subrr{\var}{\redseqtwo}
    & = & (\lam{\vartwo}{\redseq})\subt{\var}{\tmfive_0} \seq
          (\lam{\vartwo}{\tm_1})\subtr{\var}{\redseqtwo}
    \\
    & = & (\lam{\vartwo}{\redseq\subt{\var}{\tmfive_0}}) \seq
          (\lam{\vartwo}{\tm_1\subtr{\var}{\redseqtwo}})
    \\
    & \permeq &
          \lam{\vartwo}{(\redseq\subt{\var}{\tmfive_0} \seq
                          \tm_1\subtr{\var}{\redseqtwo})}
      & \text{by \permeqRule{Abs}}
    \\
    & = & \lam{\vartwo}{\redseq\subrr{\var}{\redseqtwo}}
    \end{array}
  \]
\item Application:
  similar to the previous case, using the \permeqRule{App} rule.

\item Composition:
 % \end{enumerate}
% \end{proof}
% \begin{prop}[Transitivity under rewrite/rewrite substitution]
% \lprop{transitivity_subrr}
% \end{prop}
% \begin{proof}
We work implicitly modulo associativity of composition (``$\seq$''),
using the \permeqRule{Assoc} rule:
\[
  \begin{array}{rcll}
    (\redseq_1\seq\redseq_2)\subrr{\var}{\redseqtwo_1\seq\redseqtwo_2}
  & = &
         (\redseq_1\seq\redseq_2)\subt{\var}{\tmtwo_0}
    \seq \tm_2\subtr{\var}{\redseqtwo_1\seq\redseqtwo_2}
  \\
  & = &
         \redseq_1\subt{\var}{\tmtwo_0}
    \seq \redseq_2\subt{\var}{\tmtwo_0}
    \seq \tm_2\subtr{\var}{\redseqtwo_1\seq\redseqtwo_2}
  \\
  & \permeq &
         \redseq_1\subt{\var}{\tmtwo_0}
    \seq \redseq_2\subt{\var}{\tmtwo_0}
    \seq \tm_2\subtr{\var}{\redseqtwo_1}
    \seq \tm_2\subtr{\var}{\redseqtwo_2}
    & \text{by ($\star$)}%\rlem{transitivity_lifting}}
  \\
  & \permeq &
         \redseq_1\subt{\var}{\tmtwo_0}
    \seq \tm_1\subtr{\var}{\redseqtwo_1}
    \seq \redseq_2\subt{\var}{\tmtwo_1}
    \seq \tm_2\subtr{\var}{\redseqtwo_2}
    & \text{by \rlem{coherence}}
  \\
  & \permeq &
         \redseq_1\subrr{\var}{\redseqtwo_1}
    \seq \redseq_2\subrr{\var}{\redseqtwo_2}
  \end{array}
\]
The $(\star)$ is the following property which can be proved by induction on the derivation of
$\judgTerm{\tenv,\var:\typ}{\tm}{\typtwo}$:
% \begin{lem}[Transitivity under lifting substitution]
% \llem{transitivity_lifting}
Let $\judgTerm{\tenv,\var:\typ}{\tm}{\typtwo}$
and
$\judgRewr{\tenv}{\redseq}{\tmfour_0}{\tmfour_1}{\typ}$
and
$\judgRewr{\tenv}{\redseqtwo}{\tmfour_1}{\tmfour_2}{\typ}$.
Then:
$
  \tm\subtr{\var}{\redseq}\seq\tm\subtr{\var}{\redseqtwo}
  \permeq
  \tm\subtr{\var}{\redseq\seq\redseqtwo}
$.
%\end{lem}

\end{enumerate}
\end{proof}

% Second, permutation equivalence is a congruence with respect to
% rewrite/rewrite substitution, that is,
% if $\redseq \permeq \redseq'$ and $\redseqtwo \permeq \redseqtwo'$
% then
% $\redseq\subrr{\var}{\redseqtwo} \permeq \redseq'\subrr{\var}{\redseqtwo'}$. This follows from
% \rlem{permeq_endpoints_are_termeq},
% \rlem{endpoint:coherence},
% \rlem{congruence_permeq_subs} and
% \rlem{subrr_recursion}(1,2). 
Second, an analog of the substitution lemma holds:
% namely
% $\redseq
%    \subrr{\var}{\redseqtwo}
%    \subrr{\vartwo}{\redseqthree}
%  \permeq
%  \redseq
%    \subrr{\vartwo}{\redseqthree}
%    \subrr{\var}{\redseqtwo\subrr{\vartwo}{\redseqthree}}$.
   \begin{prop}%[Substitution property for rewrite/rewrite substitution]
\lprop{substitution_property_for_subrr}
Suppose that
%\begin{itemize}
$\judgRewr{\tenv,\var:\typ,\vartwo:\typtwo}{\redseq}{\tm_0}{\tm_1}{\typthree}$,
$\judgRewr{\tenv,\vartwo:\typtwo}{\redseqtwo}{\tmtwo_0}{\tmtwo_1}{\typ}$, and
$\judgRewr{\tenv}{\redseqthree}{\tmthree_0}{\tmthree_1}{\typtwo}$.
%\end{itemize}
Then
$
  \redseq
    \subrr{\var}{\redseqtwo}
    \subrr{\vartwo}{\redseqthree}
  \permeq
  \redseq
    \subrr{\vartwo}{\redseqthree}
    \subrr{\var}{\redseqtwo\subrr{\vartwo}{\redseqthree}}
$.
\end{prop}
Finally, as discussed above, a $\beta$-rule for arbitrary rewrites
holds in the following form:
% $(\lam{\var}{\redseq})\,\redseqtwo
% \permeq \redseq\subrr{\var}{\redseqtwo}$.
\begin{lem}[Rewrite/rewrite $\beta$-reduction rule]
\llem{betaRR}
Let $\judgRewr{\tenv,\var:\typ}{\redseq}{\tm_0}{\tm_1}{\typtwo}$
and $\judgRewr{\tenv}{\redseqtwo}{\tmtwo_0}{\tmtwo_1}{\typ}$.
Then the following equivalence, called \permeqRule{BetaRR}, holds:
\[
  (\lam{\var}{\redseq})\,\redseqtwo
  \permeq
  \redseq\subrr{\var}{\redseqtwo}
  \HS\HS (\permeqRule{BetaRR})
\]
\end{lem}
\begin{proof}
\[
  \begin{array}{rcll}
    (\lam{\var}{\redseq})\,\redseqtwo
  & \permeq &
    ((\lam{\var}{\redseq})\seq(\lam{\var}{\tm_1}))\,\redseqtwo
    & \permeqRule{IdR}
  \\
  & \permeq &
    ((\lam{\var}{\redseq})\seq(\lam{\var}{\tm_1}))\,(\tmtwo_0\seq\redseqtwo)
    & \permeqRule{IdL}
  \\
  & \permeq &
    (\lam{\var}{\redseq})\,\tmtwo_0 \seq (\lam{\var}{\tm_1})\,\redseqtwo
    & \permeqRule{App}
  \\
  & \permeq &
    \redseq\subt{\var}{\tmtwo_0} \seq (\lam{\var}{\tm_1})\,\redseqtwo
    & \permeqRule{BetaRT}
  \\
  & \permeq &
    \redseq\subt{\var}{\tmtwo_0} \seq \tm_1\subtr{\var}{\redseqtwo}
    & \permeqRule{BetaTR}
  \\
  & = &
    \redseq\subrr{\var}{\redseqtwo}
  \end{array}
\]
\end{proof}

A summary of the three notions of substitution introduced so far is as follows:
\begin{center}
  \begin{tabular}{lll}
    $\redseq\subt{\var}{\tm}$ & rewrite/term & Standard notion of (capture-avoiding) substitution \\
    $\tm\subtr{\var}{\redseq}$ & term/rewrite & \rdef{term_rewrite_substitution}\\
    $\redseq\subrr{\var}{\redseqtwo}$ & rewrite/rewrite & \rdef{rewrite_rewrite_substitution} \\
  \end{tabular}
\end{center}

%%% Local Variables:
%%% mode: latex
%%% TeX-master: "main"
%%% End:

\section{Flattening}
\lsec{flat_permutation_equivalence}

Allowing composition to be nested within
application and abstraction can give rise to rewrites in which it
is not obvious what reduction sequences of steps are being denoted.
An example from the previous
section might be the rewrite $
((\lam{\var}{\consof{f}\,\var})\seq\rulewittwo)\,((\consof{mu}\,(\lam{\var}{\rulewittwo\,\var}))\seq
(\consof{mu}\,(\lam{\var}{\consof{g}\,\var})))$ which denotes a reduction
sequence $\consof{f}\,(\consof{mu}\,(\lam{\var}{\consof{f}\,\var}))\rewto
\consof{g}\,(\consof{mu}\,(\lam{\var}{\consof{g}\,\var}))$ that replaces both
occurrences of $\consof{f}$ with $\consof{g}$ simultaneously. This section
shows how rewrites can be ``flattened'' so as to expose an underlying reduction
sequence, expressed as a canonical or {\em flat} rewrite.

Specifically, we define an operation called {\em flattening}
that given an arbitrary rewrite $\redseq$ computes a permutation equivalent
flat rewrite $\flatten{\redseq}$.
Flat rewrites do not allow the composition operator ``$\seq$''
to appear inside $\lambda$-abstractions
nor applications, so for instance $(\rulewit\seq\rulewittwo)\,\cons$ is not flat, 
whereas the permutation equivalent rewrite $\rulewit\,\cons \seq \rulewittwo\,\cons$
is flat.
Flattening is a way of ensuring
that arbitrary rewrites can always be converted into an equivalent
higher-order proof-term in Bruggink's~\cite{thesis:bruggink:08} sense,
given that, as mentioned in the introduction, Bruggink disallows compositions
to appear inside $\lambda$-abstractions and applications.
In addition, we define an equivalence relation $\redseq \flateq \redseqtwo$
between {\em flat} rewrites, and we show that in general
$(\redseq \permeq \redseqtwo) \iff (\flatten{\redseq} \flateq \flatten{\redseqtwo})$.
This means that flattening can be understood as a {\em semantics} for rewrites.
One additional use of flattening will be to show that permutation equivalence is decidable (\cf~end of~\rsec{projection}).

\subsection{The Flattening Rewrite System}
A {\em multistep} ($\mstep,\msteptwo,\mstepthree,\hdots$) is a rewrite without any occurrences of the composition operator.
The capture-avoiding substitution
of the free occurrences of $\var$ in $\mstep$ by $\msteptwo$
is written $\mstep\subm{\var}{\msteptwo}$,
which is in turn a multistep.
\begin{defi}[Flat Multistep and Flat Rewrite]
\ldef{flat_multistep_and_flat_rewrite}  
  A {\em flat multistep} ($\mstepn,\mstepntwo,\hdots$), is a multistep
in $\beta$-normal form, \ie without subterms of the form $(\lam{\var}{\mstep})\,\msteptwo$.
A {\em flat rewrite} ($\redseqn,\redseqntwo,\hdots$), is a rewrite
given by the grammar
$
    \redseqn ::=  \mstepn \,\mid\, \redseqn\seq\redseqntwo
$.
\end{defi}
Flat rewrites use the composition operator ``$\seq$'' at the top level,
that is
they are of the form $\mstepn_1\seq\hdots\seq\mstepn_n$
(up to associativity of ``$\seq$''), where each $\mstepn_i$
is a flat multistep.
Note that we do not require the $\mstepn_i$
to be in $\eta$-normal form nor in $\etalong$-normal form.
As mentioned in the introduction, flattening is achieved by means of a
{\em rewriting system whose objects are themselves rewrites}~(\rdef{flattening_system})
which is shown to be {\em confluent and terminating}~(\rprop{body:flat_sn}).

We also formulate an equational theory
defining a relation $\redseq \flateq \redseqtwo$
of {\em flat permutation equivalence}
between flat rewrites~(\rdef{body:flat_permutation_equivalence}).
The main result of this section is that
permutation equivalence is {\em sound and complete}
with respect
to flat permutation equivalence~(\rthm{body:soundness_completeness_flateq}).

\begin{rem}
\lremark{multistep_substitution_vs_subtr_subrt}
A substitution $\mstep\subm{\var}{\msteptwo}$
in which $\mstep$ is a term is a term/rewrite substitution,
\ie $\tm\subm{\var}{\msteptwo} = \tm\subtr{\var}{\msteptwo}$.
A substitution
in which $\msteptwo$ is a term is a rewrite/term substitution,
\ie $\mstep\subm{\var}{\tm} = \mstep\subt{\var}{\tm}$.
\end{rem}

\begin{defi}[Flattening Rewrite System $\flatteningSystem$]
\ldef{flattening_system}
The flattening system $\flatteningSystem$ is given by the 
following rules, closed under arbitrary contexts,
defined between \emph{typable} rewrites:
\[
  \begin{array}{rclrl}
    \lam{\var}{(\redseq\seq\redseqtwo)} 
    & \tof &
    (\lam{\var}{\redseq})\seq(\lam{\var}{\redseqtwo})
    &
    & \flatRule{Abs}
  \\
    (\redseq\seq\redseqtwo)\,\mstep
    & \tof &
    (\redseq\,\refl{\rsrc{\mstep}})\seq(\redseqtwo\,\mstep)
    & 
    & \flatRule{App1}
  \\
    \mstep\,(\redseq\seq\redseqtwo)
    & \tof &
    (\mstep\,\redseq)\seq(\refl{\rtgt{\mstep}}\,\redseqtwo)
    & 
    & \flatRule{App2}
  \\
    (\redseq_1\seq\redseq_2)\, (\redseqtwo_1\seq\redseqtwo_2)
    & \tof &
    ((\redseq_1\seq\redseq_2)\,\refl{\rsrc{\redseqtwo_1}})
    \seq
    (\refl{\rtgt{\redseq_2}}\,(\redseqtwo_1\seq\redseqtwo_2))
    & 
    & \flatRule{App3}
  \\
    (\lam{\var}{\mstep})\,\msteptwo
    & \tof &
    \mstep\subm{\var}{\msteptwo}  
    & 
    & \flatRule{BetaM}
  \\
    \lam{\var}{\mstep\,\var}
    & \tof &
    \mstep
    & \text{if $\var \notin \fv{\mstep}$}
    & \flatRule{EtaM}
  \end{array}
\]
Note that rules \flatRule{BetaM} and \flatRule{EtaM} apply to multisteps only.
The reduction relation $\tof$ is the union of all these rules,
closed by compatibility under arbitrary contexts.
We write $\flatten{\redseq}$ for the unique $\tof$-normal form of $\redseq$. Uniqueness will be proved below (\rprop{body:flat_sn}).
\end{defi}

\begin{exa}
Consider a rewriting rule
$\rewr{\rulewit}{\cons}{\constwo}{\typ}$.
The rewrite $(\lam{\var}{(\var\seq\var)})\,\rulewit$
can be flattened as follows:
  \[
  \begin{array}{lll}
    (\lam{\var}{(\var\seq\var)})\,\rulewit
    &
    \tof_{\flatRule{Abs}}
    & ((\lam{\var}{\var}) \seq (\lam{\var}{\var}))\,\rulewit
  \\
    &
    \tof_{\flatRule{App1}}
    & (\lam{\var}{\var})\,\cons \seq (\lam{\var}{\var})\,\rulewit
  \\
    &
    \tof_{\flatRule{BetaM}}
    & \cons \seq (\lam{\var}{\var})\,\rulewit
  \\
    &
    \tof_{\flatRule{BetaM}}
    & \cons \seq \rulewit
  \end{array}
  \]
\end{exa}

% The following result is proved by noting that \flatRule{BetaM} and \flatRule{EtaM} steps can be postponed
% after steps of other kinds and then providing a well-founded measure for steps in  $\flatteningSystem$ without \flatRule{BetaM} and \flatRule{EtaM} to prove it is SN.

\subsection{Basic Properties of the Flattening Rewrite System}

We present some basic properties of  $\flatteningSystem$. First a remark. Flattening does not alter the source or target of a rewrite, modulo $\beta\eta$-equality. Moreover, in the particular case that an applied flattening rule is one not in the set $\set{\flatRule{BetaM},\flatRule{EtaM}}$, then the source and target are literally the same.

\begin{rem}
\lremark{flat_not_beta_steps_preserve_endpoints}
If $\redseq \tof_{\flatRuleAnon} \redseqtwo$
and $\flatRuleAnon \notin \set{\flatRule{BetaM},\flatRule{EtaM}}$
then $\rsrc{\redseq} = \rsrc{\redseqtwo}$
and $\rtgt{\redseq} = \rtgt{\redseqtwo}$.
\end{rem}

Next we consider flattening below term/rewrite and rewrite/term
substitution (\rlem{flattening_below_subs}). The proof the first item is
straightforward by induction on $\tm$ and that of the second by induction on
the context under which the step $\redseq \tof \redseq'$ takes place.

\begin{lem}%[Flattening below term/rewrite substitution]
\llem{flattening_below_subs}
%\llem{flattening_below_subtr}
Suppose $\redseq \tof \redseq'$. Then:
\begin{enumerate}
\item $\tm\subtr{\var}{\redseq} \tofs \tm\subtr{\var}{\redseq'}$.
%\end{lem}
% \begin{proof}
% Straightforward by induction on $\tm$.
% \end{proof}

%\begin{lem}%[Flattening below rewrite/term substitution]
%\llem{flattening_below_subrt}
%If $\redseq \tof \redseq'$
  \item $\redseq\subt{\var}{\tm} \tof \redseq'\subt{\var}{\tm}$.
\end{enumerate}
\end{lem}

Similarly, we have the following proved by induction on the target judgment:
\begin{lem}
\llem{flattening_preserves_betam_and_etam_steps}
    If $\flatRuleAnon \in \set{\flatRule{BetaM},\flatRule{EtaM}}$ then:
    \begin{itemize}
    \item
      $\mstep \tof_{\flatRuleAnon} \mstep'$ implies
      $\mstep\subm{\var}{\msteptwo} \tof_{\flatRuleAnon}
       \mstep'\subm{\var}{\msteptwo}$
    \item
      $\msteptwo \tof_{\flatRuleAnon} \msteptwo'$ implies
        $\mstep\subm{\var}{\msteptwo} \tofsSub{\flatRuleAnon}
       \mstep\subm{\var}{\msteptwo'}$.
    \end{itemize}
    % These are easy to prove,
    % resorting to substitution lemmas for multisteps %\rlem{multistep_substitution_lemma} when appropriate.
    % when appropriate.

  \end{lem}

The following result states that \flatRule{BetaM} and \flatRule{EtaM} steps can be postponed
after steps of other kinds. We will make use of it in our proof of SN of flattening. The proof is by induction on the context under which the step
$\redseq \tof_{\flatRuleAnon} \redseqtwo$ takes place.

\begin{lem}%[\flatRule{BetaM} and \flatRule{EtaM} postponement]
\llem{flat_beta_eta_postponement}
Let $\redseq \tof_{\flatRuleAnon} \redseqtwo
             \tof_{\flatRuleAnonTwo} \redseqthree$
where $\flatRuleAnon \in \set{\flatRule{BetaM},\flatRule{EtaM}}$
and $\flatRuleAnonTwo \notin \set{\flatRule{BetaM},\flatRule{EtaM}}$.
Then there is a rewrite $\redseqfour$
such that $\redseq \tof_{\flatRuleAnonTwo} \redseqfour
                   \tof_{\flatRuleAnon}^+ \redseqthree$.
\end{lem}

The following auxiliary result will be of use in proving confluence of $\tof$; its  proof is by  induction on $\redseq$. Here $\tofeqSub{\flatRuleAnon}$ denotes
the {\em reflexive} closure of $\tof_{\flatRuleAnon}$:
\begin{lem}[Source-Target of Beta/Eta Flattening]
\llem{flatten_beta_eta_source_target}
If $\redseq \tof_{\flatRuleAnon} \redseqtwo$
where $\flatRuleAnon \in \set{\flatRule{BetaM},\flatRule{EtaM}}$,
then $\refl{\rsrc{\redseq}} \tofeqSub{\flatRuleAnon} \refl{\rsrc{\redseqtwo}}$
and $\refl{\rtgt{\redseq}} \tofeqSub{\flatRuleAnon} \refl{\rtgt{\redseqtwo}}$.

\end{lem}

 We next address strong normalization (\rprop{flat_sn}) and confluence (\rprop{flat_confluent}).
\subsection{Strong Normalization of the Flattening Rewrite System}

Since the union of \flatRule{BetaM} and \flatRule{EtaM} can
easily be seen to be SN (given that they essentially operate on
simply typed terms) and in view of \rlem{flat_beta_eta_postponement}, in
order to obtain SN of $\flatteningSystem$ we are left with the task
of proving that the union of all the remaining rules is SN. This can be done by defining an appropriate well-founded measure
on rewrites. Consider
$\#(\redseq) \eqdef (\rheavy{\redseq}, \rweight{\redseq})$ where:

\begin{defi}[Heavy applications]
An application $\redseq\,\redseqtwo$ is {\em heavy}
if $\redseq$ and $\redseqtwo$ are not multisteps,
\ie if both $\redseq$ and $\redseqtwo$ contain compositions~(``$\seq$'').
We write $\rheavy{\redseq}$ to stand for the number of heavy
applications in $\redseq$. More precisely:
\[
  \begin{array}{rcll}
    \rheavy{\var} = \rheavy{\cons} = \rheavy{\rulewit} & \eqdef & 0 \\
    \rheavy{\lam{\var}{\redseq}} & \eqdef & \rheavy{\redseq} \\
    \rheavy{\redseq\,\redseqtwo} & \eqdef &
      \rheavy{\redseq} + \rheavy{\redseqtwo} +
      \begin{cases}
        1 & \text{if $\redseq\,\redseqtwo$ is heavy} \\
        0 & \text{otherwise} \\
      \end{cases}
    \\
    \rheavy{\redseq\seq\redseqtwo} & \eqdef &
      \rheavy{\redseq} + \rheavy{\redseqtwo}
  \end{array}
\]
\end{defi}

\begin{defi}[Weight of a Rewrite]
The {\em weight} of a rewrite $\redseq$
is a non-negative integer $\rweight{\redseq}$
defined inductively as follows:
\[
  \begin{array}{rcll}
    \rweight{\var} = \rweight{\cons} = \rweight{\rulewit} & \eqdef & 0 \\
    \rweight{\lam{\var}{\redseq}} & \eqdef & 2\,\rweight{\redseq} \\
    \rweight{\redseq\,\redseqtwo} & \eqdef &
      2\,\rweight{\redseq} + 2\,\rweight{\redseqtwo}
    \\
    \rweight{\redseq\seq\redseqtwo} & \eqdef &
      1 + \rweight{\redseq} + \rweight{\redseqtwo}
  \end{array}
\]
\end{defi}

We have the following two results, both of which are proved by induction on the context under which the 
step $\redseq \tof_{\flatRuleAnon} \redseqtwo$ takes place:

\begin{lem}[Decrease of heavy applications]
\llem{rheavy_decrease}
Let $\redseq \tof_{\flatRuleAnon} \redseqtwo$
where $\flatRuleAnon \notin \set{\flatRule{BetaM},\flatRule{EtaM}}$.
Then $\rheavy{\redseq} \geq \rheavy{\redseqtwo}$.
Furthermore if $\flatRuleAnon = \flatRule{App3}$
then $\rheavy{\redseq} > \rheavy{\redseqtwo}$.
\end{lem}

\begin{lem}[Decrease of weight]
\llem{rweight_decrease}
Let $\redseq \tof_{\flatRuleAnon} \redseqtwo$
where $\flatRuleAnon \in \set{\flatRule{Abs},\flatRule{App1},\flatRule{App2}}$.
Then $\rweight{\redseq} > \rweight{\redseqtwo}$.
\end{lem}

This suffices to show that $\#(\bullet)$ is well-founded.
A sketch of the entire argument of the proof of SN for $\flatteningSystem$ is presented below.

\begin{prop}
\lprop{flat_sn}
The flattening system $\flatteningSystem$ is strongly normalizing.
\end{prop}
\begin{proof}
Recall that \flatRule{BetaM} and \flatRule{EtaM} steps can be postponed
after steps of other kinds (\rlem{flat_beta_eta_postponement}).
Hence, by standard rewriting techniques,
SN of $\flatteningSystem$
can be reduced
to SN of $\tof_{\flatRule{BetaM}} \cup
          \tof_{\flatRule{EtaM}}$
on one hand,
plus SN of $\tof_{\flatRule{Abs}} \cup
            \tof_{\flatRule{App1}} \cup
            \tof_{\flatRule{App2}} \cup
            \tof_{\flatRule{App3}}$
on the other one.

It is immediate to show that the union of \flatRule{BetaM} and \flatRule{EtaM}
is SN, given that (typable) multisteps can be understood as
simply typed $\lambda$-terms,
by regarding constants ($\cons,\constwo,\hdots$)
and rule symbols ($\rulewit,\rulewittwo,\hdots$) 
as free variables of their corresponding types.
Hence termination of $\tof_{\flatRule{BetaM}} \cup \tof_{\flatRule{EtaM}}$
is reduced to termination of $\beta\eta$-reduction in the simply-typed
$\lambda$-calculus.

To show that the system without \flatRule{BetaM} and \flatRule{EtaM}
is SN, consider the measure on rewrites
given by $\#(\redseq) \eqdef (\rheavy{\redseq}, \rweight{\redseq})$
with the lexicographic order.
It is then easy to show that
if $\redseq \tof_{\flatRuleAnon} \redseqtwo$
with $\flatRuleAnon \notin \set{\flatRule{BetaM},\flatRule{EtaM}}$
then $\#(\redseq) > \#(\redseqtwo)$.
Indeed, by \rlem{rheavy_decrease}
we know that \flatRule{App3} steps strictly decrease the first component 
and other kinds of steps do not increase it.
Moreover, by \rlem{rweight_decrease},
we know that \flatRule{Abs}, \flatRule{App1}, and \flatRule{App2}
steps strictly decrease the second component.
\end{proof}

\subsection{Confluence of the Flattening Rewrite System}

Confluence of $\flatteningSystem$ follows from Newman's lemma. 

\begin{prop}
\lprop{flat_confluent}
The flattening system $\flatteningSystem$ is confluent.
\end{prop}
\begin{proof}
By Newman's lemma, given that $\flatteningSystem$ is SN~(\rprop{flat_sn}),
it suffices to show that it is weakly Church--Rosser (WCR).
Let $\redseq \tof \redseq_1$ and $\redseq \tof \redseq_2$. One shows that there exists a rewrite $\redseq_3$ such
that $\redseq \tofs \redseq_3$ and $\redseq_1 \tofs \redseq_3$.
The proof is by induction on $\redseq$. 
\end{proof}

\begin{prop}
\lprop{body:flat_sn}
The flattening system $\flatteningSystem$
is strongly normalizing and confluent.
\end{prop}

\subsection{Flat Permutation Equivalence.}
We now turn to the definition of the relation $\redseq \flateq \redseqtwo$
of flat permutation equivalence.
The key notion to define is the following ternary relation:

\begin{defi}[Splitting]
Let $\judgRewr{\tenv}{\mstep}{\tm}{\tmtwo}{\typ}$
and $\judgRewr{\tenv}{\mstep_1}{\tm'}{\tmthree_1}{\typ}$
and $\judgRewr{\tenv}{\mstep_2}{\tmthree_2}{\tmtwo'}{\typ}$
be multisteps.
We say that $\mstep$ {\em splits} into $\mstep_1$ and $\mstep_2$
if the following inductively defined ternary relation,
written $\judgSplit{\mstep}{\mstep_1}{\mstep_2}$, holds:
\[
  \indrule{SVar}{
  }{
    \judgSplit{\var}{\var}{\var}
  }
  \indrule{SCon}{
  }{
    \judgSplit{\cons}{\cons}{\cons}
  }
  \indrule{SRuleL}{
  }{
    \judgSplit{\rulewit}{\rulewit}{\refl{\rtgt{\rulewit}}}
  }
  \indrule{SRuleR}{
  }{
    \judgSplit{\rulewit}{\refl{\rsrc{\rulewit}}}{\rulewit}
  }
\]
\[
  \indrule{SAbs}{
    \judgSplit{\mstep}{\mstep_1}{\mstep_2}
  }{
    \judgSplit{\lam{\var}{\mstep}}{\lam{\var}{\mstep_1}}{\lam{\var}{\mstep_2}}
  }
  \indrule{SApp}{
    \judgSplit{\mstep}{\mstep_1}{\mstep_2}
    \HS
    \judgSplit{\msteptwo}{\msteptwo_1}{\msteptwo_2}
  }{
    \judgSplit{\mstep\,\msteptwo}{\mstep_1\,\msteptwo_1}{\mstep_2\,\msteptwo_2}
  }
\]
\end{defi}

\begin{exa}
\lexample{splitting}
Consider the HRS of~\rexample{rewrite}. It has the following two rewriting rules:
  $
  \langle \rulewit,
  \consof{mu}(\lam{\vartwo}{\var\,\vartwo}),
      \var\,(\consof{mu}(\lam{\vartwo}{\var\,\vartwo}))
      \rangle
$
and
$
\langle
\rulewittwo,
\consof{f}\,\var,
\consof{g}\,\var
\rangle$.  Then $\rulewit\,\rulewittwo$ splits into $(\lam{\var}{\consof{mu}\,(\lam{\vartwo}{\var\,\vartwo})})\,
 \rulewittwo$ and $\rulewit\,(\lam{\var}{\consof{g}\,\var})$:
 \begin{center}
   $\judgSplit{
     \rulewit\,\rulewittwo
   }{
     (\lam{\var}{\consof{mu}\,(\lam{\vartwo}{\var\,\vartwo})})\,
     \rulewittwo
   }{
     \rulewit\,(\lam{\var}{\consof{g}\,\var})
   }$
 \end{center}
 Note that $\judgSplit{
              \rulewit\,\rulewittwo
            }{
              (\lam{\var}{\consof{mu}\,(\lam{\vartwo}{\var\,\vartwo})})\,
              \rulewittwo
            }{
              \rulewit\,(\lam{\var}{\consof{g}\,\var})
            }$ follows from $\indrulename{SApp}$, $\indrulename{SRuleR}$ for the upper left hypothesis and $\indrulename{SRuleL}$ for the upper right one.
            
Also,  $\rulewittwo\,(\consof{mu}\,\rulewittwo)$ splits into $ (\lam{\var}{\consof{f}\,\var})\,(\consof{mu}\,\rulewittwo)$ and $\rulewittwo\,(\consof{mu}\,(\lam{\var}{\consof{g}\,\var}))$:
 \begin{center}
   $\judgSplit{
              \rulewittwo\,(\consof{mu}\,\rulewittwo)
            }{
              (\lam{\var}{\consof{f}\,\var})\,(\consof{mu}\,\rulewittwo)
            }{
              \rulewittwo\,(\consof{mu}\,(\lam{\var}{\consof{g}\,\var}))
            }$
          \end{center}
\end{exa}

\begin{defi}[Flat permutation equivalence]
\ldef{body:flat_permutation_equivalence}
Flat permutation equivalence judgments are of the form:
  $\judgFlatPerm{\tenv}{\redseq}{\tm}{\tmtwo}{{\redseq}'}{\tm'}{\tmtwo'}{\typ}$,
meaning that $\redseq$ and $\redseq'$
are equivalent rewrites,
with sources $\tm$ and $\tm'$ respectively,
and targets $\tmtwo$ and $\tmtwo'$ respectively.
The rewrites $\redseq$ and $\redseq'$ are
{\bf assumed to be in $\tof$-normal form},
which in particular means that they must be flat rewrites (\rdef{flat_multistep_and_flat_rewrite}).
Sometimes we write $\redseq \flateq \redseq'$
if $\tenv,\tm,\tmtwo,\tm',\tmtwo',\typ$ are irrelevant
or clear from the context.
Derivability is defined by
the two following axioms,
which are closed by
reflexivity, symmetry, transitivity, and closure under
{\em composition contexts}
(given by
$\sctx ::= \ctxhole \mid \sctx\seq\redseq \mid \redseq\seq\sctx$):
\begin{center}$
  \begin{array}{r@{\ }c@{\ }lrl}
    (\redseq\seq\redseqtwo)\seq\redseqthree
    & \flateq &
    \redseq\seq(\redseqtwo\seq\redseqthree)
    &
    & \flateqRule{Assoc}
    \\
    \mstep
    & \flateq &
    \flatten{\mstep_1}
    \seq
    \flatten{\mstep_2}
    &
      \text{if $\judgSplit{\mstep}{\mstep_1}{\mstep_2}$}
    & \flateqRule{Perm}
  \end{array}
  $
  \end{center}
\end{defi}

Note that in $\flateqRule{Perm}$, $\flatten{-}$ operates over multisteps. So the only rules of $\flatteningSystem$ that are applied here are the \flatRule{BetaM} and \flatRule{EtaM} rules.

\begin{exa}
With the same notation as in~\rexample{mu_permeq},
it can be checked that the rewrites
$\consof{mu}\,(\lam{\var}{\rulewittwo\,\var})
 \seq \rulewit\,(\lam{\var}{\consof{g}\,\var})$
and
$\rulewit\,(\lam{\var}{\consof{f}\,\var})
 \seq \consof{f}\,(\consof{mu}\,(\lam{\var}{\rulewittwo\,\var}))
 \seq \rulewittwo \,(\consof{mu}\,(\lam{\var}{\consof{g}\,\var}))$
are permutation equivalent by means of flattening.
Indeed, using the \flateqRule{Perm} rule three times:
\[
  \begin{array}{lll@{\ }l}
         \consof{mu}\,\rulewittwo
    \seq \rulewit\,\consof{g}
  & \flateq &
      \rulewit\,\rulewittwo
    & \text{\rexample{splitting}}                
  \\
  & \flateq &
         \rulewit\,\consof{f}
    \seq \rulewittwo\,(\consof{mu}\,\rulewittwo)
    & 
  \\
  & \flateq &
         \rulewit\,\consof{f}
    \seq (\consof{f}(\consof{mu}\,\rulewittwo)
          \seq \rulewittwo(\consof{\mu}\,\consof{g}))
    & \text{\rexample{splitting}} 
  \end{array}
\]
  Hence:
  \begin{center}
    $\begin{array}{lll}
    \flatten{(
           \consof{mu}\,(\lam{\var}{\rulewittwo\,\var})
      \seq \rulewit\,(\lam{\var}{\consof{g}\,\var})
    )}
   & = &
    \consof{mu}\,\rulewittwo \seq \rulewit\,\consof{g} \\
  & \flateq & 
    \rulewit\,\consof{f}
    \seq (\consof{f}(\consof{mu}\,\rulewittwo)
          \seq \rulewittwo(\consof{\mu}\,\consof{g})) \\
  & = &
    \flatten{(
             \rulewit\,(\lam{\var}{\consof{f}\,\var})
        \seq \consof{f}\,(\consof{mu}\,(\lam{\var}{\rulewittwo\,\var}))
        \seq \rulewittwo \,(\consof{mu}\,(\lam{\var}{\consof{g}\,\var}))
        )}
        \end{array}
        $
        \end{center}
\end{exa}

We next address soundness and completeness of $\flateq$ with respect to $\permeq$.

\subsection{Soundness of Flat Permutation Equivalence}
Soundness of flat permutation equivalence, namely that if $\judgRewr{\tenv}{\redseq}{\tm}{\tmtwo}{\typ}$
and $\judgRewr{\tenv}{\redseqtwo}{\tm'}{\tmtwo'}{\typ}$, then  $\flatten{\redseq} \flateq \flatten{\redseqtwo}$ implies $\redseq \permeq \redseqtwo$, is stated as \rprop{soundness_flateq}. It follows from the fact that reduction $\tof$ in the flattening system $\flatteningSystem$
is included in permutation equivalence~($\redseq \tof \redseqtwo$ implies $\redseq \permeq \redseqtwo$, \rlem{flattening_sound_wrt_permeq})
and, similarly, flat permutation equivalence
is included in permutation equivalence~($\redseq \flateq \redseqtwo$ implies $\redseq \permeq \redseqtwo$, \rlem{flateq_sound_wrt_permeq}). 

\begin{lem}
\llem{flattening_sound_wrt_permeq}
If $\redseq \tof \redseqtwo$ then $\redseq \permeq \redseqtwo$.
\end{lem}
\begin{proof}
It suffices to show that all the axioms of
the flattening system $\flatteningSystem$
relate permutation equivalent rewrites. We present two sample cases:
\begin{enumerate}
% \item
%   \flatRule{Abs}: 
%   \[
%     \lam{\var}{(\redseq\seq\redseqtwo)} \permeq
%     (\lam{\var}{\redseq})\seq(\lam{\var}{\redseqtwo})
%     \HS\text{by \permeqRule{Abs}}
%   \]
% \item
%   \flatRule{App1}:
%   \[
%     \begin{array}{rcll}
%       (\redseq\seq\redseqtwo)\,\mstep
%     & \permeq &
%       (\redseq\seq\redseqtwo)\,(\refl{\rsrc{\mstep}}\seq\mstep)
%       & \text{by \permeqRule{IdL}}
%     \\
%     & \permeq &
%       (\redseq\,\refl{\rsrc{\mstep}})\seq(\redseqtwo\,\mstep)
%       & \text{by \permeqRule{App}}
%     \end{array}
%   \]
% \item
%   \flatRule{App2}:
%   \[
%     \begin{array}{rcll}
%       \mstep\,(\redseq\seq\redseqtwo)
%     & \permeq &
%       (\mstep\seq\refl{\rtgt{\mstep}})\,(\redseq\seq\redseqtwo)
%       & \text{by \permeqRule{IdR}}
%     \\
%     & \permeq &
%       (\mstep\,\redseq)\seq(\refl{\rtgt{\mstep}}\,\redseqtwo)
%       & \text{by \permeqRule{App}}
%     \end{array}
%   \]
\item
  \flatRule{App3}:
  \[
    \begin{array}{rcll}
      (\redseq_1\seq\redseq_2)(\redseqtwo_1\seq\redseqtwo_2)
    & \permeq &
      ((\redseq_1\seq\redseq_2)\seq\refl{\rtgt{\redseq_2}})
      (\refl{\rsrc{\redseqtwo_1}}\seq(\redseqtwo_1\seq\redseqtwo_2))
      & \text{by \permeqRule{IdL} and \permeqRule{IdR}}
    \\
    & \permeq &
      ((\redseq_1\seq\redseq_2)\,\refl{\rsrc{\redseqtwo_1}})\seq
      (\refl{\rtgt{\redseq_2}}\,(\redseqtwo_1\seq\redseqtwo_2))
      & \text{by \permeqRule{App}}
    \end{array}
  \]
\item
  \flatRule{BetaM}:
  \[
    \begin{array}{rcll}
      (\lam{\var}{\mstep})\,\msteptwo
    & \permeq &
      \mstep\subrr{\var}{\msteptwo}
      & \text{by \permeqRule{BetaRR}~(\rlem{betaRR})}
    \\
    & \permeq &
      \mstep\subm{\var}{\msteptwo}
      & %\text{by \rlem{subm_permeq_subrr}}
    \end{array}
  \]
The last step  follows from the property that  $\mstep\subm{\var}{\msteptwo} \permeq \mstep\subrr{\var}{\msteptwo}$, which is proved by
induction on the derivation of $\judgRewr{\tenv,\var:\typ}{\mstep}{\tmfour_0}{\tmfour_1}{\typtwo}$.

% \item
%   \flatRule{EtaM}:
%   \[
%     \begin{array}{rcll}
%       (\lam{\var}{\mstep})\,\var
%     & \permeq &
%       \mstep
%       & \text{by \permeqRule{Eta}, if $\var\notin\fv{\mstep}$}
%     \end{array}
%   \]
\end{enumerate}
\end{proof}

% \begin{rem}
% Every time that flattening $\flatten{-}$ is used in the
% rules defining $\flateq$, it operates over a multistep.
% So the only rules that are needed are the \flatRule{BetaM} and \flatRule{EtaM}
% rules.
% \end{rem}

% \begin{rem}
% Recall that, by definition, flat rewrites are given by the grammar
% $\redseqn ::= \mstepn \mid \redseqn\seq\redseqn$.
% This corresponds to the set of all and only
% the rewrites of the form $\kctxof{\mstepn_1,\hdots,\mstepn_n}$.
% \end{rem}

\begin{lem}%[Soundness of splitting with respect to permutation equivalence]
\llem{splitting_sound_wrt_permeq}
Let $\judgRewr{\tenv}{\mstep}{\tm}{\tmtwo}{\typ}$
and $\judgRewr{\tenv}{\mstep_1}{\tm'}{\tmthree_1}{\typ}$
and $\judgRewr{\tenv}{\mstep_2}{\tmthree_2}{\tmtwo'}{\typ}$
be such that $\judgSplit{\mstep}{\mstep_1}{\mstep_2}$.
Then $\mstep \permeq \mstep_1\seq\mstep_2$
\end{lem}
\begin{proof}
By induction on the derivation of $\judgSplit{\mstep}{\mstep_1}{\mstep_2}$ using \permeqRule{IdL}, \permeqRule{IdR}, \permeqRule{Abs}, \permeqRule{App} and the \ih.
\end{proof}

\begin{lem}
\llem{flateq_sound_wrt_permeq}
Let $\judgRewr{\tenv}{\redseq}{\tm}{\tmtwo}{\typ}$
and $\judgRewr{\tenv}{\redseqtwo}{\tm'}{\tmtwo'}{\typ}$
be such that $\redseq \flateq \redseqtwo$.
Then $\redseq \permeq \redseqtwo$.
\end{lem}
\begin{proof}
By induction on the derivation of $\redseq \flateq \redseqtwo$.
Reflexivity, transitivity, symmetry, and closure under composition
contexts is immediate.
The interesting case is when an axiom is applied at the root:
\begin{enumerate}
\item
  \flateqRule{Assoc}:
  Let $(\redseq\seq\redseqtwo)\seq\redseqthree \flateq
       \redseq\seq(\redseqtwo\seq\redseqthree)$.
  Then by \permeqRule{Assoc}
  also $(\redseq\seq\redseqtwo)\seq\redseqthree \permeq
       \redseq\seq(\redseqtwo\seq\redseqthree)$.
\item
  \flateqRule{Perm}:
  Let $\mstep \flateq \mstep_1\seq\mstep_2$ where
  $\judgSplit{\mstep}{\mstep_1}{\mstep_2}$.
  Then by \rlem{splitting_sound_wrt_permeq}
  we have that $\mstep \permeq \mstep_1\seq\mstep_2$.
\end{enumerate}
\end{proof}

We thus obtain soundness of flat permutation equivalence as an immediate consequence of ~\rlem{flattening_sound_wrt_permeq} and~\rlem{flateq_sound_wrt_permeq}:

\begin{prop}[Soundness of flat permutation equivalence]
\lprop{soundness_flateq}
Let $\judgRewr{\tenv}{\redseq}{\tm}{\tmtwo}{\typ}$
and $\judgRewr{\tenv}{\redseqtwo}{\tm'}{\tmtwo'}{\typ}$.
Then  $\flatten{\redseq} \flateq \flatten{\redseqtwo}$ implies $\redseq \permeq \redseqtwo$.
\end{prop}
% \begin{proof}
% Immediate
% given that reduction $\tof$ in the flattening system $\flatteningSystem$
% is included in permutation equivalence~(\rlem{flattening_sound_wrt_permeq})
% and that flat permutation equivalence
% is included in permutation equivalence~(\rlem{flateq_sound_wrt_permeq}).
% \end{proof}

\subsection{Completeness}

Completeness, namely that $\redseq \permeq \redseqtwo$
implies $\flatten{\redseq} \flateq \flatten{\redseqtwo}$, and stated
as~\rprop{flateq_complete_wrt_permeq} below, requires considerably more work.
We begin by presenting some auxiliary notions and results.
First a shorthand: 
$\fsrc{\redseq}$ and $\ftgt{\redseq}$
denote the $\tof$-normal forms of the source and target, respectively,
that is, $\flatten{(\refl{\rsrc{\redseq}})}$
and $\flatten{(\refl{\rtgt{\redseq}})}$.
  We begin by introducing certain contexts called {\em composition trees}:
  \[
    \kctx ::= \ctxhole
         \mid \kctx\seq\kctx
  \]
  For each $n \geq 1$, an {\em $n$-hole composition tree}
  is a composition tree with $n$ occurrences of the {\em hole} $\ctxhole$.
  If $\kctx$ is an $n$-hole composition tree,
  we write $\kctxof{\redseq_1,\hdots,\redseq_n}$ to stand for the
  rewrite that results from replacing the $i$-th hole of $\kctx$ for $\redseq_i$
  for each $1 \leq i \leq n$.
%  \smallskip\\
  For example $((\ctxhole\seq\ctxhole)\seq\ctxhole)$ is a 3-hole
  composition tree and
  $((\ctxhole\seq\ctxhole)\seq\ctxhole)\ctxof{\redseq,\redseqtwo,\redseqthree} =
   (\redseq\seq\redseqtwo)\seq\redseqthree$. Composition trees give us ready access to all the multisteps in a flat rewrite.  We next list some useful properties on flat rewrites, presented in terms of compositions trees.

   \begin{description}

     \item[P1] Rules \flatRule{Abs}, \flatRule{App1}, \flatRule{App2}, and \flatRule{App3} of~\rdef{flattening_system} can be generalized to composition trees. \begin{enumerate}
\item {\bf Generalized \flatRule{Abs}.}
  $\lam{\var}{\kctxof{\redseq_1,\hdots,\redseq_n}} \tofs
   \kctxof{\lam{\var}{\redseq_1},\hdots,\lam{\var}{\redseq_n}}$.
\item {\bf Generalized \flatRule{App1}.}
  $\kctxof{\redseq_1,\hdots,\redseq_{n-1},\redseq_n}\,\mstep \tofs
   \kctxof{
     (\redseq_1\,\refl{\rsrc{\mstep}}),
     \hdots,
     (\redseq_{n-1}\,\refl{\rsrc{\mstep}}),
     (\redseq_n\,\mstep)
   }$.
\item {\bf Generalized \flatRule{App2}.}
  $\mstep\,\kctxof{\redseq_1,\redseq_2,\hdots,\redseq_n} \tofs
   \kctxof{
     (\mstep\,\redseq_1),
     (\refl{\rtgt{\mstep}}\,\redseq_2),
     \hdots,
     (\refl{\rtgt{\mstep}}\,\redseq_n)
   }$.
\item {\bf Generalized \flatRule{App3}.}
  If $n, m > 1$ then:
  \[
    \kctxof{\redseq_1,\hdots,\redseq_n}\,
    \kctx'\ctxof{\redseqtwo_1,\hdots,\redseqtwo_m}
    \tofs
    \kctxof{
      (\redseq_1\,\refl{\rsrc{\redseqtwo_1}}),\hdots,
      (\redseq_n\,\refl{\rsrc{\redseqtwo_1}})
    }
    \seq
    \kctx'\ctxof{
      (\refl{\rtgt{\redseq_n}}\,\redseqtwo_1),\hdots,
      (\refl{\rtgt{\redseq_n}}\,\redseqtwo_m)
    }
  \]
\end{enumerate}
%This may be verified by induction on $\kctx$. We'll refer to these as Generalized \flatRule{Abs}, Generalized \flatRule{App1}, Generalized \flatRule{App2}, and Generalized \flatRule{App3}, resp.

       \item[P2] Congruence below abstraction and application also generalize to composition trees, obtaining \emph{generalized congruence equations} in following sense: 
   %\llem{flateq_congruence_below_abstraction}
\[
  \begin{array}{c}
  \kctx_1\ctxof{
    \flatten{(\lam{\var}{\mstep_1})},
    \hdots,
    \flatten{(\lam{\var}{\mstep_n})}
  }
  \flateq
  \kctx_2\ctxof{
    \flatten{(\lam{\var}{\msteptwo_1})},
    \hdots,
    \flatten{(\lam{\var}{\msteptwo_m})}
    }
    \\
        \kctx_1\ctxof{
      \flatten{(\mstep_1\,\refl{\tm})},
      \hdots,
      \flatten{(\mstep_n\,\refl{\tm})}
    }
    \flateq
    \kctx_2\ctxof{
      \flatten{(\msteptwo_1\,\refl{\tm})},
      \hdots,
      \flatten{(\msteptwo_m\,\refl{\tm})}
    }
    \\
        \kctx_1\ctxof{
      \flatten{(\refl{\tm}\,\mstep_1)},
      \hdots,
      \flatten{(\refl{\tm}\,\mstep_n)}
    }
    \flateq
    \kctx_2\ctxof{
      \flatten{(\refl{\tm}\,\msteptwo_1)},
      \hdots,
      \flatten{(\refl{\tm}\,\msteptwo_m)}
    }
    \end{array}
  \]
where 
$\kctx_1\ctxof{\mstep_1,\hdots,\mstep_n} \flateq
 \kctx_2\ctxof{\msteptwo_1,\hdots,\msteptwo_m}$ and let $\tm$ be an arbitrary term.
Note that the multisteps $\mstep_i$ and $\msteptwo_i$ are in
$\tof$-normal form because the $\flateq$ relation only relates
flat rewrites.
But observe that $\lam{\var}{\mstep_i}$ and $\lam{\var}{\msteptwo_i}$
may not necessarily be in $\tof$-normal form because
there may be an \flatRule{EtaM} redex at the root. Thus the proof of the first item  requires some care.

\item[P3] Flattening of an application of two flat rewrites expressed in terms of their composition trees.
%\llem{flattening_of_application_up_to_flateq}
Let $\redseq = \kctx_1\ctxof{\mstep_1,\hdots,\mstep_n}$
and $\redseqtwo = \kctx_2\ctxof{\msteptwo_1,\hdots,\msteptwo_m}$
be flat rewrites.
Then
\[
  \redseq\,\redseqtwo
  \ \tofs\flateq\ %
  \kctx_1\ctxof{
    \flatten{(\mstep_1\,\fsrc{\msteptwo_1})},
    \hdots,
    \flatten{(\mstep_n\,\fsrc{\msteptwo_1})}
  }
  \seq
  \kctx_2\ctxof{
    \flatten{(\ftgt{\mstep_n}\,\msteptwo_1)},
    \hdots,
    \flatten{(\ftgt{\mstep_n}\,\msteptwo_m)}
  }
\]

\item[P4]  Equivalence for term/rewrite substitution of a composition.
%\begin{lem}[Equivalence for term/rewrite substitution of a composition]
%\llem{flatten_equiv_subtr_composition}
Let $\tm$ be a term,
let $\mstep_1,\hdots,\mstep_n$ arbitrary multisteps,
and let $\kctx$ a composition tree.
Then:
\[
  \tm\subtr{\var}{\kctxof{\mstep_1,\hdots,\mstep_n}}
  \,\tofs\flateq\,
  \flatten{\kctxof{\tm\subtr{\var}{\mstep_1},\hdots,\tm\subtr{\var}{\mstep_n}}}
\]
%\end{lem}
\end{description}

\begin{prop}[Completeness of flat permutation equivalence]
\lprop{flateq_complete_wrt_permeq}
Let $\judgRewr{\tenv}{\redseq}{\tm}{\tmtwo}{\typ}$
and $\judgRewr{\tenv}{\redseqtwo}{\tm'}{\tmtwo'}{\typ}$.
Then $\redseq \permeq \redseqtwo$ implies $\flatten{\redseq} \flateq \flatten{\redseqtwo}$.
\end{prop}
\begin{proof}
By induction on the derivation of $\redseq \permeq \redseqtwo$.
In the proof, sometimes we implicitly use the fact that $\tof$ is
strongly normalizing~(\rprop{flat_sn})
and confluent~(\rprop{flat_confluent}).
In particular, note that
$\flatten{(\redseq\seq\redseqtwo)} = \flatten{\redseq}\seq\flatten{\redseqtwo}$
and more in general
$\flatten{\kctxof{\redseq_1,\hdots,\redseq_n}} =
 \kctxof{\flatten{\redseq_1},\hdots,\flatten{\redseq_n}}$.
In the inductive proof, the cases for reflexivity,
symmetry, and transitivity are immediate. We analyze the cases when
a rule is applied at the root, as well as congruence closure under
rewrite constructors. Root cases {\bf \permeqRule{Assoc}} and {\bf \permeqRule{Abs}} are straightforward. Root case {\bf \permeqRule{IdR}} is similar to {\bf \permeqRule{IdL}}, so among the two we only consider the latter. Likewise, the root case {\bf \permeqRule{BetaRT}} is similar to {\bf \permeqRule{BetaTR}}, so among the two we only consider the latter. Congruence case {\bf Congruence under an abstraction} follows from P1 (generalized \flatRule{Abs}), P2 and the \ih. Congruence case {\bf Congruence under a composition} follows from the \ih. 
This leaves the following cases:
\begin{enumerate}
\item {\bf \permeqRule{IdL}.}
  Let $\refl{\rsrc{\redseq}}\seq\redseq \permeq \redseq$.
  Let $\flatten{\redseq} = \kctxof{\mstep_1,\hdots,\mstep_n}$.
  Note that $\rsrc{(\flatten{\redseq})} = \rsrc{\mstep_1}$.
  Then:
  \[
    \begin{array}{rcll}
      \flatten{(\refl{\rsrc{\redseq}}\seq\redseq)}
    & = &
      \flatten{(\refl{\rsrc{\redseq}})} \seq \flatten{\redseq}
    \\
    & = &
      \flatten{(\refl{\rsrc{(\flatten{\redseq})}})} \seq \flatten{\redseq}
      & \text{since % by \rlem{coherence_of_flat_source_and_target}
          $\refl{\rsrc{(\flatten{\redseq})}} \tofs \flatten{(\refl{\rsrc{\redseq}})}$
        }
    \\
    & = &
      \flatten{(\refl{\rsrc{\mstep_1}})} \seq \flatten{\redseq}
      & \text{since $\rsrc{(\flatten{\redseq})} = \rsrc{\mstep_1}$}
    \\
    & = &
      \flatten{(\refl{\rsrc{\mstep_1}})} \seq \kctxof{\mstep_1,\hdots,\mstep_n}
    \\
    & = &
      \kctxof{(\flatten{(\refl{\rsrc{\mstep_1}})}\seq\mstep_1),\hdots,\mstep_n}
      & \text{by \flateqRule{Assoc}} %\rlem{generalized_flateq_assoc}
    \\
    & \flateq &
      \kctxof{\mstep_1,\hdots,\mstep_n}
      & \text{since $\judgSplit{\mstep_1}{\refl{\rsrc{\mstep_1}}}{\mstep_1}$}
%              by \rlem{left_right_splitting}
    \\
    & = &
      \flatten{\redseq}
    \end{array}
  \]
  In the second step above, note that    $\refl{\rsrc{(\flatten{\redseq})}} \tofs \flatten{(\refl{\rsrc{\redseq}})}$ follows from the fact that 
  steps other than \flatRule{BetaM} and \flatRule{EtaM}
preserve the endpoints,
while 
\flatRule{BetaM} and \flatRule{EtaM} reduction steps commute
with taking the endpoints. 
Hence we have that
$\refl{\rsrc{\redseq}} \tofs \refl{\rsrc{(\flatten{\redseq})}}$ and we conclude by confluence of flattening.
\item {\bf \permeqRule{App}.}
  Let $(\redseq_1\,\redseq_2)\seq(\redseqtwo_1\,\redseqtwo_2)
       \permeq
       (\redseq_1\seq\redseqtwo_1)\,(\redseq_2\seq\redseqtwo_2)$.
  Consider the $\tof$-normal forms of each rewrite:
  \[
    \begin{array}{ll}
      \flatten{\redseq_1} = \kctx_1\ctxof{\mstep_1,\hdots,\mstep_n}
    &
      \flatten{\redseq_2} = \kctx_2\ctxof{\msteptwo_1,\hdots,\msteptwo_m}
    \\
      \flatten{\redseqtwo_1} = \kctxB_1\ctxof{\mstepB_1,\hdots,\mstepB_p}
    &
      \flatten{\redseqtwo_2} = \kctxB_2\ctxof{\msteptwoB_1,\hdots,\msteptwoB_q}
    \end{array}
  \]
  Before going on, we make the following claim:
  \[
  \ftgt{\mstep_n} = \fsrc{\mstepB_1}
  \HS\text{ and }\HS
  \ftgt{\msteptwo_m} = \fsrc{\msteptwoB_1}
  \HS\HS\text{$(\star)$}
  \]
  For the first equality,
  note that $\redseq_1$ and $\redseqtwo_1$ are composable,
  so $\rtgt{\redseq_1} \termeq \rsrc{\redseqtwo_2}$
  are $\beta\eta$-equivalent terms.
  Moreover, by \rremark{flat_not_beta_steps_preserve_endpoints}
  and \rlem{flatten_beta_eta_source_target}
  we have that $\rtgt{(\flatten{\redseq_1})} \termeq \rtgt{\redseq_1}$
  and $\rsrc{(\flatten{\redseqtwo_1})} \termeq \rsrc{\redseqtwo_1}$.
  This means that $\rtgt{\mstep_n} \termeq \rsrc{\mstepB_1}$,
  so by confluence and strong normalization of flattening
  $\ftgt{\mstep_n} = \fsrc{\mstepB_1}$.
  Similarly, for the second equality,
  since $\redseq_2$ and $\redseqtwo_2$ are composable,
  we have that $\ftgt{\msteptwo_m} = \fsrc{\msteptwoB_1}$.

  Furthermore, we claim that the two following conditions hold:
  \begin{enumerate}
  \item[{\bf (I)}]
        $\redseq_1\,\redseq_2 \tofs\flateq
          \kctx_1\ctxof{
            \flatten{(\mstep_1\,\fsrc{\msteptwo_1})},
            \hdots,
            \flatten{(\mstep_n\,\fsrc{\msteptwo_1})}
          }
          \seq
          \kctx_2\ctxof{
            \flatten{(\fsrc{\mstepB_1}\,\msteptwo_1)},
            \hdots,
            \flatten{(\fsrc{\mstepB_1}\,\msteptwo_m)}
          }$
  \item[{\bf (II)}]
        $\redseqtwo_1\,\redseqtwo_2 \tofs\flateq
          \kctxB_1\ctxof{
            \flatten{(\mstepB_1\,\ftgt{\msteptwo_m})},
            \hdots,
            \flatten{(\mstepB_p\,\ftgt{\msteptwo_m})}
          }
          \seq
          \kctxB_2\ctxof{
            \flatten{(\ftgt{\mstepB_p}\,\msteptwoB_1)},
            \hdots,
            \flatten{(\ftgt{\mstepB_p}\,\msteptwoB_q)}
          }$
  \end{enumerate}
  To prove {\bf (I)}, note that:
  \[
    \begin{array}{rlll}
    &&
      \redseq_1\,\redseq_2
    \\
    & \tofs\flateq &
      \kctx_1\ctxof{
        \flatten{(\mstep_1\,\fsrc{\msteptwo_1})},
        \hdots,
        \flatten{(\mstep_n\,\fsrc{\msteptwo_1})}
      }
      \seq
      \kctx_2\ctxof{
        \flatten{(\ftgt{\mstep_n}\,\msteptwo_1)},
        \hdots,
        \flatten{(\ftgt{\mstep_n}\,\msteptwo_m)}
      }
      & \text{by P3 above}% of two flat rewrites expressed in terms of their composition trees.\rlem{flattening_of_application_up_to_flateq}
    \\
    & = &
      \kctx_1\ctxof{
        \flatten{(\mstep_1\,\fsrc{\msteptwo_1})},
        \hdots,
        \flatten{(\mstep_n\,\fsrc{\msteptwo_1})}
      }
      \seq
      \kctx_2\ctxof{
        \flatten{(\fsrc{\mstepB_1}\,\msteptwo_1)},
        \hdots,
        \flatten{(\fsrc{\mstepB_1}\,\msteptwo_m)}
      }
      & \text{by the claim $(\star)$ above}
    \end{array}
  \]
  The proof of {\bf (II)} is symmetric to the proof of {\bf (I)}.

  To conclude the proof of the \flateqRule{App} case,
  let us rewrite the left-hand side.
  We use the associativity rule (\flateqRule{Assoc}) implicitly:
    \[
      \begin{array}{rlll}
      &&
        (\redseq_1\,\redseq_2)\seq(\redseqtwo_1\,\redseqtwo_2)
      \\
      & \tofs\flateq &
        \kctx_1\ctxof{
          \flatten{(\mstep_1\,\fsrc{\msteptwo_1})},
          \hdots,
          \flatten{(\mstep_n\,\fsrc{\msteptwo_1})}
        }
        \seq
        \kctx_2\ctxof{
          \flatten{(\fsrc{\mstepB_1}\,\msteptwo_1)},
          \hdots,
          \flatten{(\fsrc{\mstepB_1}\,\msteptwo_m)}
        }
        \seq \\
        &&
        \kctxB_1\ctxof{
          \flatten{(\mstepB_1\,\ftgt{\msteptwo_m})},
          \hdots,
          \flatten{(\mstepB_p\,\ftgt{\msteptwo_m})}
        }
        \seq
        \kctxB_2\ctxof{
          \flatten{(\ftgt{\mstepB_p}\,\msteptwoB_1)},
          \hdots,
          \flatten{(\ftgt{\mstepB_p}\,\msteptwoB_q)}
        }
        \\&&\HS\text{by claims {\bf (I)} and {\bf (II)}}
      % (Missing associativity step)
      \\
      & \flateq &
         \kctx_1\ctxof{
           \flatten{(\mstep_1\,\fsrc{\msteptwo_1})},
           \hdots,
           \flatten{(\mstep_n\,\fsrc{\msteptwo_1})}
         }
         \seq
         \kctxB_1\ctxof{
           \flatten{(\mstepB_1\,\fsrc{\msteptwo_1})},
           \hdots,
           \flatten{(\mstepB_p\,\fsrc{\msteptwo_1})}
         }
         \seq
         \\&&
         \kctx_2\ctxof{
           \flatten{(\ftgt{\mstepB_p}\,\msteptwo_1)},
           \hdots,
           \flatten{(\ftgt{\mstepB_p}\,\msteptwo_m)}
         }
         \seq
         \kctxB_2\ctxof{
           \flatten{(\ftgt{\mstepB_p}\,\msteptwoB_1)},
           \hdots,
           \flatten{(\ftgt{\mstepB_p}\,\msteptwoB_q)}
         }
         \\&&\HS\text{($\star$)} %by \rlem{flateq_generalized_swap}
      \end{array}
    \]
    The $(\star)$ above refers to the following swapping property which can proved by induction on $\kctx_1$:
    \[
  \begin{array}{rcll}
  &
    \kctx_1\ctxof{
      \flatten{(\mstep_1\,\refl{\rsrc{\msteptwo_1}})},
      \hdots,
      \flatten{(\mstep_n\,\refl{\rsrc{\msteptwo_1}})}
    }
    \seq
    \kctx_2\ctxof{
      \flatten{(\refl{\rtgt{\mstep_n}}\,\msteptwo_1)},
      \hdots,
      \flatten{(\refl{\rtgt{\mstep_n}}\,\msteptwo_m)}
    }
  \\
  \flateq &
    \kctx_2\ctxof{
      \flatten{(\refl{\rsrc{\mstep_1}}\,\msteptwo_1)},
      \hdots,
      \flatten{(\refl{\rsrc{\mstep_1}}\,\msteptwo_m)}
    }
    \seq
    \kctx_1\ctxof{
      \flatten{(\mstep_1\,\refl{\rtgt{\msteptwo_m}})},
      \hdots,
      \flatten{(\mstep_n\,\refl{\rtgt{\msteptwo_m}})}
    }
  \end{array}
\]
  where $\kctx_1,\kctx_2$ are arbitrary composition trees 
and 
$\mstep_1,\hdots,\mstep_n,\msteptwo_1,\hdots,\msteptwo_m$ are arbitrary multisteps.

    On the other hand, rewriting the right-hand side:
    \[
      \begin{array}{rcll}
      &&
        (\redseq_1\seq\redseqtwo_1)\,(\redseq_2\seq\redseqtwo_2)
      \\
      & \tofs &
        (
          \kctx_1\ctxof{\mstep_1,\hdots,\mstep_n}
          \seq
          \kctxB_1\ctxof{\mstepB_1,\hdots,\mstepB_p}
        )
        \,
        (
          \kctx_2\ctxof{\msteptwo_1,\hdots,\msteptwo_m}
          \seq
          \kctxB_2\ctxof{\msteptwoB_1,\hdots,\msteptwoB_q}
        )
      \\
        & \tofs &
        \kctx_1\ctxof{
          (\mstep_1\,\refl{\rsrc{\msteptwo_1}}),
          \hdots,
          (\mstep_n\,\refl{\rsrc{\msteptwo_1}})
        }
        \seq
        \kctxB_1\ctxof{
          (\mstepB_1\,\refl{\rsrc{\msteptwo_1}}),
          \hdots,
          (\mstepB_p\,\refl{\rsrc{\msteptwo_1}})
        }
        \seq
        \\&&
        \kctx_2\ctxof{
          (\refl{\rtgt{\mstepB_p}}\,\msteptwo_1),
          \hdots,
          (\refl{\rtgt{\mstepB_p}}\,\msteptwo_m)
        }
        \seq
        \kctxB_2\ctxof{
          (\refl{\rtgt{\mstepB_p}}\,\msteptwoB_1),
          \hdots,
          (\refl{\rtgt{\mstepB_p}}\,\msteptwoB_q)
        }
        \\&&\HS\text{by P1 (generalized
                     \flatRule{App3})} %~(\rlem{generalized_flattening_kctx})
      \\
      & \tofs &
         \kctx_1\ctxof{
           \flatten{(\mstep_1\,\fsrc{\msteptwo_1})},
           \hdots,
           \flatten{(\mstep_n\,\fsrc{\msteptwo_1})}
         }
         \seq
         \kctxB_1\ctxof{
           \flatten{(\mstepB_1\,\fsrc{\msteptwo_1})},
           \hdots,
           \flatten{(\mstepB_p\,\fsrc{\msteptwo_1})}
         }
         \seq
         \\&&
         \kctx_2\ctxof{
           \flatten{(\ftgt{\mstepB_p}\,\msteptwo_1)},
           \hdots,
           \flatten{(\ftgt{\mstepB_p}\,\msteptwo_m)}
         }
         \seq
         \kctxB_2\ctxof{
           \flatten{(\ftgt{\mstepB_p}\,\msteptwoB_1)},
           \hdots,
           \flatten{(\ftgt{\mstepB_p}\,\msteptwoB_q)}
         }
      \end{array}
    \]

\item {\bf \permeqRule{BetaTR}.}
  Let $(\lam{\var}{\refl{\tm}})\,\redseq \permeq \tm\subtr{\var}{\redseq}$,
  and suppose that $\flatten{\redseq} = \kctxof{\mstep_1,\hdots,\mstep_n}$.
  First note that,
  $(\lam{\var}{\refl{\tm}})\,\redseq
   \tofs
   \kctxof{
     ((\lam{\var}{\refl{\tm}})\,\mstep_1),
     \hdots,
     ((\lam{\var}{\refl{\tm}})\,\mstep_n)
   }$.
  Indeed, if $n = 1$ this is immediate, and if $n > 1$
  this is a consequence of P1 (generalized \flatRule{App2}). %~(\rlem{generalized_flattening_kctx}).
  Hence:
  \[
    \begin{array}{rlll}
      (\lam{\var}{\refl{\tm}})\,\redseq
    & \tofs &
      \kctxof{
        ((\lam{\var}{\refl{\tm}})\,\mstep_1),
        \hdots,
        ((\lam{\var}{\refl{\tm}})\,\mstep_n)
      }
    \\
    & \tofs &
      \kctxof{
        \tm\subm{\var}{\mstep_1},
        \hdots,
        \tm\subm{\var}{\mstep_n}
      }
      & \text{by \flatRule{BetaM} ($n$ times)}
    \\
    & = &
      \kctxof{
        \tm\subtr{\var}{\mstep_1},
        \hdots,
        \tm\subtr{\var}{\mstep_n}
      }
      & \text{by \rremark{multistep_substitution_vs_subtr_subrt}(2)}
    \\
    & \tofs &
      \flatten{
        \kctxof{\tm\subtr{\var}{\mstep_1},\hdots,\tm\subtr{\var}{\mstep_n}}
      }
    \\
    & \flateq\tofsinv &
      \tm\subtr{\var}{\kctxof{\mstep_1,\hdots,\mstep_n}}
      & \text{by P4}%\rlem{flatten_equiv_subtr_composition}
    \\
    & \tofsinv &
      \tm\subtr{\var}{\redseq}
      & \text{by \rlem{flattening_below_subs}(1)}
    \end{array}
  \]
% \item {\bf \permeqRule{BetaRT}.}
%   Let $(\lam{\var}{\redseq})\,\refl{\tm} \permeq \redseq\subt{\var}{\tm}$,
%   and suppose that $\flatten{\redseq} = \kctxof{\mstep_1,\hdots,\mstep_n}$.
%   First note that
%   $(\lam{\var}{\redseq})\,\refl{\tm} \tofs
%    \kctxof{(\lam{\var}{\mstep_1})\,\refl{\tm},\hdots,
%            (\lam{\var}{\mstep_n})\,\refl{\tm}}$.
%   Indeed, if $n = 1$ this is immediate, and if $n > 1$ this is a
%   consequence of P1 (generalized \flatRule{Abs}
%   and \flatRule{App1} rules). %~(\rlem{generalized_flattening_kctx}).
%   Hence:
%   \[
%     \begin{array}{rlll}
%       (\lam{\var}{\redseq})\,\refl{\tm}
%     & \tofs &
%       \kctxof{(\lam{\var}{\mstep_1})\,\refl{\tm},\hdots,
%               (\lam{\var}{\mstep_n})\,\refl{\tm}}
%     \\
%     & \tofs &
%       \kctxof{
%         \mstep_1\subm{\var}{\refl{\tm}},
%         \hdots,
%         \mstep_n\subm{\var}{\refl{\tm}}
%       }
%       & \text{by \flatRule{BetaM} ($n$ times)}
%     \\
%     & = &
%       \kctxof{
%         \mstep_1\subt{\var}{\tm},
%         \hdots,
%         \mstep_n\subt{\var}{\tm}
%       }
%       & \text{by \rremark{multistep_substitution_vs_subtr_subrt}(2)}
%     \\
%     & = &
%       \kctxof{\mstep_1,\hdots,\mstep_n}\subt{\var}{\tm}
%     \\
%     & \tofsinv &
%       \redseq\subt{\var}{\tm}
%       & \text{by \rlem{flattening_below_subs}(2)}
%     \end{array}
%   \]
\item {\bf \permeqRule{Eta}.}
  Let $\lam{\var}{\redseq\,\var} \permeq \redseq$
  where $\var\notin\fv{\redseq}$.
  Let $\flatten{\redseq} = \kctxof{\mstep_1,\hdots,\mstep_n}$.
  It suffices to note that
  $\flatten{(\lam{\var}{\redseq\,\var})} = \flatten{\redseq}$. Indeed:
  \[
    \begin{array}{rlll}
      \lam{\var}{\redseq\,\var}
    & \tofs &
      \lam{\var}{\flatten{\redseq}\,\var}
    \\
    & = &
      \lam{\var}{\kctxof{\mstep_1,\hdots,\mstep_n}\,\var}
    \\
    & \tofs &
      \lam{\var}{
        \kctxof{(\mstep_1\,\var),\hdots,(\mstep_n\,\var)}
      }
      & \text{by P1 (generalized
              \flatRule{App1})}%~(\rlem{generalized_flattening_kctx})
    \\
    & \tofs &
      \kctxof{\lam{\var}{(\mstep_1\,\var)},\hdots,\lam{\var}{(\mstep_n\,\var)}}
      & \text{by P1 (generalized
              \flatRule{Abs})}%~(\rlem{generalized_flattening_kctx})
    \\
    & \tofs &
      \kctxof{\mstep_1,\hdots,\mstep_n}
      & \text{by \flatRule{EtaM} ($n$ times)}
    \\
    & = &
      \flatten{\redseq}
    \end{array}
  \]
  Note that we may apply the \flatRule{EtaM} rule
  because, for each $1 \leq i \leq n$,
  we have that $\var \notin \fv{\mstep_i}$.
  This in turn is justified by noting that
  $\var \notin \fv{\kctxof{\mstep_1,\hdots,\mstep_n}}
             = \fv{\flatten{\redseq}}$,
  which is a consequence of the fact that
  flattening does not create free variables.
% \item {\bf Congruence under an abstraction.}
%   Let $\lam{\var}{\redseq} \permeq \lam{\var}{\redseqtwo}$
%   be derived from $\redseq \permeq \redseqtwo$.
%   Consider their $\tof$-normal forms,
%   $\flatten{\redseq} = \kctx_1\ctxof{\mstep_1,\hdots,\mstep_n}$
%   and $\flatten{\redseqtwo} = \kctx_2\ctxof{\msteptwo_1,\hdots,\msteptwo_m}$.
%   By \ih we have that $\flatten{\redseq} \flateq \flatten{\redseqtwo}$.
%   Then:
%   \[
%     \begin{array}{rlll}
%       \lam{\var}{\redseq}
%     & \tofs &
%       \lam{\var}{\kctx_1\ctxof{\mstep_1,\hdots,\mstep_n}}
%     \\
%     & \tofs & 
%       \kctx_1\ctxof{\lam{\var}{\mstep_1},\hdots,\lam{\var}{\mstep_n}}
%        & \text{by P1 (generalized \flatRule{Abs})}% (\rlem{generalized_flattening_kctx})
%     \\
%     & \tofs & 
%       \kctx_1\ctxof{
%         \flatten{(\lam{\var}{\mstep_1})},
%         \hdots,
%         \flatten{(\lam{\var}{\mstep_n})}
%       }
%     \\
%     & \flateq &
%       \kctx_2\ctxof{
%         \flatten{(\lam{\var}{\msteptwo_1})},
%         \hdots,
%         \flatten{(\lam{\var}{\msteptwo_m})}
%       }
%       & \text{by P2 (%\rlem{flateq_congruence_below_abstraction},
%               as $\flatten{\redseq} \flateq \flatten{\redseqtwo}$)}
%     \\
%     & \tofsinv &
%       \kctx_2\ctxof{\lam{\var}{\msteptwo_1},\hdots,\lam{\var}{\msteptwo_m}}
%     \\
%     & \tofsinv &
%       \lam{\var}{\kctx_2\ctxof{\msteptwo_1,\hdots,\msteptwo_m}}
%       & \text{by P1 (generalized \flatRule{Abs})}% (\rlem{generalized_flattening_kctx})
%     \\
%     & \tofsinv &
%       \lam{\var}{\redseqtwo}
%     \end{array}
%   \]
\item {\bf Congruence under an application.}
  Let $\redseq_1\,\redseq_2 \permeq \redseqtwo_1\,\redseqtwo_2$
  be derived from $\redseq_1 \permeq \redseqtwo_1$
  and $\redseq_2 \permeq \redseqtwo_2$.
  Consider the $\tof$-normal forms of each rewrite:
  \[
    \begin{array}{ll}
      \flatten{\redseq_1} = \kctx_1\ctxof{\mstep_1,\hdots,\mstep_n}
    &
      \flatten{\redseq_2} = \kctx_2\ctxof{\msteptwo_1,\hdots,\msteptwo_m}
    \\
      \flatten{\redseqtwo_1} = \kctxB_1\ctxof{\mstepB_1,\hdots,\mstepB_p}
    &
      \flatten{\redseqtwo_2} = \kctxB_2\ctxof{\msteptwoB_1,\hdots,\msteptwoB_q}
    \end{array}
  \]
  By \ih we have that
  $\flatten{\redseq_1} \flateq \flatten{\redseqtwo_1}$
  and
  $\flatten{\redseq_2} \flateq \flatten{\redseqtwo_2}$.
  Before going on, we make the following claim:
  \[
    \fsrc{\msteptwo_1} = \fsrc{\msteptwoB_1}
    \HS\text{ and }\HS
    \fsrc{\mstep_n} = \fsrc{\mstepB_p}
    \HS\HS(\star)
  \]
  For the first equality, note that
  $\redseq_2 \permeq \redseqtwo_2$,
  so $\rsrc{\redseq_2} \termeq \rsrc{\redseqtwo_2}$
  are $\beta\eta$-equivalent terms by \rlem{permeq_endpoints_are_termeq}.
  Moreover, by \rremark{flat_not_beta_steps_preserve_endpoints}
  and \rlem{flatten_beta_eta_source_target}
  we have that $\rsrc{(\flatten{\redseq_2})} \termeq \rsrc{\redseq_2}$
  and $\rsrc{(\flatten{\redseqtwo_2})} \termeq \rsrc{\redseqtwo_2}$.
  This means that $\rsrc{\msteptwo_1} \termeq \rsrc{\msteptwoB_1}$,
  so by confluence and strong normalization of flattening
  $\fsrc{\msteptwo_1} = \fsrc{\msteptwoB_1}$.
  Similarly, for the second equality,
  since $\redseq_1 \permeq \redseqtwo_1$,
  we have that $\ftgt{\mstep_n} = \ftgt{\mstepB_p}$.
  Then:
  \[
    \begin{array}{rlll}
    &&
      \redseq_1\,\redseq_2
    \\
    & \tofs\flateq &
      \kctx_1\ctxof{
        \flatten{(\mstep_1\,\fsrc{\msteptwo_1})},
        \hdots,
        \flatten{(\mstep_n\,\fsrc{\msteptwo_1})}
      }
      \seq
      \kctx_2\ctxof{
        \flatten{(\ftgt{\mstep_n}\,\msteptwo_1)},
        \hdots,
        \flatten{(\ftgt{\mstep_n}\,\msteptwo_m)}
      }
      & \text{by P3}%\rlem{flattening_of_application_up_to_flateq}
    \\
    & \flateq &
      \kctxB_1\ctxof{
        \flatten{(\mstepB_1\,\fsrc{\msteptwo_1})},
        \hdots,
        \flatten{(\mstepB_p\,\fsrc{\msteptwo_1})}
      }
      \seq
      \kctx_2\ctxof{
        \flatten{(\ftgt{\mstep_n}\,\msteptwo_1)},
        \hdots,
        \flatten{(\ftgt{\mstep_n}\,\msteptwo_m)}
      }
      & \text{by P2 (%\rlem{flateq_congruence_below_application}
              as $\flatten{\redseq_1} \flateq \flatten{\redseqtwo_1}$)}
    \\
    & \flateq &
      \kctxB_1\ctxof{
        \flatten{(\mstepB_1\,\fsrc{\msteptwo_1})},
        \hdots,
        \flatten{(\mstepB_p\,\fsrc{\msteptwo_1})}
      }
      \seq
      \kctxB_2\ctxof{
        \flatten{(\ftgt{\mstep_n}\,\msteptwoB_1)},
        \hdots,
        \flatten{(\ftgt{\mstep_n}\,\msteptwoB_q)}
      }
      & \text{by P2 (%\rlem{flateq_congruence_below_application}
              as $\flatten{\redseq_2} \flateq \flatten{\redseqtwo_2}$)}
    \\
    & = &
      \kctxB_1\ctxof{
        \flatten{(\mstepB_1\,\fsrc{\msteptwoB_1})},
        \hdots,
        \flatten{(\mstepB_p\,\fsrc{\msteptwoB_1})}
      }
      \seq
      \kctxB_2\ctxof{
        \flatten{(\ftgt{\mstepB_p}\,\msteptwoB_1)},
        \hdots,
        \flatten{(\ftgt{\mstepB_p}\,\msteptwoB_q)}
      }
      & \text{by the claim $(\star)$ above}
    \\
    & \tofsinv &
      \redseqtwo_1\,\redseqtwo_2
      & \text{by P3}% \rlem{flattening_of_application_up_to_flateq}
    \end{array}
  \]
% \item {\bf Congruence under a composition.}
%   Let $\redseq_1\seq\redseq_2 \permeq \redseqtwo_1\seq\redseqtwo_2$
%   be derived from
%   $\redseq_1 \permeq \redseqtwo_1$ and $\redseq_2 \permeq \redseqtwo_2$.
%   Then:
%   \[
%     \begin{array}{rlll}
%       \flatten{(\redseq_1\seq\redseq_2)}
%     & = &
%       \flatten{\redseq_1}\seq\flatten{\redseq_2}
%     \\
%     & \flateq &
%       \flatten{\redseqtwo_1}\seq\flatten{\redseqtwo_2}
%       & \text{by \ih}
%     \\
%     & = &
%       \flatten{(\redseqtwo_1\seq\redseqtwo_2)}
%     \end{array}
%   \]
\end{enumerate}
\end{proof}

\begin{thm}[Soundness and completeness of flat permutation equivalence]
\lthm{body:soundness_completeness_flateq}
Let $\judgRewr{\tenv}{\redseq}{\tm}{\tmtwo}{\typ}$
and $\judgRewr{\tenv}{\redseqtwo}{\tm'}{\tmtwo'}{\typ}$.
Then $\redseq \permeq \redseqtwo$
if and only if $\flatten{\redseq} \flateq \flatten{\redseqtwo}$.
\end{thm}
% \begin{proof}
% The ($\Leftarrow$) direction \edu{has been developed above}\eduso{  is immediate,
% given that reduction $\tof$ in the flattening system $\flatteningSystem$
% is included in permutation equivalence~($\redseq \tof \redseqtwo$ implies $\redseq \permeq \redseqtwo$)
% and, similarly, flat permutation equivalence
% is included in permutation equivalence~($\redseq \flateq \redseqtwo$ implies $\redseq \permeq \redseqtwo$). }
% \smallskip\\
% The ($\Rightarrow$) direction is by induction on the derivation of $\redseq
% \permeq \redseqtwo$. It is subtle and requires numerous auxiliary results
% % (\SeeAppendix{see~\rsec{appendix:flattening_completeness} in~\cite{techReport}}).
% (\SeeAppendix{see Section D.8 in~\cite{techReport}}).
% \end{proof}

%%% Local Variables:
%%% mode: latex
%%% TeX-master: "main"
%%% End:

\section{Projection}
\lsec{projection}

This section presents projection equivalence.
Two rewrites $\redseq$ and $\redseqtwo$ are said to be projection equivalent if the
steps performed by $\redseq$ are included in those performed by $\redseqtwo$
and vice-versa. We proceed in stages as follows.
First, we define {\em projection} of
multisteps over multisteps~(\rdef{projection_operator})
and prove some of its properties~(\rprop{body:properties_of_projection}).
Second, we extend
projection to flat rewrites~(\rdef{projection_for_rewrites}).
Third, we extend
projection to arbitrary rewrites~(\rdef{projection_for_arbitrary_rewrites})
and, again, we prove some of its properties~(\rprop{body:properties_of_arbitrary_projection}).
Finally, we show that the induced notion of {\em projection equivalence} 
turns out to coincide with
permutation equivalence~(\rthm{body:projection_equivalence}).

\subsection{Projection for Multisteps.}
Consider the rewrites $\consof{mu}\,\rulewittwo$ and $\rulewit\,\consof{f}$, using the notation of~\rexample{mu_permeq}, each representing one step. Since rewrites are subject to $\beta\eta$-equivalence, to define projection one must ``line up'' rule symbols with the left-hand side of the rewrite rules they witness\footnote{See also the discussion on p.~120 of~\cite{thesis:bruggink:08}.}. For example, if the above two multisteps were rewritten as $(\lam{\vartwo}{\consof{mu}\,(\lam{\var}{\vartwo\,\var})})\,\rulewittwo$ and $\rulewit\,(\lam{\var}{\consof{f}\,\var})$, respectively, then one can reason inductively as follows to compute the projection of the former over the latter (the inference rules themselves are introduced in~\rdef{projection_judgement}):
    \[
      \indrule{\small{ProjApp}}{\indrule{\small{ProjRuleR}}
        {\emptyPremise}
        {\judgProj{\lam{\vartwo}{\consof{mu}\,(\lam{\var}{\vartwo\,\var})}}{\rulewit}{\lam{\vartwo}{\vartwo\,(\consof{mu}\,(\lam{\var}{\vartwo\,\var}))}}}
        \indrule{\small{ProjRuleL}}
        {\emptyPremise}
        {\judgProj{\rulewittwo}{\lam{\var}{\consof{f}\,\var}}{\rulewittwo}}
      }
      {\judgProj{ (\lam{\vartwo}{\consof{mu}\,(\lam{\var}{\vartwo\,\var})})\,\rulewittwo}{\rulewit\,(\lam{\var}{\consof{f}\,\var})}{(\lam{\vartwo}{\vartwo\,(\consof{mu}\,(\lam{\var}{\vartwo\,\var}))})\,\rulewittwo}}
    \]
The flat normal form of $(\lam{\vartwo}{\vartwo\,(\consof{mu}\,(\lam{\var}{\vartwo\,\var}))})\,\rulewittwo$
is the rewrite 
$\rulewittwo\,(\consof{mu}\,\rulewittwo)$. Hence we would deduce
$\judgProj{\consof{mu}\,\rulewittwo}{\rulewit\,\consof{f}}{\rulewittwo\,(\consof{mu}\,\rulewittwo)}$.
We begin by introducing an auxiliary notion of projection on coinitial multisteps that may not be flat (\ie~may not be in $\flatRule{BetaM},\flatRule{EtaM}$-normal form) called \emph{weak projection}. We then make use of this notion, to define projection for flat multisteps (\rdef{projection_operator}).

\begin{defi}[Weak projection and compatibility]
  \ldef{projection_judgement}
Let $\judgRewr{\tenv}{\mstep}{\tm}{\tmtwo}{\typ}$
and $\judgRewr{\tenv}{\msteptwo}{\tm'}{\tmthree}{\typ}$
be multisteps, not necessarily in normal form,
such that $\tm \termeq \tm'$.
The judgment $\judgProj{\mstep}{\msteptwo}{\mstepthree}$
is defined as follows:
\[
{\small
  \begin{array}{c}
  \indrule{ProjVar}{
  }{
    \judgProj{\var}{\var}{\var}
  }
  \indrule{ProjCon}{
  }{
    \judgProj{\cons}{\cons}{\cons}
    }
    \\
    \\
  \indrule{ProjRule}{
  }{
    \judgProj{\rulewit}{\rulewit}{\rtgt{\rulewit}}
  }
  \indrule{ProjRuleL}{
  }{
    \judgProj{\rulewit}{\rsrc{\rulewit}}{\rulewit}
  }
  \indrule{ProjRuleR}{
  }{
    \judgProj{\rsrc{\rulewit}}{\rulewit}{\rtgt{\rulewit}}
    }
    \\
    \\
  \indrule{ProjAbs}{
    \judgProj{\mstep}{\msteptwo}{\mstepthree}
  }{
    \judgProj{
       \lam{\var}{\mstep}
     }{
       \lam{\var}{\msteptwo}
     }{
       \lam{\var}{\mstepthree}
     }
  }
  \indrule{ProjApp}{
    \judgProj{\mstep_1}{\msteptwo_1}{\mstepthree_1}
    \HS
    \judgProj{\mstep_2}{\msteptwo_2}{\mstepthree_2}
  }{
    \judgProj{
       \mstep_1\,\mstep_2
     }{
       \msteptwo_1\,\msteptwo_2
     }{
       \mstepthree_1\,\mstepthree_2
     }
  }
  \end{array}
}
\]

We say that $\mstep$ and $\msteptwo$ are {\em compatible},
written $\judgCompat{\mstep}{\msteptwo}$
if, intuitively speaking,
$\mstep$ and $\msteptwo$ are coinitial,
and are ``almost'' $\eta$-expanded and $\beta$-normal forms,
with the exception that
the head of the term may be the source of a rule,
\ie a term of the form $\rsrc{\rulewit}$.
Compatibility is defined as follows:
\[
  \begin{array}{c}
  \indrule{CVar}{
    (\judgCompat{\mstep_i}{\msteptwo_i})_{i=1}^{m}
  }{
    \judgCompat{
      \lam{\varseq{\var}}{\vartwo\,\varseq{\mstep}}
    }{
      \lam{\varseq{\var}}{\vartwo\,\varseq{\msteptwo}}
    }
  }
  \indrule{CCon}{
    (\judgCompat{\mstep_i}{\msteptwo_i})_{i=1}^{m}
  }{
    \judgCompat{
      \lam{\varseq{\var}}{\cons\,\varseq{\mstep}}
    }{
      \lam{\varseq{\var}}{\cons\,\varseq{\msteptwo}}
    }
  }
  \indrule{CRule}{
    (\judgCompat{\mstep_i}{\msteptwo_i})_{i=1}^{m}
  }{
    \judgCompat{
      \lam{\varseq{\var}}{\rulewit\,\varseq{\mstep}}
    }{
      \lam{\varseq{\var}}{\rulewit\,\varseq{\msteptwo}}
    }
  }
  \\
  \\
  \indrule{CRuleL}{
    (\judgCompat{\mstep_i}{\msteptwo_i})_{i=1}^{m}
  }{
    \judgCompat{
      \lam{\varseq{\var}}{\rulewit\,\varseq{\mstep}}
    }{
      \lam{\varseq{\var}}{\rsrc{\rulewit}\,\varseq{\msteptwo}}
    }
  }
  \indrule{CRuleR}{
    (\judgCompat{\mstep_i}{\msteptwo_i})_{i=1}^{m}
  }{
    \judgCompat{
      \lam{\varseq{\var}}{\rsrc{\rulewit}\,\varseq{\mstep}}
    }{
      \lam{\varseq{\var}}{\rulewit\,\varseq{\msteptwo}}
    }
  }
  \end{array}
\]
\end{defi}
The interesting cases are the two last rules, which state essentially
that a rule symbol is compatible with its source term.
Clearly if $\judgCompat{\mstep}{\msteptwo}$, then there exists a unique $\mstepthree$ such that $\judgProj{\mstep}{\msteptwo}{\mstepthree}$.
Moreover, weak projection is coherent with respect to flattening:
\begin{lem}[Coherence of projection]
\llem{coherence_of_projection}
Let $\mstep_1,\msteptwo_1,\mstep_2,\msteptwo_2$ be multisteps
such that
  $\judgCompat{\mstep_1}{\msteptwo_1}$,
  $\judgCompat{\mstep_2}{\msteptwo_2}$,
  $\flatten{\mstep_1} = \flatten{\mstep_2}$,
  $\flatten{\msteptwo_1} = \flatten{\msteptwo_2}$,
  $\judgProj{\mstep_1}{\msteptwo_1}{\mstepthree_1}$,
  and
  $\judgProj{\mstep_2}{\msteptwo_2}{\mstepthree_2}$.
Then $\flatten{\mstepthree_1} = \flatten{\mstepthree_2}$.
\end{lem}
\begin{proof}
The proof proceeds by induction on the
derivation of $\judgCompat{\mstep_1}{\msteptwo_1}$. Note that the last rule used
to derive $\judgCompat{\mstep_2}{\msteptwo_2}$
must be the same as the last rule used
to derive $\judgCompat{\mstep_1}{\msteptwo_1}$.
  The core of the argument in these cases is to apply the \ih
using the following \emph{matching property}:
If $\flatten{(\rsrc{\rulewit}\,\mstep_1\hdots\mstep_n)} =
    \flatten{(\rsrc{\rulewit}\,\msteptwo_1\hdots\msteptwo_n)}$
then $\flatten{\mstep_i} = \flatten{\msteptwo_i}$
for all $1 \leq i \leq n$. This property is well-known and a direct consequence of the assumption that LHS of rewrite rule are rule-patterns.
\qedhere

\end{proof}

\begin{lem}[Compatibilization of coinitial multisteps]
\llem{compatibilization_of_coinitial_multisteps}
Let $\mstep,\msteptwo$ be multisteps
such that $\fsrc{\mstep} = \fsrc{\msteptwo}$.
Then there exist multisteps $\mstepB,\msteptwoB$
such that $\judgCompat{\mstepB}{\msteptwoB}$
and moreover $\flatten{\mstepB} = \flatten{\mstep}$
and $\flatten{\msteptwoB} = \flatten{\msteptwo}$.
\end{lem}
\begin{proof}
We begin with a few observations.
Recall that if $\mstep \tofs \mstep'$
then $\rsrc{\mstep} \termeq \rsrc{\mstep'}$,
which is a consequence of
\rremark{flat_not_beta_steps_preserve_endpoints}
and \rlem{flatten_beta_eta_source_target}.
This means that, without loss of generality, we may assume that
$\mstep,\msteptwo$ are in $\tofnoeta,\etalong$-normal form.
Note that, given this assumption,
we have that $\altsrc{\mstep}$ and $\altsrc{\mstep}$ are
in $\etalong$-normal form.
Moreover,
$\flatten{(\altsrc{\mstep})}
 = \flatten{(\rsrc{\mstep})}
 = \flatten{(\rsrc{\msteptwo})}
 = \flatten{(\altsrc{\msteptwo})}$.
From this,
%by \rlem{flattening_of_eta_expanded_multisteps},
we have that
$\flattennoeta{(\altsrc{\mstep})}
 = \flattennoeta{(\altsrc{\msteptwo})}$.

The proof proceeds by induction
on the sum of the number of applications in $\mstep$ and $\msteptwo$,
using a straightforward characterization of flat multisteps. %~(\rlem{normal_multisteps}).
We consider three cases, depending on whether the head of $\mstep$
is a variable, a constant, or a rule symbol:
\begin{enumerate}
\item
  \label{compatibilization_of_coinitial_multisteps__case_var}
  {\bf Variable,
       $\mstep = \lam{\var_1\hdots\var_n}{\vartwo\,\mstep_1\hdots\mstep_m}$.}
  Note that
  $\flattennoeta{(\altsrc{\mstep})} =
   \lam{\var_1\hdots\var_n}{
     \vartwo\,\flattennoeta{(\altsrc{\mstep_1})}
              \hdots
              \flattennoeta{(\altsrc{\mstep_m})}
   }$.
  Since $\flattennoeta{(\altsrc{\mstep})} =
         \flattennoeta{(\altsrc{\msteptwo})}$,
  the head of $\msteptwo$ cannot be a constant or a rule symbol.
  Hence the only possibility is that
  $\msteptwo = \lam{\var_1\hdots\var_p}{\vartwo\,\msteptwo_1\hdots\msteptwo_q}$.
  Furthermore, it can only be the case that $p = n$ and $q = m$
  and $\flattennoeta{(\altsrc{\mstep_i})} =
       \flattennoeta{(\altsrc{\msteptwo_i})}$
  for all $1 \leq i \leq m$.
  This in turn implies that
  $\fsrc{\mstep_i} = \fsrc{\msteptwo_i}$ for all $1 \leq i \leq m$.
  Hence, by \ih, for each $1 \leq i \leq m$,
  there exist multisteps such that
  $\judgCompat{\mstepB_i}{\msteptwoB_i}$
  where $\flatten{\mstepB_i} = \flatten{\mstep_i}$
  and $\flatten{\msteptwoB_i} = \flatten{\msteptwo_i}$.
  To conclude, note that
  taking
  $\mstepB := \lam{\var_1\hdots\var_n}{\vartwo\,\mstepB_1\hdots\mstepB_m}$
  and
  $\msteptwoB := \lam{\var_1\hdots\var_n}{\vartwo\,\msteptwoB_1\hdots\msteptwoB_m}$
  we have that
  $\judgCompat{\mstepB}{\msteptwoB}$ by the \indrulename{CVar} rule,
  and moreover
  $\flatten{\mstepB} = \flatten{\mstep}$
  and
  $\flatten{\msteptwoB} = \flatten{\msteptwo}$.
\item
  {\bf Constant,
    $\mstep = \lam{\var_1\hdots\var_n}{\cons\,\mstep_1\hdots\mstep_m}$.
  }
  Note that
  $\flattennoeta{(\altsrc{\mstep})} =
   \lam{\var_1\hdots\var_n}{
     \cons\,\flattennoeta{(\altsrc{\mstep_1})}\hdots
            \flattennoeta{(\altsrc{\mstep_m})}
   }$.
  Since $\flattennoeta{(\altsrc{\msteptwo})} =
         \flattennoeta{(\altsrc{\mstep})}$,
  the head of $\msteptwo$ cannot be a variable.
  We consider two cases, depending on whether the head of $\msteptwo$
  is a constant or a rule symbol:
  \begin{enumerate}
  \item
    {\bf Constant,
      $\msteptwo = \lam{\var_1\hdots\var_p}{
                     \cons\,\msteptwo_1\hdots\msteptwo_q
                   }$.}
    The proof of this case proceeds similarly
    as for case~\ref{compatibilization_of_coinitial_multisteps__case_var},
    when the heads of $\mstep$ and $\msteptwo$ are both variables.
  \item
    \label{compatibilization_of_coinitial_multisteps__case_cons_rule}
    {\bf Rule symbol,
      $\msteptwo = \lam{\var_1\hdots\var_p}{\rulewit\,\msteptwo_1\hdots\msteptwo_q}$.
    }
    Recall that we assume that $\mstep$ and $\msteptwo$ are
    in $\etalong$-normal form. This implies that $n = p$.
    Then
    $\flattennoeta{(\altsrc{\msteptwo})} =
     \flattennoeta{
       (\lam{\var_1\hdots\var_n}{
         \altsrc{\rulewit}\,\altsrc{\msteptwo_1}\hdots\altsrc{\msteptwo_q}
       })
     }$.
    By \rlem{coinitial_cons_rule}, there exist multisteps
    $\mstepthree_1,\hdots,\mstepthree_q$
    such that
    $\mstep =
     \flattennoeta{
       (\lam{\var_1\hdots\var_n}{
         \altsrc{\rulewit}\,\mstepthree_1\hdots\mstepthree_q
       })
     }$
    and $\flattennoeta{(\altsrc{\mstepthree_i})} =
         \flattennoeta{(\altsrc{\msteptwo_i})}$
    for all $1 \leq i \leq q$.
    Moreover, each $\mstepthree_i$ has less applications than $\mstep$.
    So by \ih, for each $1 \leq i \leq q$
    there exist multisteps such that
    $\judgCompat{\mstepthreeB_i}{\msteptwoB_i}$
    where
    $\flatten{\mstepthreeB_i} = \flatten{\mstepthree_i}$.
    and
    $\flatten{\msteptwoB_i} = \flatten{\msteptwo_i}$.
    Taking
      $\mstepB :=
       \lam{\var_1\hdots\var_n}{
         \altsrc{\rulewit}\,\mstepthreeB_1\hdots\mstepthreeB_q
       }$
    and
      $\msteptwoB :=
       \lam{\var_1\hdots\var_n}{
         \rulewit\,\msteptwoB_1\hdots\msteptwoB_q
       }$
    we have that $\judgCompat{\mstepB}{\msteptwoB}$.
    It is easy to note that $\flatten{\msteptwoB} = \flatten{\msteptwo}$.
    Moreover:
    \[
      \begin{array}{rlll}
        \flatten{\mstepB}
      & = &
        \flatten{(\lam{\var_1\hdots\var_n}{
          \altsrc{\rulewit}\,\mstepthreeB_1\hdots\mstepthreeB_q
        })}
      \\
      & = &
        \flatten{(\lam{\var_1\hdots\var_n}{
          \altsrc{\rulewit}\,\mstepthree_1\hdots\mstepthree_q
        })}
        & \text{by confluence of flattening~(\rprop{flat_confluent})}
      \\
      & = &
        \flatten{(\flattennoeta{(\lam{\var_1\hdots\var_n}{
          \altsrc{\rulewit}\,\mstepthree_1\hdots\mstepthree_q
        })})}
        & \text{by confluence of flattening~(\rprop{flat_confluent})}
      \\
      & = &
        \flatten{\mstep}
      \end{array}
    \]
  \end{enumerate}
\item
  {\bf Rule symbol,
    $\mstep = \lam{\var_1\hdots\var_n}{\rulewit\,\mstep_1\hdots\mstep_m}$.
  }
  Note that
  $\flattennoeta{(\altsrc{\mstep})} =
   \lam{\var_1\hdots\var_n}{
     \flattennoeta{(
       \altsrc{\rulewit}\,\altsrc{\mstep_1}\hdots
                          \altsrc{\mstep_m}
     )}
   }$.
  Since $\flattennoeta{(\altsrc{\msteptwo})} =
         \flattennoeta{(\altsrc{\mstep})}$,
  the head of $\msteptwo$ cannot be a variable.
  We consider two cases, depending on whether the head of $\msteptwo$
  is a constant or a rule symbol:
  \begin{enumerate}
  \item
    {\bf Constant,
      $\msteptwo = \lam{\var_1\hdots\var_p}{\cons\,\msteptwo_1\hdots\msteptwo_q}$.}
    The proof of this case proceeds symmetrically
    as for case~\ref{compatibilization_of_coinitial_multisteps__case_cons_rule},
    when the head of $\mstep$ is a constant
    and the head of $\msteptwo$ is a rule symbol.
  \item
    {\bf Rule symbol,
      $\msteptwo = \lam{\var_1\hdots\var_p}{\rulewittwo\,\msteptwo_1\hdots\msteptwo_q}$.
    }
    Recall that we assume that $\mstep$ and $\msteptwo$ are in
    $\etalong$-normal form. This implies that $n = p$.
    Then
    $\flattennoeta{(\altsrc{\msteptwo})} =
     \lam{\var_1\hdots\var_n}{
       \flattennoeta{(
         \altsrc{\rulewittwo}\,\altsrc{\msteptwo_1}\hdots\altsrc{\msteptwo_q}
       )}
     }$.
    Note that this implies that
    $\flattennoeta{(\altsrc{\rulewit}\,\altsrc{\mstep_1}\hdots\altsrc{\msteptwo_m})}
     =
     \flattennoeta{(\altsrc{\rulewittwo}\,\altsrc{\msteptwo_1}\hdots\altsrc{\msteptwo_q})}$.
    By orthogonality,
    this means that $\rulewit = \rulewittwo$
    and $m = q$, and moreover
    $\flattennoeta{(\altsrc{\mstep_i})} = \flattennoeta{(\altsrc{\msteptwo_i})}$
    for all $1 \leq i \leq m$.
    By \ih, for each $1 \leq i \leq m$,
    there exist multisteps such that
    $\judgCompat{\mstepB_i}{\msteptwoB_i}$,
    where $\flatten{\mstepB_i} = \flatten{\mstep_i}$
    and $\flatten{\msteptwoB_i} = \flatten{\msteptwo_i}$.
    Taking
    $\mstepB =
     \lam{\var_1\hdots\var_n}{\rulewit\,\mstepB_1\hdots\mstepB_m}$
    and
    $\msteptwoB =
     \lam{\var_1\hdots\var_n}{\rulewit\,\msteptwoB_1\hdots\msteptwoB_m}$
    it is then easy to check that
    $\judgCompat{\mstepB}{\msteptwoB}$,
    and moreover $\flatten{\mstepB} = \flatten{\mstep}$
    and $\flatten{\msteptwoB} = \flatten{\msteptwo}$.
  \end{enumerate}
\end{enumerate}
\end{proof}

Thus for arbitrary, coinitial multisteps $\mstep$ and $\msteptwo$, it suffices to show that we can always find corresponding \emph{compatible} ``almost'' $\eta$-expanded and $\beta$-normal forms, as mentioned above.
\begin{prop}[Existence and uniqueness of projection]
\lprop{body:existence_uniqueness_projection}
Let $\mstep,\msteptwo$ be such that $\rsrc{\mstep} \termeq \rsrc{\msteptwo}$.
Then:
\begin{enumerate}
\item
  \resultName{Existence.}
  There exist multisteps $\mstepC,\msteptwoC,\mstepthreeC$
  such that $\flatten{\mstepC} = \flatten{\mstep}$
  and $\flatten{\msteptwoC} = \flatten{\msteptwo}$
  and $\judgProj{\mstepC}{\msteptwoC}{\mstepthreeC}$.
\item
  \resultName{Compatibility.}
  Furthermore, $\mstepC$ and $\msteptwoC$ can
  be chosen in such a way that $\judgCompat{\mstepC}{\msteptwoC}$.
\item
  \resultName{Uniqueness.}
  If
  $\flatten{(\mstepC')} = \flatten{\mstep}$
  and
  $\flatten{(\msteptwoC')} = \flatten{\msteptwo}$
  and
  $\judgProj{\mstepC'}{\msteptwoC'}{\mstepthreeC'}$
  then $\flatten{(\mstepthreeC')} = \flatten{\mstepthree}$.
\end{enumerate}
\end{prop}

\begin{proof}
Note that $\fsrc{\mstep} = \fsrc{\msteptwo}$,
so by \rlem{compatibilization_of_coinitial_multisteps}
there exist multisteps $\mstepB,\msteptwoB$
such that $\judgCompat{\mstepB}{\msteptwoB}$
and moreover $\flatten{\mstepB} = \flatten{\mstep}$
and $\flatten{\msteptwoB} = \flatten{\msteptwo}$.
Since for any $\judgCompat{\mstep}{\msteptwo}$, there exists a unique $\mstepthree$
such that $\judgProj{\mstep}{\msteptwo}{\mstepthree}$, we deduce that there exists a multistep $\mstepthreeB$
such that $\judgProj{\mstepB}{\msteptwoB}{\mstepthreeB}$.
This proves items~1. and~2.
For item~3., first a {\bf claim.} Let $\judgProj{\mstep}{\msteptwo}{\mstepthree}$.
Then there exist multisteps
$\mstepB,\msteptwoB,\mstepthreeB$
such that:
\begin{enumerate}
  \item $\mstep \tofs \mstepB$
  and
  $\msteptwo \tofs \msteptwoB$
  and
  $\mstepthree \tofs \mstepthreeB$
\item
  $\judgCompat{\mstepB}{\msteptwoB}$
\item
  $\judgProj{\mstepB}{\msteptwoB}{\mstepthreeB}$
\end{enumerate}
{\bf Proof of the claim.} First, if $\judgCompat{\mstep}{\msteptwo}$ holds, we are done.
Otherwise, by strong normalization of flattening~(\rprop{flat_sn})
it suffices to show that
if $\judgProj{\mstep}{\msteptwo}{\mstepthree}$
and $\judgCompat{\mstep}{\msteptwo}$ does {\em not} hold,
then there exist steps
$\mstep \tof \mstepB$
and
$\msteptwo \tof \msteptwoB$
and
$\mstepthree \tof \mstepthreeB$
such that $\judgProj{\mstepB}{\msteptwoB}{\mstepthreeB}$.
This can be proved by induction on the derivation of
$\judgProj{\mstep}{\msteptwo}{\mstepthree}$. {\bf End of proof of claim.}

Now suppose that
$\mstepB',\msteptwoB',\mstepthreeB'$ are such that
$\flatten{(\mstepB')} = \flatten{\mstep}$
and
$\flatten{(\msteptwoB')} = \flatten{\msteptwo}$
and
$\judgProj{\mstepB'}{\msteptwoB'}{\mstepthreeB'}$.
By the claim %\rlem{compatibilization_of_projection}
this implies that
there exist multisteps $\mstepB'',\msteptwoB'',\mstepthreeB''$
such that
$\mstepB' \tofs \mstepB''$
and
$\msteptwoB' \tofs \msteptwoB''$
and
$\mstepthreeB' \tofs \mstepthreeB''$,
and moreover
$\judgCompat{\mstepB''}{\msteptwoB''}$
and $\judgProj{\mstepB''}{\msteptwoB''}{\mstepthreeB''}$.
Note that
$\flatten{\mstepB'} = \flatten{\mstepB''}$
given that $\mstepB' \tofs \mstepB''$.
Similarly,
$\flatten{\msteptwoB'} = \flatten{\msteptwoB''}$
and
$\flatten{\mstepthreeB'} = \flatten{\mstepthreeB''}$.
Hence by \rlem{coherence_of_projection}
we may conclude that
$\flatten{\mstepthreeB} = \flatten{\mstepthreeB'}$,
as required.
\end{proof}

We can now define projection on arbitrary coinitial rewrites as follows.

\begin{defi}[Projection operator for multisteps]
\ldef{projection_operator}
Let $\mstep,\msteptwo$ be such that $\rsrc{\mstep} \termeq \rsrc{\msteptwo}$.
We write $\mstep/\msteptwo$
for the unique multistep of the form $\flatten{\mstepthreeC}$
such that there exist $\mstepC,\msteptwoC$
such that $\flatten{\mstepC} = \flatten{\mstep}$
and $\flatten{\msteptwoC} = \flatten{\msteptwo}$
and $\judgProj{\mstepC}{\msteptwoC}{\mstepthreeC}$,
as guaranteed by~\rprop{body:existence_uniqueness_projection}.
The proof is constructive (this relies on the HRS being orthogonal),
thus providing an effective method to compute $\mstep/\msteptwo$.
\end{defi}

\begin{prop}[Properties of projection for multisteps]
  \lprop{body:properties_of_projection}
  \lprop{body:cube_lemma}
\quad
\begin{enumerate}
\item $\mstep/\msteptwo = \flatten{(\mstep/\msteptwo)} = \flatten{\mstep}/\flatten{\msteptwo}$
\item Projection commutes with abstraction and application,
      that is,
      $(\lam{\var}{\mstep})/(\lam{\var}{\msteptwo}) = \flatten{(\lam{\var}{(\mstep/\msteptwo)})}$
      and
      $(\mstep_1\,\mstep_2)/(\msteptwo_1\,\msteptwo_2) =
       \flatten{((\mstep_1/\msteptwo_1)\,(\mstep_2/\msteptwo_2))}$,
      provided that
      $\mstep_1/\msteptwo_1$ and $\mstep_2/\msteptwo_2$
      are defined.
\item
  The set of multisteps with the projection operator form
  a residual system~\cite[Def.~8.7.2]{Terese03}:
  \begin{enumerate}
  \item 
    $(\mstep/\msteptwo)/(\mstepthree/\msteptwo) =
     (\mstep/\mstepthree)/(\msteptwo/\mstepthree)$,
    known as the \resultName{Cube Lemma}.
  \item $\mstep/\mstep = \ftgtb{\mstep}$
    and, as particular cases:
    $\refl{\tm}/\refl{\tm} = \flatten{\refl{\tm}}$,
    $\var/\var = \var$, 
    $\cons/\cons = \cons$, and 
    $\rulewit/\rulewit = \ftgtb{\rulewit}$.
  \item $\fsrcb{\mstep}/\mstep = \ftgtb{\mstep}$
        and, as a particular case,
        $\fsrcb{\rulewit}/\rulewit = \ftgtb{\rulewit}$.
  \item $\mstep/\fsrcb{\mstep} = \flatten{\mstep}$
        and, as a particular case,
        $\rulewit/\fsrcb{\rulewit} = \rulewit$.
  \end{enumerate}
\end{enumerate}
\end{prop}

\begin{proof}
We prove each item:
\begin{enumerate}
\item $\mstep/\msteptwo = \flatten{(\mstep/\msteptwo)}$:
  By definition
  there exist $\mstepB,\msteptwoB,\mstepthreeB$
  such that $\flatten{\mstepB} = \flatten{\mstep}$
  and $\flatten{\msteptwoB} = \flatten{\msteptwo}$
  and $\judgProj{\mstepB}{\msteptwoB}{\mstepthreeB}$
  where $\mstep/\msteptwo = \flatten{\mstepthreeB}$.
  Then
  $\mstep/\msteptwo
  = \flatten{\hat{\mstepthree}}
  = \flatten{(\flatten{\hat{\mstepthree}})}
  = \flatten{(\mstep/\msteptwo)}$.
\item $\mstep/\msteptwo = \flatten{\mstep}/\flatten{\msteptwo}$:
  By definition of $\mstep/\msteptwo$
  there exist $\mstepB,\msteptwoB,\mstepthreeB$
  such that $\flatten{\mstepB} = \flatten{\mstep}$
  and $\flatten{\msteptwoB} = \flatten{\msteptwo}$
  and $\judgProj{\mstepB}{\msteptwoB}{\mstepthreeB}$
  where $\mstep/\msteptwo = \flatten{\mstepthreeB}$.
  Then since $\flatten{\mstep} = \flatten{(\flatten{\mstep})}$
  and $\flatten{\msteptwo} = \flatten{(\flatten{\msteptwo})}$,
  the triple $\mstepB,\msteptwoB,\mstepthreeB$
  also fulfills the conditions for the definition of
  $\flatten{\mstep}/\flatten{\msteptwo}$.
  But~\rprop{body:existence_uniqueness_projection} ensures uniqueness,
  so
  $\mstep/\msteptwo
  = \flatten{\mstepthreeB}
  = \flatten{\mstep}/\flatten{\msteptwo}$.
\item $\mstep/\mstep = \ftgt{\mstep}$:
  It suffices to note that $\judgProj{\mstep}{\mstep}{\rtgt{\mstep}}$ 
  holds, as can be checked by induction on $\mstep$.
\item $\mstep/\fsrc{\mstep} = \flatten{\mstep}$:
  It suffices to note that
  $\judgProj{\mstep}{\rsrc{\mstep}}{\mstep}$,
  as can be checked by induction on $\mstep$.
\item $\fsrc{\mstep}/\mstep = \ftgt{\mstep}$:
  It suffices to note that
  $\judgProj{\rsrc{\mstep}}{\mstep}{\rtgt{\mstep}}$,
  as can be checked by induction on $\mstep$.
\item $(\lam{\var}{\mstep})/(\lam{\var}{\msteptwo})
      = \flatten{(\lam{\var}{(\mstep/\msteptwo)})}$:
  Observe that the left-hand side of the equation
  is defined if and only if the right-hand
  side is defined, given that
  $\rsrc{(\lam{\var}{\mstep})} \termeq \rsrc{(\lam{\var}{\msteptwo})}$
  if and only if $\rsrc{\mstep} \termeq \rsrc{\msteptwo}$,
  which is easy to check.

  By definition of $\mstep/\msteptwo$
  there exist $\mstepB,\msteptwoB,\mstepthreeB$
  such that $\flatten{\mstepB} = \flatten{\mstep}$
  and $\flatten{\msteptwoB} = \flatten{\msteptwo}$
  and $\judgProj{\mstepB}{\msteptwoB}{\mstepthreeB}$
  where $\mstep/\msteptwo = \flatten{\mstepthreeB}$.
  Then note, by confluence of flattening~(\rprop{flat_confluent}),
  that
  $\flatten{(\lam{\var}{\mstepB})}
   = \flatten{(\lam{\var}{\flatten{\mstepB}})}
   = \flatten{(\lam{\var}{\flatten{\mstep}})}
   = \flatten{(\lam{\var}{\mstep})}$
  and, similarly,
  $\flatten{(\lam{\var}{\msteptwoB})} = \flatten{(\lam{\var}{\msteptwo})}$.
  Moreover, by the \indrulename{ProjAbs} rule,
  $\judgProj{
     \lam{\var}{\mstepB}
   }{
     \lam{\var}{\msteptwoB}
   }{
     \lam{\var}{\mstepthreeB}
   }$.
  By uniqueness of projection~(\rprop{body:existence_uniqueness_projection})
  this means that
  $(\lam{\var}{\mstep})/(\lam{\var}{\msteptwo})
  = \flatten{(\lam{\var}{\mstepthreeB})}
  = \flatten{(\lam{\var}{\flatten{\mstepthreeB}})}
  = \flatten{(\lam{\var}{(\mstep/\msteptwo)})}$.
\item $(\mstep_1\,\mstep_2)/(\msteptwo_1\,\msteptwo_2) =
       \flatten{((\mstep_1/\msteptwo_1)\,(\mstep_2/\msteptwo_2))}$:
  Observe that, if the right-hand side of the equation is defined,
  then the left-hand side is also defined, given that
  if $\rsrc{\mstep_1} \termeq \rsrc{\msteptwo_1}$ 
  and $\rsrc{\mstep_2} \termeq \rsrc{\msteptwo_2}$
  then
  $\rsrc{(\mstep_1\,\mstep_2)} \termeq \rsrc{(\msteptwo_1\,\msteptwo_2)}$.

  Note that, by hypothesis, the right-hand side of the equation is defined.
  By definition of $\mstep_1/\msteptwo_1$
  there exist $\mstepB_1,\msteptwoB_1,\mstepthreeB_1$
  such that $\flatten{\mstepB_1} = \flatten{\mstep_1}$
  and $\flatten{\msteptwoB_1} = \flatten{\msteptwo_1}$
  and $\judgProj{\mstepB_1}{\msteptwoB_1}{\mstepthreeB_1}$
  where $\mstep_1/\msteptwo_1 = \flatten{\mstepthreeB_1}$.
  Similarly,
  by definition of $\mstep_2/\msteptwo_2$
  there exist $\mstepB_2,\msteptwoB_2,\mstepthreeB_2$
  such that $\flatten{\mstepB_2} = \flatten{\mstep_2}$
  and $\flatten{\msteptwoB_2} = \flatten{\msteptwo_2}$
  and $\judgProj{\mstepB_2}{\msteptwoB_2}{\mstepthreeB_2}$
  where $\mstep_2/\msteptwo_2 = \flatten{\mstepthreeB_2}$.

  Note, by confluence of flattening~(\rprop{flat_confluent}),
  $\flatten{(\mstepB_1\,\mstepB_2)}
  = \flatten{(\flatten{\mstepB_1}\,\flatten{\mstepB_2})}
  = \flatten{(\flatten{\mstep_1}\,\flatten{\mstep_2})}
  = \flatten{(\mstep_1\,\mstep_2)}$ and,
  similarly,
  $\flatten{(\msteptwoB_1\,\msteptwoB_2)}
  = \flatten{(\msteptwo_1\,\msteptwo_2)}$.
  Moreover, by the \indrulename{ProjApp} rule,
  $\judgProj{
     \mstepB_1\,\mstepB_2
   }{
     \msteptwoB_1\,\msteptwoB_2
   }{
     \mstepthreeB_1\,\mstepthreeB_2
   }$.
  By uniqueness of projection~\rprop{body:existence_uniqueness_projection}
  this means that
  $\mstep_1\,\mstep_2/\msteptwo_1\,\msteptwo_2
  = \flatten{(\mstepthreeB_1\,\mstepthreeB_2)}
  = \flatten{(\flatten{\mstepthreeB_1}\,\flatten{\mstepthreeB_2})}
  = \flatten{((\mstep_1/\msteptwo_1)\,(\mstep_2/\msteptwo_2))}$.
\end{enumerate}
\end{proof}

\begin{exa}
Let $\rulewittwo : \lam{\var}{\consof{f}\,\var} \to \lam{\var}{\consof{g}\,\var}$.
Then:
\[
  \begin{array}{lll}
  
    (\lam{\var}{(\lam{\var}{\consof{f}\,\var})\,\var})/(\lam{\var}{\rulewittwo\,\var})
  & = &
    \flatten{(
      \lam{\var}{
          ((\lam{\var}{\consof{f}\,\var})\,\var)/(\rulewittwo\,\var)
      }
    )}
  \\
   & = &
    \flatten{(
      \lam{\var}{
        \flatten{(
          ((\lam{\var}{\consof{f}\,\var})/\rulewittwo)
          (\var/\var)
        )}
      }
    )}
  \\
  & = &
    \flatten{(
      \lam{\var}{
        \flatten{(
          (\lam{\var}{\consof{g}\,\var})\,\var
        )}
      }
        )}
        \\
  & = &
    \flatten{(
      \lam{\var}{
        \consof{g}\,\var
      }
        )}
        \\
  & = &
    \consof{g}
  \end{array}
\]
\end{exa}

\subsection{Projection for Flat Rewrites.}
The projection operator from~\rdef{projection_operator}
is extended to operate on flat rewrites.
One may try to define
$\redseq/\redseqtwo$
using equations such as
$(\redseq_1\seq\redseq_2)/\redseqtwo =
 (\redseq_1/\redseqtwo)\seq(\redseq_2/(\redseqtwo/\redseq_1))$.
However, it is not {\em a priori} clear that this recursive definition
is well-founded\footnote{Another way to prove
well-foundedness is by interpretation,
as done in~\cite[Example~6.5.43]{Terese03}.}.
This is why the following definition proceeds in three stages:

\begin{defi}[Projection operator for flat rewrites]
\ldef{projection_for_rewrites}
We define:
\begin{enumerate}
\item %\inlineItem{1.}
  projection of a flat multistep over a coinitial flat rewrite
  ($\mstep\projmr\redseq$),
  by induction on $\redseq$;
\item %\inlineItem{2.}
  projection of a flat rewrite over a coinitial flat multistep
  ($\redseq\projrm\mstep$),
  by induction on $\redseq$;
and
\item %\inlineItem{3.}
  projection of a flat rewrite over a coinitial flat rewrite
  ($\redseq\projrr\redseqtwo$)
  by induction on $\redseqtwo$,
  as follows:
\end{enumerate}  
\[
  \begin{array}{rcl@{\hspace{1cm}}rcl}
     \mstep\projmr\msteptwo
     & \eqdef &
     \mstep/\msteptwo
   &
     \mstep\projmr(\redseq_1\seq\redseq_2)
     & \eqdef &
      (\mstep\projmr\redseq_1)\projmr\redseq_2
  \\
    \msteptwo\projrm\mstep
    & \eqdef &
    \msteptwo/\mstep
  &
    (\redseq_1\seq\redseq_2)\projrm\mstep
    & \eqdef &
    (\redseq_1\projrm\mstep)\seq(\redseq_2\projrm(\mstep\projmr\redseq_1))
  \\
    \redseq\projrr\mstep
      & \eqdef &
      \redseq\projrm\mstep
    &
      \redseq\projrr(\redseqtwo_1\seq\redseqtwo_2)
      & \eqdef &
      (\redseq\projrr\redseqtwo_1)\projrr\redseqtwo_2
  \end{array}
\]
\end{defi}

Note that $\projrr$ generalizes $\projrm$ and $\projmr$
in the sense that $\mstep\projmr\redseq = \mstep\projrr\redseq$
and $\redseq\projrm\mstep = \redseq\projrr\mstep$.
With these definitions, the key equation
$
  (\redseq_1\seq\redseq_2)\projrr\redseqtwo =
  (\redseq_1\projrr\redseqtwo)\seq(\redseq_2\projrr(\redseqtwo\projrr\redseq_1))
$
can be shown to hold.
From this point on, we overload $\redseq/\redseqtwo$
to stand for either of these projection operators.
The key equation ensures that this abuse of notation is harmless.

Next we address some important properties of projection
for flat rewrites.
First, projection of a rewrite over a sequence, and of a sequence
over a rewrite, obey the expected equations
$
  \redseq/(\redseqtwo_1\seq\redseqtwo_2) =
  (\redseq/\redseqtwo_1)/\redseqtwo_2
$
and
  $(\redseq_1\seq\redseq_2)/\redseqtwo =
  (\redseq_1/\redseqtwo)\seq(\redseq_2/(\redseqtwo/\redseq_1))$.
  The former follows from \rdef{projection_for_rewrites} and the latter is
  proved by induction on $\sz{\redseq_1}+\sz{\redseq_2}+\sz{\redseqtwo}$ for
  $\sz{\cdot}$ an appropriate notion of size of rewrites. Second, flat
  permutation equivalence is a congruence with respect to projection:
  if $\redseq \flateq \redseqtwo$ 
and $\redseqthree$ is an arbitrary flat rewrite
coinitial to $\redseq$, then both
(1) $\redseqthree/\redseq = \redseqthree/\redseqtwo$;
and (2) $\redseq/\redseqthree \flateq \redseqtwo/\redseqthree$ (\rprop{congruence_flateq_projection}).
The proof of both items is by induction on $\redseq \flateq \redseqtwo$; the more challenging case is when $\redseq \flateq \redseqtwo$ follows from the $\flateqRule{Perm}$ rule:
\begin{itemize}
\item For (1), we prove (\llem{projection_of_splittings_and_rewrites}):
If $\judgSplit{\mstep_1}{\mstep_2}{\mstep_3}$
and $\redseq$ is an arbitrary flat rewrite coinitial to $\mstep_1$,
then $\redseq/\mstep_1 = (\redseq/\mstep_2)/\mstep_3$.

\item For (2), we prove (\rprop{generalized_flateq_perm}):
If $\judgSplit{\mstep_1}{\mstep_2}{\mstep_3}$
and $\msteptwo$ is an arbitrary multistep coinitial to $\mstep_1$,
then
$\judgSplitF{
   (\mstep_1/\msteptwo)
 }{
   (\mstep_2/\msteptwo)
 }{
   (\mstep_3/(\msteptwo/\mstep_2))
 }$.
\end{itemize}

\begin{lem}%[Rewrite/rewrite projection of a sequence]
\llem{projrr_sequence}
If $(\redseq_1\seq\redseq_2)$ and $\redseqtwo$ are coinitial flat
rewrites, then:
\[
  (\redseq_1\seq\redseq_2)\projrr\redseqtwo =
  (\redseq_1\projrr\redseqtwo)\seq(\redseq_2\projrr(\redseqtwo\projrr\redseq_1))
\]
\end{lem}
\begin{proof}
The proof is by induction on $\sz{\redseq_1} + \sz{\redseq_2} + \sz{\redseqtwo}$.
\end{proof}

\begin{lem}[Basic properties of rewrite projection]
\llem{basic_properties_of_rewrite_projection}
Let $\redseq$ stand for a flat rewrite.
Then both 
%\begin{enumerate}
$\fsrc{\redseq}/\redseq = \ftgt{\redseq}$ and
$\redseq/\fsrc{\redseq} = \redseq$ hold.
%\end{enumerate}
\end{lem}
\begin{proof}
Both items are proved by induction on $\redseq$ and make use of~\rprop{body:properties_of_projection}. We present the proof of the second one as a sample. 
%   \quad
% \begin{enumerate}
% \item $\fsrc{\redseq}/\redseq = \ftgt{\redseq}$:
%   By induction on $\redseq$.
%   If $\redseq = \mstep$ is a multistep,
%   then $\fsrc{\mstep}/\mstep = \ftgt{\mstep}$
%   by~\rprop{body:properties_of_projection}.
%   If $\redseq = \redseq_1\seq\redseq_2$,
%   then:
%   \[
%   \begin{array}{rlll}
%     \fsrc{(\redseq_1\seq\redseq_2)}/(\redseq_1\seq\redseq_2)
%   & = &
%     (\fsrc{\redseq_1}/\redseq_1)/\redseq_2
%     & \text{by definition}
%   \\
%   & = &
%     \ftgt{\redseq_1}/\redseq_2
%     & \text{by \ih}
%   \\
%   & = &
%     \fsrc{\redseq_2}/\redseq_2
%     & \text{as $\rtgt{\redseq_1} \termeq \rsrc{\redseq_2}$}
%   \\
%   & = &
%     \ftgt{\redseq_2}
%     & \text{by \ih}
%   \\
%   & = &
%     \ftgt{(\redseq_1\seq\redseq_2)}
%   \end{array}
%   \]
%\item $\redseq/\fsrc{\redseq} = \redseq$:
%  By induction on $\redseq$.
  If $\redseq = \mstep$ is a multistep,
  then $\mstep/\fsrc{\mstep} = \flatten{\mstep} = \mstep$
  by~\rprop{body:properties_of_projection}, using the fact that
  $\redseq$ is flat by hypothesis.
  If $\redseq = \redseq_1\seq\redseq_2$ then:
  \[
  \begin{array}{rlll}
  &&
    (\redseq_1\seq\redseq_2)/\fsrc{(\redseq_1\seq\redseq_2)}
  \\
  & = &
    (\redseq_1/\fsrc{(\redseq_1\seq\redseq_2)})
    \seq
    (\redseq_2/(\fsrc{(\redseq_1\seq\redseq_2)}/\redseq_1))
    & \text{by \rlem{projrr_sequence}}
  \\
  & = &
    (\redseq_1/\fsrc{(\redseq_1\seq\redseq_2)})
    \seq
    (\redseq_2/(\fsrc{\redseq_1}/\redseq_1))
  \\
  & = &
    (\redseq_1/\fsrc{(\redseq_1\seq\redseq_2)})
    \seq
    (\redseq_2/\ftgt{\redseq_1})
    & \text{by item~1. of this lemma}
  \\
  & = &
    (\redseq_1/\fsrc{\redseq_1})
    \seq
    (\redseq_2/\fsrc{\redseq_2})
    & \text{as $\rsrc{\redseq_1\seq\redseq_2} = \rsrc{\redseq_1}$
            and $\rtgt{\redseq_1} = \rsrc{\redseq_2}$}
  \\
  & = &
    \redseq_1\seq\redseq_2
    & \text{by \ih}
  \end{array}
  \]
%\end{enumerate}
\end{proof}

First consider the relation $\tofnoeta$, which is defined like $\tof$
but without the $\flatRule{EtaM}$ rule.

\begin{lem}[Canonical $\mapsto^\circ,\etalong$-normal splitting]
\llem{canonical_splitting}
If $\judgSplit{\mstep_1}{\mstep_2}{\mstep_3}$
then there exist $\mstep'_1,\mstep'_2,\mstep'_3$
such that $\judgSplit{\mstep'_1}{\mstep'_2}{\mstep'_3}$
where
$\flatten{\mstep_1} = \flatten{(\mstep'_1)}$
and
$\flatten{\mstep_2} = \flatten{(\mstep'_2)}$
and
$\flatten{\mstep_3} = \flatten{(\mstep'_3)}$,
and moreover $\mstep'$ is in $\tofnoeta,\etalong$-normal form.
\end{lem}

\begin{proof}
Suppose  $\judgSplit{\mstep_1}{\mstep_2}{\mstep_3}$. 
Then 
we have that
$\judgSplit{\mstep''_1}{\mstep''_2}{\mstep''_3}$
where
$\mstep''_1 = \flatten{\mstep_1}$
and $\mstep_2 \tofs \mstep''_2$
and $\mstep_3 \tofs \mstep''_3$.
This fact may be proved  by induction on the length of a reduction to normal form $\mstep_1 \tofs \flatten{\mstep_1}$. One then proceeds by induction on the shape of $\mstep_1''$,
it suffices to show that there exist $\mstep'_1,\mstep'_2,\mstep'_3$
such that
$\judgSplit{\mstep'_1}{\mstep'_2}{\mstep'_3}$
where $\flatten{(\mstep'_1)} = \flatten{(\mstep''_1)}$
and $\flatten{(\mstep'_2)} = \flatten{(\mstep''_2)}$
and $\flatten{(\mstep'_3)} = \flatten{(\mstep''_3)}$,
and moreover $\mstep'_1$ is in $\tofnoeta,\etalong$-normal form. This is done by direct inspection considering each of the cases: $\mstep''_1$ headed by a variable, $\mstep''_1$ headed by a constant, and $\mstep''_1$ headed by a rule symbol.
\end{proof}

\begin{defi}[$\eta$-expanded source]
If $\mstep$ is a multistep in $\etalong$-normal form,
we write $\altsrc{\mstep}$ for the source of $\mstep$
in which the sources of rule symbols are also $\eta$-expanded.
More precisely:
\[
  \begin{array}{rcll}
    \altsrc{\var} & \eqdef & \var \\
    \altsrc{\cons} & \eqdef & \cons \\
    \altsrc{\rulewit} & \eqdef & \tm'
    & \text{
        if $\rewr{\rulewit}{\tm}{\tmtwo}{\typ} \in \ruleset$
        and $\tm'$ is the $\etalong$-normal form of $\tm$
      } \\
    \altsrc{(\lam{\var}{\mstep})} & \eqdef & \lam{\var}{\altsrc{\tm}} \\
    \altsrc{(\mstep\,\msteptwo)} & \eqdef & \altsrc{\mstep}\,\altsrc{\msteptwo}
  \end{array}
\]
Note that $\altsrc{\mstep} \tofs \rsrc{\mstep}$
so, in particular,
$\flatten{(\altsrc{\mstep})} = \flatten{(\rsrc{\mstep})} = \fsrc{\mstep}$.
\end{defi}

\begin{lem}[Flattening of $\eta$-expanded multisteps]
\llem{flattening_of_eta_expanded_multisteps}
Let $\mstep,\msteptwo$
be multisteps in $\etalong$-normal form
such that
$\flatten{\mstep} = \flatten{\msteptwo}$.
Then $\flattennoeta{\mstep} = \flattennoeta{\msteptwo}$.
\end{lem}
\begin{proof}
This essentially follows from the fact that \flatRule{EtaM}-redexes
may be postponed after \flatRule{BetaM}-redexes
(a standard result, regarding multisteps
as terms of the simply-typed $\lambda$-calculus) and that
flattening preserves
$\etalong$-normal forms.
\end{proof}

\begin{lem}[Constructor/rule matching]
\llem{coinitial_cons_rule}
Let $\mstep = \lam{\var_1\hdots\var_n}{\cons\,\mstep_1\hdots\mstep_m}$
and also
let $\msteptwo = \lam{\vartwo_1\hdots\vartwo_p}{
                   \rulewit\,\msteptwo_1\hdots\msteptwo_q}$.
Let $\rewr{\rulewit}{l_0}{r_0}{\typ}\in\ruleset$
and let $l,r$ be the $\toetaexp$-normal forms of $l_0,r_0$ respectively.
Suppose that $\rsrc{\mstep} \termeq \rsrc{\msteptwo}$
and, moreover, that $\mstep$ and $\msteptwo$ are in
$\tofnoeta,\etalong$-normal form. Then there exist multisteps
$\mstepthree_1,\ldots,\mstepthree_q$ such that:
\[
  \mstep = \flattennoeta{
             (\lam{\var_1\hdots\var_n}{
                \refl{l}\,\mstepthree_1\hdots\mstepthree_q})
           }
\]
and
$\flattennoeta{(\altsrc{\mstepthree_i})}
 = \flattennoeta{(\altsrc{\msteptwo_i})}$
for all $1 \leq i \leq q$.\smallskip\\
Furthermore, the multisteps $\mstepthree_i$
can be chosen in such a way that they are in $\tofnoeta,\etalong$-normal form,
and the number of applications in each
multistep $\mstepthree_i$, for all $1 \leq i \leq k$,
is strictly less than the number of applications in $\mstep$.
\end{lem}

\begin{proof}
Note that the condition $\rsrc{\mstep} \termeq \rsrc{\msteptwo}$
implies that $\flatten{(\rsrc{\mstep})} = \flatten{(\rsrc{\msteptwo})}$.
Then also $\flatten{(\altsrc{\mstep})} = \flatten{(\altsrc{\msteptwo})}$.
By \rlem{flattening_of_eta_expanded_multisteps} this implies that:
\begin{equation}
  \flattennoeta{(\altsrc{\mstep})} = \flattennoeta{(\altsrc{\msteptwo})}
  \label{eq:orthogonal:multisteps}
\end{equation}
Suppose
$l=\lam{\varthree_1}{\ldots \lam{\varthree_k}{\constwo\,r_1\ldots r_{k'}}}$.
From~(\ref{eq:orthogonal:multisteps}) we have:
\[
  \flattennoeta{(\lam{\var_1\hdots\var_n}{
                  \cons\,\altsrc{\mstep_1}\hdots\altsrc{\mstep_m}})}
  =
  \flattennoeta{(\lam{\vartwo_1\hdots\vartwo_p}{
                  l\,\altsrc{\msteptwo}_1\hdots\altsrc{\msteptwo_q}})}
\]
Since
$\cons\,\altsrc{\mstep_1}\hdots\altsrc{\mstep_m}$ and
$l\,\altsrc{\msteptwo}_1\hdots\altsrc{\msteptwo_q}$ have base
types, then $n=p$
and we may assume that $\var_i=\vartwo_i$, for all $1 \leq i \leq p$.
Moreover,
\[
  \flattennoeta{(\cons\,\altsrc{\mstep_1}\hdots\altsrc{\mstep_m})}
  = \flattennoeta{(l\,\altsrc{\msteptwo}_1\hdots\altsrc{\msteptwo_q})}
\]
Therefore $\constwo=\cons$, $q=k$ and $m=k'$, and:
\[
  \flattennoeta{(\altsrc{\mstep_i})}
  =
  \flattennoeta{
    (r_i\subm{\varthree_1}{\altsrc{\msteptwo}_1}
        \ldots
        \subm{\varthree_k}{\altsrc{\msteptwo_k}})
  }
\]
for all $1 \leq i \leq m$.
Note that each $r_i$ is part of the pattern of
the rewrite rule $l$. By orthogonality of $\ruleset$,
there exist $\mstepthree_1,\hdots,\mstepthree_k$ such that:
\[
  \mstep_i =
  \flattennoeta{(r_i\subm{\varthree_1}{\mstepthree_1}
                    \ldots
                    \subm{\varthree_k}{\mstepthree_k})}
  \HS\text{for all $1 \leq i \leq m$}
\]
and 
\[
  \flattennoeta{(\altsrc{\mstepthree_i})} 
  = \flattennoeta{(\altsrc{\msteptwo_i})}
  \HS\text{for all $1 \leq i \leq k$}
\]
Therefore:
\[
  \mstep
  = \lam{\var_1\hdots\var_n}{\cons\,\mstep_1\ldots \mstep_m}
  =
  \flattennoeta{
    (\lam{\var_1\hdots\var_n}{l\,\mstepthree_1\hdots\mstepthree_k})
  }
\]

Furthermore, we claim that each $\mstepthree_i$, for $1 \leq i \leq k$,
can be chosen in such a way that it is in $\tofnoeta,\etalong$-normal form
and with strictly less applications than $\mstep$.
First, note that if the $\mstepthree_i$ are not in $\tofnoeta$-normal form,
then we may take $\mstepthreeB_i := \flattennoeta{\mstepthree_i}$
instead, and we still have that
$
  \mstep
  =
  \flattennoeta{
    (\lam{\var_1\hdots\var_n}{l\,\mstepthree_1\hdots\mstepthree_k})
  }
  =
  \flattennoeta{
    (\lam{\var_1\hdots\var_n}{
       l\,\mstepthreeB_1
          \hdots
          \mstepthreeB_k
     }
    )
  }
$
by confluence of flattening~(\rprop{flat_confluent}).
So let us assume that the $\mstepthree_i$ are in $\tofnoeta$-normal form.
Furthermore, suppose that
$\mstepthree_i$ is of arity $N$,
\ie its type is of the form
$\typ_{1} \imp \hdots \imp \typ_{N} \imp \btyp$
with $\btyp$ a base type,
and suppose that
$\mstepthree_i = \lam{\varfour_1\hdots\varfour_M}{\mstepthree'_i}$
where $\mstepthree'_i$ is not a $\lambda$-abstraction.
Note that $M \leq N$.
Take
$\mstepthreeB_i :=
 \lam{\varfour_1\hdots\varfour_M\varfour_{M+1}\hdots\varfour_{N}}{
   \mstepthree'_i\,\varfour_{M+1}\hdots\varfour_{N}
 }$.
Note that each $\mstepthreeB_i$ is in $\tofnoeta$-normal form.

To conclude, we claim that
$
  \mstep
  =
  \flattennoeta{
    (\lam{\var_1\hdots\var_n}{\refl{l}\,\mstepthreeB_1\hdots\mstepthreeB_k})
  }
$
and that each $\mstepthreeB_i$
has strictly less applications than $\mstep$.
To see this, note that
each variable $\varthree_1,\hdots,\varthree_k$.
occurs free exactly once in the body of $l$,
applied to different bound variables.
More precisely, for a fixed index $1 \leq i_0 \leq k$
there is exactly one $1 \leq j_0 \leq m$
such that $\varthree_{i_0}$ occurs free in $r_j$.
Then the multistep that results from substituting in $r_j$
each $\varthree_i$ by $\mstepthreeB_i$ other than for $i = i_0$
contains exactly one occurrence of $\varthree_{i_0}$
applied to $N$ variables, where $N$ is the arity of $\mstepthreeB_{i_0}$,
that is, it is of the form:
\[
  r_j\subm{\varthree_1}{\mstepthree_1}
     \hdots
     \subm{\varthree_{(i_0-1)}}{\mstepthree_{(i_0-1)}}
     \subm{\varthree_{(i_0+1)}}{\mstepthree_{(i_0+1)}}
     \hdots
     \subm{\varthree_k}{\mstepthree_k}
  =
  \cctxof{\varthree_{i_0}\,\varfour_1 \hdots \varfour_N}
\]
where $\varfour_1,\hdots,\varfour_N$ are bound by $\cctx$
and the hole of $\cctx$ is not applied.
As a consequence, we have that:
\[
  \begin{array}{rll}
    \mstep
  & = &
    \lam{\var_1\hdots\var_n}{
      \cons\mstep_1
           \hdots
           \mstep_{i_0-1}\,
           \cctxof{\flattennoeta{(\mstepthreeB_{i_0}\,\varfour_1\hdots\,\varfour_N)}}\,
           \mstep_{i_0+1}
           \hdots
           \mstep_m
    }
  \\
  & = &
    \lam{\var_1\hdots\var_n}{
      \cons\mstep_1
           \hdots
           \mstep_{i_0-1}\,
           \cctxof{
             \flattennoeta{
               ((\lam{\varfour_1\hdots\varfour_N}{
                   \mstepthree_{i_0}\,\varfour_{M+1}\hdots\varfour_{N}
                 })
                 \,\varfour_1\hdots\,\varfour_N)
             }
           }\,
           \mstep_{i_0+1}
           \hdots
           \mstep_m
    }
  \\
  & = &
    \lam{\var_1\hdots\var_n}{
      \cons\mstep_1
           \hdots
           \mstep_{i_0-1}\,
           \cctxof{
             \mstepthree_{i_0}\,\varfour_{M+1}\hdots\varfour_{N}
           }\,
           \mstep_{i_0+1}
           \hdots
           \mstep_m
    }
  \end{array}
\]
Hence, since $\mstep$ is in $\etalong$-normal form,
and $\mstepthree_{i_0}\,\varfour_{M+1}\hdots\varfour_{N}$
is a subterm of $\mstep$ which is not applied,
we know that $\mstepthree_{i_0}\,\varfour_{M+1}\hdots\varfour_{N}$
is in $\etalong$-normal form.
Hence
$\mstepthreeB_{i_0} =
 \lam{\varfour_1\hdots\varfour_n}{
   \mstepthree_{i_0}\,\varfour_{M+1}\hdots\varfour_{N}
 }$
is also in $\etalong$-normal form.

Finally,
note that
$
  \flattennoeta{
    (\lam{\var_1\hdots\var_n}{\refl{l}\,\mstepthreeB_1\hdots\mstepthreeB_k})
  }
$
is in $\tofnoeta,\etalong$-normal form
and that
we have that
$\flatten{\mstep} =
 \flatten{
   (\lam{\var_1\hdots\var_n}{\refl{l}\,\mstepthree_1\hdots\mstepthree_k})
 }
 =
 \flatten{
   (\lam{\var_1\hdots\var_n}{\refl{l}\,\mstepthreeB_1\hdots\mstepthreeB_k})
 }$
so by \rlem{flattening_of_eta_expanded_multisteps}
we have that
$\flattennoeta{\mstep} =
 \flattennoeta{
   (\lam{\var_1\hdots\var_n}{\refl{l}\,\mstepthreeB_1\hdots\mstepthreeB_k})
 }$.
Moreover, for each $1 \leq i_0 \leq k$
the multistep
$\mstepthreeB_{i_0} =
 \lam{\var_1\hdots\var_N}{\mstepthree_{i_0}\,\var_{M+1}\hdots\var_N}$
has the same number of applications as
$\mstepthree_{i_0}\,\var_{M+1}\hdots\var_M$,
which is a subterm of $\mstep$,
and has strictly less applications than $\mstep$.
\end{proof}

\begin{prop}[Generalized \flateqRule{Perm} rule]
\lprop{generalized_flateq_perm}
If $\judgSplit{\mstep}{\mstep_1}{\mstep_2}$
then $\flatten{\mstep} \flateq \flatten{\mstep_1}\seq\flatten{\mstep_2}$.
\\
Note that this generalizes the \flateqRule{Perm} rule,
which requires $\mstep$ to be in $\tof$-normal form.
\end{prop}
\begin{proof}
Suppose  $\judgSplit{\mstep}{\mstep_1}{\mstep_2}$. 
Then 
we have that
$\judgSplit{\flatten{\mstep}}{\mstep'_1}{\mstep'_2}$
such that $\mstep_1 \tofs \mstep'_1$
and $\mstep_2 \tofs \mstep'_2$. This fact may be proved  by induction on the length of a reduction
to normal form $\mstep \tofs \flatten{\mstep}$.
By the \flateqRule{Perm} rule,
$\flatten{\mstep} \flateq \flatten{(\mstep'_1)} \seq \flatten{(\mstep'_2)}$.
Moreover, by strong normalization~(\rprop{flat_sn})
and confluence~(\rprop{flat_confluent}) of flattening,
$\flatten{\mstep_1} = \flatten{(\mstep'_1)}$
and
$\flatten{\mstep_2} = \flatten{(\mstep'_2)}$,
which means that
$\flatten{\mstep} \flateq \flatten{\mstep_1} \seq \flatten{\mstep_2}$.
\end{proof}

\begin{defi}[Splitting, up to flattening]
We write $\judgSplitF{\mstep_1}{\mstep_2}{\mstep_3}$
if there exist multisteps $\mstep'_1,\mstep'_2,\mstep'_3$
such that $\judgSplit{\mstep'_1}{\mstep'_2}{\mstep'_3}$,
and moreover $\flatten{(\mstep'_1)} = \mstep_1$
and $\flatten{(\mstep'_2)} = \mstep_2$
and $\flatten{(\mstep'_3)} = \mstep_3$.
\end{defi}

\begin{lem}[Projection of a splitting over a multistep]
\llem{projection_splitting_over_multistep}
If $\judgSplit{\mstep_1}{\mstep_2}{\mstep_3}$
and $\msteptwo$ is an arbitrary multistep coinitial to $\mstep_1$,
then
$\judgSplitF{
   (\mstep_1/\msteptwo)
 }{
   (\mstep_2/\msteptwo)
 }{
   (\mstep_3/(\msteptwo/\mstep_2))
 }$.
\end{lem}
\begin{proof}
Suppose that
$\judgSplit{\mstep_1}{\mstep_2}{\mstep_3}$.
Recall from~\rprop{body:properties_of_projection}
that $\mstep_1/\msteptwo$ does not depend on the representative
of $\mstep_1$ (up to flattening),
and similarly for $\mstep_2/\msteptwo$ and $\mstep_3/(\msteptwo/\mstep_2)$.
Hence
by \rlem{canonical_splitting} we may assume without loss of generality
that $\mstep_1$ is in $\tofnoeta,\etalong$-normal form.
We may also assume without loss of generality that $\msteptwo$ is in
$\tofnoeta,\etalong$-normal form. 
We show that
  $\judgSplitF{
     (\mstep_1/\msteptwo)
   }{
     (\mstep_2/\msteptwo)
   }{
     (\mstep_3/(\msteptwo/\mstep_2))
   }$
holds,
by induction on the number of applications in $\mstep_1$
and by case analysis on the head of $\mstep_1$, using the fact that it is
in $\tofnoeta,\etalong$-normal form. We omit the cases where $\mstep_1$ headed
by a variable or a constant and present the one where it is headed by a rule
symbol.

  Then $\mstep_1 = \lam{\vec{\var}}{\rulewit\,\mstep_{11}\hdots\mstep_{1n}}$.
  Note that $\judgSplit{\mstep_1}{\mstep_2}{\mstep_3}$
  must be derived by a number of instances of the \indrulename{SAbs} rule,
  followed by $n$ instances of \indrulename{SApp} rule,
  followed by an instance
  of either \indrulename{SRuleL} or \indrulename{SRuleL} at the head.
  We consider two subcases:
  \begin{enumerate}
  \item \indrulename{SRuleL}:
    Then
    $\mstep_2 = \lam{\vec{\var}}{\rulewit\,\mstep_{21}\hdots\mstep_{2n}}$
    and
    $\mstep_3 = \lam{\vec{\var}}{\rtgt{\rulewit}\,\mstep_{31}\hdots\mstep_{3n}}$
    where $\judgSplit{\mstep_{1i}}{\mstep_{2i}}{\mstep_{3i}}$
    for all $1 \leq i \leq n$.
    Since $\mstep_1$ and $\msteptwo$ are coinitial,
    the head of $\msteptwo$ must be a constant or a rule symbol.
    We consider two further subcases:
    \begin{enumerate}
    \item {\bf $\msteptwo$ headed by a constant.}
      \label{projection_splitting_over_multistep__case_ruleL_const}
      By \rlem{coinitial_cons_rule}
      we have that
      $\msteptwo =
       \flattennoeta{(\lam{\vec{\var}}{
         \altsrc{\rulewit}\,\msteptwo_1\hdots\msteptwo_n
       })}$,
      where $\mstep_{1i}$ and $\msteptwo_i$ are coinitial
      for all $1 \leq i \leq n$.
      By~\rprop{body:properties_of_projection} we have that:
      \[
      \begin{array}{rll}
        \mstep_1/\msteptwo
      & = &
        \flatten{(\lam{\var}{
          (\rulewit/\rsrc{\rulewit})
            (\mstep_{11}/\msteptwo_1)
            \hdots
            (\mstep_{1n}/\msteptwo_n)
        })}
      \\
        \mstep_2/\msteptwo
      & = &
        \flatten{(\lam{\var}{
          (\rulewit/\rsrc{\rulewit})
            (\mstep_{21}/\msteptwo_1)
            \hdots
            (\mstep_{2n}/\msteptwo_n)
        })}
      \\
        \mstep_3/(\msteptwo/\mstep_2)
      & = &
        \flatten{(\lam{\var}{
          (\rtgt{\rulewit}/(\rsrc{\rulewit}/\rulewit))
            (\mstep_{31}/(\msteptwo_1/\mstep_{21}))
            \hdots
            (\mstep_{3n}/(\msteptwo_n/\mstep_{2n}))
        })}
      \end{array}
      \]
      By \ih we know that
      $\judgSplitF{
         (\mstep_{1i}/\msteptwo_i)
       }{
         (\mstep_{2i}/\msteptwo_i)
       }{
         (\mstep_{3i}/(\msteptwo_i/\mstep_{2i}))
       }$ for all $1 \leq i \leq n$.
      Moreover,
      note that
      $\rulewit/\fsrc{\rulewit} = \rulewit$
      and
      $\ftgt{\rulewit}/(\fsrc{\rulewit}/\rulewit)
      = \ftgt{\rulewit}/\ftgt{\rulewit}
      = \ftgt{\rulewit}$
      by~\rprop{body:properties_of_projection},
      and that $\judgSplit{\rulewit}{\rulewit}{\rtgt{\rulewit}}$.
      Hence:
      \[
        \judgSplitFTable{
          \lam{\vec{\var}}{
            \rulewit\,
              (\mstep_{11}/\msteptwo_1)
              \hdots
              (\mstep_{1n}/\msteptwo_n)
          }
        }{
          \lam{\vec{\var}}{
            \rulewit\,
              (\mstep_{21}/\msteptwo_1)
              \hdots
              (\mstep_{2n}/\msteptwo_n)
          }
        }{
          \lam{\vec{\var}}{
            \rtgt{\rulewit}\,
              (\mstep_{31}/(\msteptwo_1/\mstep_{21}))
              \hdots
              (\mstep_{3n}/(\msteptwo_n/\mstep_{2n}))
          }
        }
      \]
      This in turn implies that
      $\judgSplitF{
         \mstep_1/\msteptwo
       }{
         \mstep_2/\msteptwo
       }{
         \mstep_3/(\msteptwo/\mstep_1)
       }$.
    \item {\bf $\msteptwo$ headed by a rule symbol.}
      Then by orthogonality
      $\msteptwo = \lam{\vec{\var}}{\rulewit\,\msteptwo_1\hdots\msteptwo_n}$.
      By~\rprop{body:properties_of_projection} we have that:
      \[
      \begin{array}{rll}
        \mstep_1/\msteptwo
      & = &
        \flatten{(\lam{\var}{
          (\rulewit/\rulewit)
            (\mstep_{11}/\msteptwo_1)
            \hdots
            (\mstep_{1n}/\msteptwo_n)
        })}
      \\
        \mstep_2/\msteptwo
      & = &
        \flatten{(\lam{\var}{
          (\rulewit/\rulewit)
            (\mstep_{21}/\msteptwo_1)
            \hdots
            (\mstep_{2n}/\msteptwo_n)
        })}
      \\
        \mstep_3/(\msteptwo/\mstep_2)
      & = &
        \flatten{(\lam{\var}{
          (\rtgt{\rulewit}/(\rulewit/\rulewit))
            (\mstep_{31}/(\msteptwo_1/\mstep_{21}))
            \hdots
            (\mstep_{3n}/(\msteptwo_n/\mstep_{2n}))
        })}
      \end{array}
      \]
      Then the proof proceeds similarly
      as for case~\ref{projection_splitting_over_multistep__case_ruleL_const},
      when the head of $\msteptwo$ is a constant,
      noting that $\rulewit/\rulewit = \ftgt{\rulewit}$
      and that
      $\ftgt{\rulewit}/(\rulewit/\rulewit)
      = \ftgt{\rulewit}/\ftgt{\rulewit}
      = \ftgt{\rulewit}$
      by~\rprop{body:properties_of_projection},
      and that $\judgSplit{\rtgt{\rulewit}}{\rtgt{\rulewit}}{\rtgt{\rulewit}}$.
    \end{enumerate}
  \item \indrulename{SRuleR}:
    Then
    $\mstep_2 = \lam{\vec{\var}}{\rsrc{\rulewit}\,\mstep_{21}\hdots\mstep_{2n}}$
    and
    $\mstep_3 = \lam{\vec{\var}}{\rulewit\,\mstep_{31}\hdots\mstep_{3n}}$
    where $\judgSplit{\mstep_{1i}}{\mstep_{2i}}{\mstep_{3i}}$
    for all $1 \leq i \leq n$.
    Since $\mstep_1$ and $\msteptwo$ are coinitial,
    the head of $\msteptwo$ must be a constant or a rule symbol.
    We consider two further subcases:
    \begin{enumerate}
    \item {\bf $\msteptwo$ headed by a constant.}
      By \rlem{coinitial_cons_rule}
      we have that
      $\msteptwo =
       \flattennoeta{(\lam{\vec{\var}}{
         \altsrc{\rulewit}\,\msteptwo_1\hdots\msteptwo_n
       })}$,
      where $\mstep_{1i}$ and $\msteptwo_i$ are coinitial
      for all $1 \leq i \leq n$. Hence:
      \[
      \begin{array}{rll}
        \mstep_1/\msteptwo
      & = &
        \flatten{(\lam{\var}{
          (\rulewit/\rsrc{\rulewit})\,
            (\mstep_{11}/\msteptwo_1)
            \hdots
            (\mstep_{1n}/\msteptwo_n)
        })}
      \\
        \mstep_2/\msteptwo
      & = &
        \flatten{(\lam{\var}{
          (\rsrc{\rulewit}/\rsrc{\rulewit})\,
            (\mstep_{21}/\msteptwo_1)
            \hdots
            (\mstep_{2n}/\msteptwo_n)
        })}
      \\
        \mstep_3/(\msteptwo/\mstep_2)
      & = &
        \flatten{(\lam{\var}{
          (\rulewit/(\rsrc{\rulewit}/\rsrc{\rulewit}))\,
            (\mstep_{31}/(\msteptwo_1/\mstep_{21}))
            \hdots
            (\mstep_{3n}/(\msteptwo_n/\mstep_{2n}))
        })}
      \end{array}
      \]
      Then the proof proceeds similarly
      as for case~\ref{projection_splitting_over_multistep__case_ruleL_const},
      noting that $\rulewit/\fsrc{\rulewit} = \rulewit$,
      that $\fsrc{\rulewit}/\fsrc{\rulewit} = \fsrc{\rulewit}$,
      and that
      $\rulewit/(\fsrc{\rulewit}/\fsrc{\rulewit})
      = \rulewit/\fsrc{\rulewit}
      = \rulewit$
      by~\rprop{body:properties_of_projection},
      and moreover that
      $\judgSplit{\rulewit}{\rsrc{\rulewit}}{\rulewit}$.
    \item {\bf $\msteptwo$ headed by a rule symbol.}
      Then by orthogonality
      $\msteptwo = \lam{\vec{\var}}{\rulewit\,\msteptwo_1\hdots\msteptwo_n}$.
      Hence:
      \[
      \begin{array}{rll}
        \mstep_1/\msteptwo
      & = &
        \flatten{(\lam{\var}{
          (\rulewit/\rulewit)\,
            (\mstep_{11}/\msteptwo_1)
            \hdots
            (\mstep_{1n}/\msteptwo_n)
        })}
      \\
        \mstep_2/\msteptwo
      & = &
        \flatten{(\lam{\var}{
          (\rsrc{\rulewit}/\rulewit)\,
            (\mstep_{21}/\msteptwo_1)
            \hdots
            (\mstep_{2n}/\msteptwo_n)
        })}
      \\
        \mstep_3/(\msteptwo/\mstep_2)
      & = &
        \flatten{(\lam{\var}{
          (\rulewit/(\rulewit/\rsrc{\rulewit}))\,
            (\mstep_{31}/(\msteptwo_1/\mstep_{21}))
            \hdots
            (\mstep_{3n}/(\msteptwo_n/\mstep_{2n}))
        })}
      \end{array}
      \]
      Then the proof proceeds similarly
      as for case~\ref{projection_splitting_over_multistep__case_ruleL_const},
      noting that $\rulewit/\rulewit = \ftgt{\rulewit}$,
      that $\fsrc{\rulewit}/\rulewit = \ftgt{\rulewit}$,
      and that
      $\rulewit/(\rulewit/\fsrc{\rulewit})
      = \rulewit/\rulewit
      = \ftgt{\rulewit}$
      by~\rprop{body:properties_of_projection},
      and moreover that
      $\judgSplit{\rtgt{\rulewit}}{\rtgt{\rulewit}}{\rtgt{\rulewit}}$.
    \end{enumerate}
  \end{enumerate}
%\end{enumerate}
\end{proof}

\begin{lem}[Projection of splittings and rewrites]
\llem{projection_of_splittings_and_rewrites}
% \llem{projection_splitting_over_rewrite}
% \llem{projection_multistep_over_splitting}
% \llem{projection_rewrite_over_splitting}
Suppose $\judgSplit{\mstep_1}{\mstep_2}{\mstep_3}$. Let $\redseq$ be an arbitrary flat rewrite coinitial to $\mstep_1$ and let $\msteptwo$ be an arbitrary multistep coinitial to $\mstep_1$. Then:
\begin{enumerate}
  \item 
$\judgSplitF{
   (\mstep_1/\redseq)
 }{
   (\mstep_2/\redseq)
 }{
   (\mstep_3/(\redseq/\mstep_2))
 }$.
\item 
$\msteptwo/\mstep_1 = (\msteptwo/\mstep_2)/\mstep_3$
  \item
$\redseq/\mstep_1 = (\redseq/\mstep_2)/\mstep_3$.

 \end{enumerate}
\end{lem}
\begin{proof}
The proof of the first item is by induction on $\redseq$. The key case is when $\redseq$ is a multistep (the base case), in which we resort to  \rlem{projection_splitting_over_multistep}. For the second item, suppose that $\judgSplit{\mstep_1}{\mstep_2}{\mstep_3}$.
Recall from~\rprop{body:properties_of_projection}
that $\msteptwo/\mstep_1$ does not depend on the representative
of $\mstep_1$ (up to flattening),
and similarly for $(\mstep/\mstep_2)/\mstep_3$.
Hence
by \rlem{canonical_splitting} we may assume without loss of generality
that $\mstep_1$ is in $\tofnoeta,\etalong$-normal form.
We may also assume without loss of generality that $\msteptwo$ is in
$\tofnoeta,\etalong$-normal form. 
We proceed by induction on the number of applications in $\mstep_1$
and by case analysis on the head of $\mstep_1$, using the fact that it is
in $\tofnoeta,\etalong$-normal form. The case where $\mstep_1$ headed by a variable relies on~\rprop{body:properties_of_projection}. The case where it is headed by a 
  constant additionally relies on \rlem{coinitial_cons_rule} and the \ih.  The case where it is headed by a rule symbol, in addition it makes use of orthogonality of the HRS.

For the third item, we proceed by induction on $\redseq$.
\begin{enumerate}
\item {\bf Multistep, $\redseq = \msteptwo$.}
  Then this is an immediate consequence of
  \rlem{projection_of_splittings_and_rewrites}(2) %\rlem{projection_multistep_over_splitting}.
\item {\bf Composition, $\redseq = \redseq_1\seq\redseq_2$.}
  Note by \rlem{projection_of_splittings_and_rewrites}(1) %\rlem{projection_splitting_over_rewrite}
  that
  $\judgSplitF{
     (\mstep_1/\redseq_1)
   }{
     (\mstep_2/\redseq_1)
   }{
     (\mstep_3/(\redseq_1/\mstep_2))
   }$
  holds, that is,
  there exist multisteps such that
  $\judgSplit{\mstepB_1}{\mstepB_2}{\mstepB_3}$
  and such that
  $\flatten{\mstepB_1} = \mstep_1/\redseq_1$
  and
  $\flatten{\mstepB_2} = \mstep_2/\redseq_1$
  and
  $\flatten{\mstepB_3} = \mstep_3/(\redseq_1/\mstep_2)$.
  Then:
  \[
    \begin{array}{rlll}
      (\redseq_1\seq\redseq_2)/\mstep_1
    & = &
      (\redseq_1/\mstep_1)\seq(\redseq_2/(\mstep_1/\redseq_1))
    \\
    & = &
      ((\redseq_1/\mstep_2)/\mstep_3)\seq(\redseq_2/(\mstep_1/\redseq_1))
      & \text{by \ih}
    \\
    & = &
      ((\redseq_1/\mstep_2)/\mstep_3)\seq(\redseq_2/\mstepB_1)
      & \text{by~\rprop{body:properties_of_projection}}
    \\
    & = &
      ((\redseq_1/\mstep_2)/\mstep_3)\seq((\redseq_2/\mstepB_2)/\mstepB_3)
      & \text{by \ih}
    \\
    & = &
      ((\redseq_1/\mstep_2)/\mstep_3)\seq
      ((\redseq_2/(\mstep_2/\redseq_1))/(\mstep_3/(\redseq_1/\mstep_2)))
      & \text{by~\rprop{body:properties_of_projection}}
    \\
    & = &
      ((\redseq_1/\mstep_2)\seq(\redseq_2/(\mstep_2/\redseq_1)))
      /\mstep_3
    \\
    & = &
      ((\redseq_1\seq\redseq_2)/\mstep_2)/\mstep_3
    \end{array}
  \]
\end{enumerate}
\end{proof}

Given the above, From this point on we can safely overload $\redseq/\redseqtwo$
to stand for
$\redseq \projmr \redseqtwo$ (if $\redseq$ is a multistep),
$\redseq \projrm \redseqtwo$ (if $\redseqtwo$ is a multistep),
or $\redseq \projrr \redseqtwo$ (in the general case).
%Note that \rlem{projrr_generalizes_projrm_and_projmr}
%and \rlem{projrr_sequence} ensure that this overloading
%is ``safe''.
%\end{convention}

\begin{prop}[Congruence of $\flateq$ with respect to projection]
\lprop{congruence_flateq_projection}
Let $\redseq \flateq \redseqtwo$,
and let $\redseqthree$ be an arbitrary flat rewrite
coinitial to $\redseq$.
Then:
\begin{enumerate}
\item $\redseqthree/\redseq = \redseqthree/\redseqtwo$
\item $\redseq/\redseqthree \flateq \redseqtwo/\redseqthree$
\end{enumerate}
\end{prop}
\begin{proof}
Both items are induction on the derivation of $\redseq \flateq \redseqtwo$.
  The reflexivity, symmetry, and transitivity cases are immediate.
  We analyze the cases in which an axiom is applied at the root,
  as well as congruence closure below composition contexts.
\begin{enumerate}
\item First item. The interesting case is when $\redseq \flateq \redseqtwo$ follows from \flateqRule{Perm}. 
  % \begin{enumerate}
  % \item {\bf Rule \flateqRule{Assoc}.}
  %   Let $(\redseq_1\seq\redseq_2)\seq\redseq_3 \flateq
  %        \redseq_1\seq(\redseq_2\seq\redseq_3)$.
  %   Then:
  %   \[
  %     \redseqthree/((\redseq_1\seq\redseq_2)\seq\redseq_3)
  %     =
  %     ((\redseqthree/\redseq_1)/\redseq_2)/\redseq_3
  %     =
  %     \redseqthree/(\redseq_1\seq(\redseq_2\seq\redseq_3))
  %   \]
  % \item {\bf Rule \flateqRule{Perm}.}
    Let $\mstep_1 \flateq \flatten{\mstep_2}\seq\flatten{\mstep_3}$
    be derived from $\judgSplit{\mstep_1}{\mstep_2}{\mstep_3}$.
    Then:
    \[
      \begin{array}{rlll}
        \redseqthree/\mstep_1
      & = & 
        (\redseqthree/\mstep_2)/\mstep_3
        & \text{by \rlem{projection_of_splittings_and_rewrites}(3)} %\rlem{projection_rewrite_over_splitting}}
      \\
      & = &
        (\redseqthree/\flatten{\mstep_2})/\flatten{\mstep_3}
        & \text{by~\rprop{body:properties_of_projection}}
      \end{array}
    \]
  % \item {\bf Congruence, left of a composition.}
  %   Let $\redseq\seq\redseqfour \flateq \redseqtwo\seq\redseqfour$
  %   be derived from $\redseq \flateq \redseqtwo$.
  %   Then:
  %   \[
  %     \begin{array}{rlll}
  %       \redseqthree/(\redseq\seq\redseqfour)
  %     & = &
  %       (\redseqthree/\redseq)/\redseqfour
  %     \\
  %     & = &
  %       (\redseqthree/\redseqtwo)/\redseqfour
  %       & \text{by \ih}
  %     \\
  %     & = &
  %       \redseqthree/(\redseqtwo\seq\redseqfour)
  %     \end{array}
  %   \]
  % \item {\bf Congruence, right of a composition.}
  %   Let $\redseqfour\seq\redseq \flateq \redseqfour\seq\redseqtwo$
  %   be derived from $\redseq \flateq \redseqtwo$.
  %   Then:
  %   \[
  %     \begin{array}{rlll}
  %       \redseqthree/(\redseqfour\seq\redseq)
  %     & = &
  %       (\redseqthree/\redseqfour)/\redseq
  %     \\
  %     & = &
  %       (\redseqthree/\redseqfour)/\redseqtwo
  %       & \text{by \ih}
  %     \\
  %     & = &
  %       \redseqthree/(\redseqfour\seq\redseqtwo)
  %     \end{array}
  %   \]
  % \end{enumerate}
  \item Second item. The case where $\redseq \flateq \redseqtwo$ follows from \flateqRule{Assoc} follows from \rlem{projrr_sequence} and \flateqRule{Assoc}  itself. Similarly, the case where $\redseq \flateq \redseqtwo$ follows from congruence of composition (left or right) follows from \rlem{projrr_sequence}, the \ih and 
                item~1. of this proposition. The interesting case, once again, is when $\redseq \flateq \redseqtwo$ follows from 
    \flateqRule{Perm}.
  % \begin{enumerate}
  % \item {\bf Rule \flateqRule{Assoc}.}
  %   Let $(\redseq_1\seq\redseq_2)\seq\redseq_3 \flateq
  %        \redseq_1\seq(\redseq_2\seq\redseq_3)$.
  %   Then:
  %   \[
  %     \begin{array}{rlll}
  %       ((\redseq_1\seq\redseq_2)\seq\redseq_3)/\redseqthree
  %     & = &
  %       ((\redseq_1\seq\redseq_2)/\redseqthree)\seq
  %       (\redseq_3/(\redseqthree/(\redseq_1\seq\redseq_2)))
  %       & \text{by \rlem{projrr_sequence}}
  %     \\
  %     & = &
  %       (
  %         (\redseq_1/\redseqthree)
  %       \seq
  %         (\redseq_2/(\redseqthree/\redseq_1))
  %       )
  %       \seq
  %       (\redseq_3/(\redseqthree/(\redseq_1\seq\redseq_2)))
  %       & \text{by \rlem{projrr_sequence}}
  %     \\
  %     & \flateq &
  %       (\redseq_1/\redseqthree)
  %       \seq
  %       (
  %         (\redseq_2/(\redseqthree/\redseq_1))
  %       \seq
  %         (\redseq_3/(\redseqthree/(\redseq_1\seq\redseq_2)))
  %       )
  %       & \text{by \flateqRule{Assoc}}
  %     \\
  %     & = &
  %       (\redseq_1/\redseqthree)
  %       \seq
  %       (
  %         (\redseq_2/(\redseqthree/\redseq_1))
  %       \seq
  %         (\redseq_3/((\redseqthree/\redseq_1)/\redseq_2)))
  %       )
  %     \\
  %     & = &
  %       (\redseq_1/\redseqthree)
  %       \seq
  %       ((\redseq_2\seq\redseq_3)/(\redseqthree/\redseq_1))
  %       & \text{by \rlem{projrr_sequence}}
  %     \\
  %     & = &
  %       (\redseq_1\seq(\redseq_2\seq\redseq_3))/\redseqthree
  %       & \text{by \rlem{projrr_sequence}}
  %     \end{array}
  %   \]
  % \item {\bf Rule \flateqRule{Perm}.}
    Let $\mstep_1 \flateq \flatten{\mstep_2}\seq\flatten{\mstep_3}$
    be derived from $\judgSplit{\mstep_1}{\mstep_2}{\mstep_3}$.
    Then by \rlem{projection_of_splittings_and_rewrites}(1) %\rlem{projection_splitting_over_rewrite} 
    we have that
    $\judgSplitF{
       (\mstep_1/\redseqthree)
     }{
       (\mstep_2/\redseqthree)
     }{
       (\mstep_3/(\redseqthree/\mstep_2))
     }$,
    that is, there exist multisteps such that
    $\judgSplit{\mstepB_1}{\mstepB_2}{\mstepB_3}$
    and such that
    $\flatten{\mstepB_1} = \mstep_1/\redseqthree$
    and
    $\flatten{\mstepB_2} = \mstep_2/\redseqthree$
    and
    $\flatten{\mstepB_3} = \mstep_3/(\redseqthree/\mstep_2)$.
    Hence:
    \[
      \begin{array}{rlll}
        \mstep_1/\redseqthree
      & = &
        \flatten{\mstepB_1}
      \\
      & \flateq &
        \flatten{\mstepB_1}\seq\flatten{\mstepB_2}
        & \text{by~\rprop{generalized_flateq_perm},
                as $\judgSplit{\mstepB_1}{\mstepB_2}{\mstepB_3}$}
      \\
      & = &
        (\mstep_2/\redseqthree)\seq(\mstep_3/(\redseqthree/\mstep_2))
      \\
      & = &
        (\flatten{\mstep_2}/\redseqthree)\seq
        (\flatten{\mstep_3}/(\redseqthree/\mstep_2))
        & \text{by \rprop{body:properties_of_projection}}
      \\
      & = &
        (\flatten{\mstep_2}\seq\flatten{\mstep_3})/\redseqthree
      \end{array}
    \]
  % \item {\bf Congruence, left of a composition.}
  %   Let $\redseq\seq\redseqfour \flateq \redseqtwo\seq\redseqfour$
  %   be derived from $\redseq \flateq \redseqtwo$.
  %   Then:
  %   \[
  %     \begin{array}{rlll}
  %       (\redseq\seq\redseqfour)/\redseqthree
  %     & = &
  %       (\redseq/\redseqthree)\seq(\redseqfour/(\redseqthree/\redseq))
  %       & \text{by \rlem{projrr_sequence}}
  %     \\
  %     & \flateq &
  %       (\redseqtwo/\redseqthree)\seq(\redseqfour/(\redseqthree/\redseq))
  %       & \text{by \ih}
  %     \\
  %     & = &
  %       (\redseqtwo/\redseqthree)\seq(\redseqfour/(\redseqthree/\redseqtwo))
  %       & \text{as $\redseqthree/\redseq = \redseqthree/\redseqtwo$,
  %               by item~1. of this proposition}
  %     \\
  %     & = &
  %       (\redseqtwo\seq\redseqfour)/\redseqthree
  %       & \text{by \rlem{projrr_sequence}}
  %     \end{array}
  %   \]
  % \item {\bf Congruence, right of a composition.}
  %   Let $\redseqfour\seq\redseq \flateq \redseqfour\seq\redseqtwo$
  %   be derived from $\redseq \flateq \redseqtwo$.
  %   Then:
  %   \[
  %     \begin{array}{rlll}
  %       (\redseqfour\seq\redseq)/\redseqthree
  %     & = &
  %       (\redseqfour/\redseqthree)\seq(\redseq/(\redseqthree/\redseqfour))
  %       & \text{by \rlem{projrr_sequence}}
  %     \\
  %     & \flateq &
  %       (\redseqfour/\redseqthree)\seq(\redseqtwo/(\redseqthree/\redseqfour))
  %       & \text{by \ih}
  %     \\
  %     & = &
  %       (\redseqfour\seq\redseqtwo)/\redseqthree
  %       & \text{by \rlem{projrr_sequence}}
  %     \end{array}
  %   \]
  % \end{enumerate}
\end{enumerate}
\end{proof}

The projection of a rewrite over itself is always empty;
specifically $\redseq/\redseq \flateq \ftgt{\redseq}$ where, recall from above, $\ftgt{\redseq}=\ftgtb{\redseq}$.

\begin{lem}[Self-erasure]
\llem{self_erasure}
$\redseq/\redseq \flateq \ftgt{\redseq}$
\end{lem}
\begin{proof}
By induction on $\redseq$.
If $\redseq = \mstep$ is a multistep, then $\mstep/\mstep = \ftgt{\mstep}$
by~\rprop{body:properties_of_projection}.
If $\redseq = \redseq_1\seq\redseq_2$ then:
\[
\begin{array}{rlll}
&&
  (\redseq_1\seq\redseq_2)/(\redseq_1\seq\redseq_2)
\\
& = &
  ((\redseq_1\seq\redseq_2)/\redseq_1)/\redseq_2
\\
& = &
  ((\redseq_1/\redseq_1)
   \seq
   (\redseq_2/(\redseq_1/\redseq_1)))/\redseq_2
  & \text{by \rlem{projrr_sequence}}
\\
& = &
  ((\redseq_1/\redseq_1)/\redseq_2)
  \seq
  ((\redseq_2/(\redseq_1/\redseq_1))/(\redseq_2/(\redseq_1/\redseq_1)))
  & \text{by \rlem{projrr_sequence}}
\\
& \flateq &
  (\ftgt{\redseq_1}/\redseq_2)
  \seq
  ((\redseq_2/(\redseq_1/\redseq_1))/(\redseq_2/(\redseq_1/\redseq_1)))
  & \text{by \ih and \rprop{congruence_flateq_projection}}
\\
& = &
  (\fsrc{\redseq_2}/\redseq_2)
  \seq
  ((\redseq_2/(\redseq_1/\redseq_1))/(\redseq_2/(\redseq_1/\redseq_1)))
  & \text{as $\rtgt{\redseq_1} \termeq \rsrc{\redseq_2}$}
\\
& = &
  \ftgt{\redseq_2}
  \seq
  ((\redseq_2/(\redseq_1/\redseq_1))/(\redseq_2/(\redseq_1/\redseq_1)))
  & \text{by \rlem{basic_properties_of_rewrite_projection}}
\\
& = &
  (\redseq_2/(\redseq_1/\redseq_1))/(\redseq_2/(\redseq_1/\redseq_1))
  & \text{by \permeqRule{IdL} and \rthm{body:soundness_completeness_flateq}}
\\
& = &
  (\redseq_2/\ftgt{\redseq_1})/(\redseq_2/\ftgt{\redseq_1})
  & \text{by \ih and \rprop{congruence_flateq_projection}}
\\
& = &
  (\redseq_2/\fsrc{\redseq_2})/(\redseq_2/\fsrc{\redseq_2})
  & \text{as $\rtgt{\redseq_1} \termeq \rsrc{\redseq_2}$}
\\
& = &
  \redseq_2/\redseq_2
  & \text{by \rlem{basic_properties_of_rewrite_projection}}
\\
& \flateq &
  \ftgt{\redseq_2}
  & \text{by \ih}
\\
& = &
  \ftgt{(\redseq_1\seq\redseq_2)}
  &
\end{array}
\]
\end{proof}

Finally, an important property is that
$\redseq\seq(\redseqtwo/\redseq) \flateq
 \redseqtwo\seq(\redseq/\redseqtwo)$,
 corresponding to a strong form of confluence. Its proof makes use of~\rlem{multistep_permutation} below.

\begin{lem}[Multistep permutation]
\llem{multistep_permutation}
Let $\mstep,\msteptwo$ be coinitial multisteps.
Then:
\[
  \mstep\seq(\msteptwo/\mstep) \permeq \msteptwo\seq(\mstep/\msteptwo)
\]
\end{lem}
\begin{proof}
Recall that if $\flatten{\mstep} = \flatten{(\mstep')}$
and $\flatten{\msteptwo} = \flatten{(\msteptwo')}$,
then $\mstep/\msteptwo = \mstep'/\msteptwo'$
by~\rprop{body:properties_of_projection}.
Moreover, by~\rlem{flattening_sound_wrt_permeq}
if $\mstep$ is a step and $\mstep'$ is its $\tofnoeta,\etalong$-normal
form, then $\mstep \permeq \mstep'$.
Hence, to prove the statement of the lemma,
it suffices to show that
$\mstep\seq(\msteptwo/\mstep)
 \permeq
 \msteptwo\seq(\mstep/\msteptwo)$
assuming that $\mstep,\msteptwo$ are coinitial multisteps
in $\tofnoeta,\etalong$-normal form.

The proof proceeds by induction on the number of applications
in $\mstep$, and by case analysis on the head of $\mstep$.
We provide a sample case, namely when $\mstep$ headed by a constant. The others
are developed similarly. The case where $\mstep$ is headed by a rule symbol
makes use of the fact that the HRS is orthogonal.

Suppose, therefore, that the $\mstep$ headed by a constant.
  Then $\mstep = \lam{\vec{\var}}{\cons\,\mstep_1\hdots\mstep_n}$.
  Since $\mstep$ and $\msteptwo$ are coinitial, the head of $\msteptwo$
  can be either $\cons$ or a rule symbol $\rulewit$.
  We consider two subcases:
  \begin{enumerate}
  \item {\bf $\msteptwo$ headed by a constant.}
    Then $\msteptwo = \lam{\vec{\var}}{\cons\,\msteptwo_1\hdots\msteptwo_n}$.
    % The proof of this case proceeds similarly as
    % for case~\ref{multistep_permutation__case_var},
    % when the heads of $\mstep$ and $\msteptwo$ are both variables.
  %     Since $\mstep$ and $\msteptwo$ are coinitial,
  % the head of $\msteptwo$ cannot be a variable or a rule symbol,
  % and in fact it must also be $\cons$,
  % that is
  % $\msteptwo = \lam{\vec{\var}}{\cons\,\msteptwo_1\hdots\msteptwo_n}$.
  Note that:
  \[
    \begin{array}{rlll}
    &&
      \mstep\seq(\msteptwo/\mstep)
    \\
    & = &
      (\lam{\vec{\var}}{\cons\,\mstep_1\hdots\mstep_n})
      \seq
      (
        (\lam{\vec{\var}}{\cons\,\msteptwo_1\hdots\msteptwo_n})
      /
        (\lam{\vec{\var}}{\cons\,\mstep_1\hdots\mstep_n})
      )
    \\
    & = &
      (\lam{\vec{\var}}{\cons\,\mstep_1\hdots\mstep_n})
      \seq
      \flatten{(\lam{\vec{\var}}{
          \cons\,
            (\msteptwo_1/\mstep_1)
            \hdots
            (\msteptwo_n/\mstep_n)
      })}
      & \text{by~\rprop{body:properties_of_projection}}
    \\
    & \permeq &
      (\lam{\vec{\var}}{\cons\,\mstep_1\hdots\mstep_n})
      \seq
      (\lam{\vec{\var}}{
          \cons\,
            (\msteptwo_1/\mstep_1)
            \hdots
            (\msteptwo_n/\mstep_n)
      })
      & \text{by~\rlem{flattening_sound_wrt_permeq}}
    \\
    & \permeq &
      \lam{\vec{\var}}{
        \cons\,
          (\mstep_1\seq(\msteptwo_1/\mstep_1))
          \hdots
          (\mstep_n\seq(\msteptwo_n/\mstep_n))
      }
      & \text{by~\permeqRule{Abs}, \permeqRule{App}}
    \end{array}
  \]
  Similarly:
  \[
    \msteptwo\seq(\mstep/\msteptwo)
    \permeq
    \lam{\vec{\var}}{
      \cons\,
        (\msteptwo_1\seq(\mstep_1/\msteptwo_1))
        \hdots
        (\msteptwo_n\seq(\mstep_n/\msteptwo_n))
    }
  \]
  Moreover, by \ih we have that
  $\mstep_i\seq(\msteptwo_i/\mstep_i) \permeq
   \msteptwo_i\seq(\mstep_i/\msteptwo_i)$
  for all $1 \leq i \leq n$, so:
  \[
    \lam{\vec{\var}}{
      \cons\,
        (\mstep_1\seq(\msteptwo_1/\mstep_1))
        \hdots
        (\mstep_n\seq(\msteptwo_n/\mstep_n))
    }
    \permeq
    \lam{\vec{\var}}{
      \cons\,
        (\msteptwo_1\seq(\mstep_1/\msteptwo_1))
        \hdots
        (\msteptwo_n\seq(\mstep_n/\msteptwo_n))
    }
    \]
    
  \item {\bf $\msteptwo$ headed by a rule symbol.}
    \label{multistep_permutation__case_cons_rule}
    Then $\msteptwo = \lam{\vec{\var}}{\rulewit\,\msteptwo_1\hdots\msteptwo_q}$.
    By \rlem{coinitial_cons_rule}
    we have that
    $\mstep_1 = \flattennoeta{(
                  \lam{\vec{\var}}{\altsrc{\rulewit}\,\mstep'_1\hdots\mstep'_q
                })}$
    where $\mstep'_i$ and $\msteptwo_i$ are coinitial
    for all $1 \leq i \leq q$,
    and each $\mstep'_i$ has strictly less applications than $\mstep$.
    Note that:
    \[
      \begin{array}{rlll}
      &&
        \mstep\seq(\msteptwo/\mstep)
      \\
      & = &
        \flattennoeta{
          (\lam{\vec{\var}}{\altsrc{\rulewit}\,\mstep'_1\hdots\mstep'_q})
        }
        \seq
        (
          (\lam{\vec{\var}}{\rulewit\,\msteptwo_1\hdots\msteptwo_q})
        /
          \flattennoeta{
            (\lam{\vec{\var}}{\altsrc{\rulewit}\,\mstep'_1\hdots\mstep'_q})
          }
        )
      \\
      & = &
        \flattennoeta{
          (\lam{\vec{\var}}{\altsrc{\rulewit}\,\mstep'_1\hdots\mstep'_q})
        }
        \seq
        \flatten{(\lam{\vec{\var}}{
           (\rulewit/\rsrc{\rulewit})
             (\msteptwo_1/\mstep'_1)
             \hdots
             (\msteptwo_q/\mstep'_q)
        })}
        & \text{by~\rprop{body:properties_of_projection}}
      \\
      & \permeq &
        (\lam{\vec{\var}}{\rsrc{\rulewit}\,\mstep'_1\hdots\mstep'_q})
        \seq
        (\lam{\vec{\var}}{
           (\rulewit/\rsrc{\rulewit})
             (\msteptwo_1/\mstep'_1)
             \hdots
             (\msteptwo_q/\mstep'_q)
        })
        & \text{by~\rlem{flattening_sound_wrt_permeq}}
      \\
      & \permeq &
        (\lam{\vec{\var}}{
           (\rsrc{\rulewit}\seq(\rulewit/\rsrc{\rulewit}))
             (\mstep'_1\seq(\msteptwo_1/\mstep'_1))
             \hdots
             (\mstep'_q\seq(\msteptwo_q/\mstep'_q))
        })
        & \text{by~\permeqRule{Abs}, \permeqRule{App}}
      \end{array}
    \]
    Similarly:
    \[
      \msteptwo\seq(\mstep/\msteptwo)
      \permeq
      (\lam{\vec{\var}}{
         (\rulewit\seq(\rsrc{\rulewit}/\rulewit))
           (\msteptwo_1\seq(\mstep'_1/\msteptwo_1))
           \hdots
           (\msteptwo_q\seq(\mstep'_q/\msteptwo_q))
      })
    \]
    Note that:
    \[
      \begin{array}{rlll}
        \rsrc{\rulewit}\seq(\rulewit/\rsrc{\rulewit})
      & = &
        \rsrc{\rulewit}\seq\rulewit
        & \text{by~\rprop{body:properties_of_projection}}
      \\
      & = &
        \rulewit
        & \text{by~\permeqRule{IdL}}
      \\
      & = &
        \rulewit\seq\rtgt{\rulewit}
        & \text{by~\permeqRule{IdR}}
      \\
      & \permeq &
        \rulewit\seq\ftgt{\rulewit}
        & \text{by~\rlem{flattening_sound_wrt_permeq}}
      \\
      & = &
         \rulewit\seq(\rsrc{\rulewit}/\rulewit)
      \end{array}
    \]
    Moreover, by \ih we know that
    $\mstep'_i\seq(\msteptwo_i/\mstep'_i) \permeq
     \msteptwo_i\seq(\mstep'_i/\msteptwo_i)$
    for all $1 \leq i \leq q$.
    Hence:
    \[
      \begin{array}{rl}
      &
      \lam{\vec{\var}}{
         (\rsrc{\rulewit}\seq(\rulewit/\rsrc{\rulewit}))
           (\mstep'_1\seq(\msteptwo_1/\mstep'_1))
           \hdots
           (\mstep'_q\seq(\msteptwo_q/\mstep'_q))
      }
      \\
      \permeq &
      \lam{\vec{\var}}{
         (\rulewit\seq(\rsrc{\rulewit}/\rulewit))
           (\msteptwo_1\seq(\mstep'_1/\msteptwo_1))
           \hdots
           (\msteptwo_q\seq(\mstep'_q/\msteptwo_q))
      }
      \end{array}
    \]
  \end{enumerate}
\end{proof}

\begin{lem}[Rewrite permutation]
\llem{rewrite_permutation}
If $\redseq,\redseqtwo$ are coinitial flat rewrites then:
\[
  \redseq\seq(\redseqtwo/\redseq)
  \flateq
  \redseqtwo\seq(\redseq/\redseqtwo)
\]
\end{lem}
\begin{proof}
We proceed by induction on $\redseq$:
\begin{enumerate}
\item {\bf Multistep, $\redseq = \mstep$.}
  To prove
    $\mstep\seq(\redseqtwo/\mstep)
    \flateq
    \redseqtwo\seq(\mstep/\redseqtwo)$
  we proceed by a nested induction on $\redseqtwo$:
  \begin{enumerate}
  \item {\bf Multistep, $\redseqtwo = \msteptwo$.}
    By \rlem{multistep_permutation}
    we have that
    $\mstep\seq(\msteptwo/\mstep) \permeq \msteptwo\seq(\mstep/\msteptwo)$.
    By~\rthm{body:soundness_completeness_flateq}
    this implies that
    $\flatten{\mstep}\seq\flatten{(\msteptwo/\mstep)}
     \flateq
     \flatten{\msteptwo}\seq\flatten{(\mstep/\msteptwo)}$.
    But $\mstep$ and $\msteptwo$ are flat,
    and by \rprop{body:properties_of_projection} we know that
    $\mstep/\msteptwo = \flatten{(\mstep/\msteptwo)}$
    and
    $\msteptwo/\mstep = \flatten{(\msteptwo/\mstep)}$.
    Hence
    $\mstep\seq(\msteptwo/\mstep) \flateq \msteptwo\seq(\mstep/\msteptwo)$.
  \item {\bf Composition, $\redseqtwo = \redseqtwo_1\seq\redseqtwo_2$.}
    To alleviate the notation,
    we work implicitly modulo the \flateqRule{Assoc} rule.
    Note that:
    \[
      \begin{array}{rcll}
        \mstep\seq((\redseqtwo_1\seq\redseqtwo_2)/\mstep)
      & = &
        \mstep\seq(\redseqtwo_1/\mstep)\seq(\redseqtwo_2/(\mstep/\redseqtwo_1))
      \\
      & \flateq &
        \redseqtwo_1\seq(\mstep/\redseqtwo_1)
        \seq(\redseqtwo_2/(\mstep/\redseqtwo_1))
        & \text{by \ih}
      \\
      & \flateq &
        \redseqtwo_1\seq
        \redseqtwo_2\seq
        ((\mstep/\redseqtwo_1)/\redseqtwo_2)
        & \text{by \ih}
      \\
      & = &
        \redseqtwo_1\seq
        \redseqtwo_2\seq
        (\mstep/(\redseqtwo_1\seq\redseqtwo_2))
      \end{array}
    \]
  \end{enumerate}
\item {\bf Composition, $\redseq = \redseq_1\seq\redseq_2$.}
  To alleviate the notation,
  we work implicitly modulo the \flateqRule{Assoc} rule.
  Note that:
  \[
    \begin{array}{rlll}
      \redseq_1\seq\redseq_2\seq(\redseqtwo/(\redseq_1\seq\redseq_2))
    & = &
      \redseq_1\seq\redseq_2\seq((\redseqtwo/\redseq_1)/\redseq_2)
    \\
    & \flateq &
      \redseq_1\seq(\redseqtwo/\redseq_1)\seq(\redseq_2/(\redseqtwo/\redseq_1))
      & \text{by \ih}
    \\
    & \flateq &
      \redseqtwo\seq(\redseq_1/\redseqtwo)\seq(\redseq_2/(\redseqtwo/\redseq_1))
      & \text{by \ih}
    \\
    & = &
      \redseqtwo\seq((\redseq_1\seq\redseq_2)/\redseqtwo)
      & \text{by \rlem{projrr_sequence}}
    \end{array}
  \]
\end{enumerate}
\end{proof}

% The proof of these properties is technical,
% by induction on the structure of the rewrites.
% We do not develop the full theory of projection for flat rewrites
% here for lack of space
% % (\SeeAppendix{but see \rsec{appendix:projection} in~\cite{techReport} 
% % for more details}).
% \SeeAppendix{(see Section E in~\cite{techReport} 
% for more details)}.

\subsection{Projection for Arbitrary Rewrites.}
As a final step, the projection operator of~\rdef{projection_for_rewrites}
may be extended to arbitrary rewrites by flattening first.
% \SeeAppendix{The proof of \rprop{body:properties_of_arbitrary_projection}
% relies crucially on the properties of projection for flat rewrites
% and on \rthm{body:soundness_completeness_flateq};
% it may be found in \rsec{appendix:generalized_projection} in~\cite{techReport}.}
% \SeeAppendix{The proof of \rprop{body:properties_of_arbitrary_projection}
% relies crucially on the properties of projection for flat rewrites
% and on \rthm{body:soundness_completeness_flateq};
% it may be found in Section G in~\cite{techReport}.}

\begin{defi}[Projection operator for arbitrary rewrites]
\ldef{projection_for_arbitrary_rewrites}
Let $\redseq,\redseqtwo$ be arbitrary coinitial rewrites.
Their projection is defined as
$\redseq\proja\redseqtwo \eqdef \flatten{\redseq}/\flatten{\redseqtwo}$.
\end{defi}

We next address some basic properties of projection over arbitrary rewrites (\rprop{body:properties_of_arbitrary_projection}) including that permutation equivalence is a congruence with respect to it. Since $\redseq\proja\redseqtwo $ is defined in terms of projection over flat rewrites, we first prove some auxiliary results on that operator, namely \rlem{proja_abstraction} and \rlem{proja_application}.

\begin{lem}[Projection of abstraction]
\llem{proja_abstraction}
The following hold:
\begin{enumerate}
\item
  If $\mstep$ is a flat multistep and $\redseq$ is a coinitial flat rewrite,
  then
  $\flatten{(\lam{\var}{\mstep})} \projmr \flatten{(\lam{\var}{\redseq})} = \flatten{(\lam{\var}{(\mstep\projmr\redseq))}}$.
\item
  If $\redseq$ is a flat rewrite and $\mstep$ is a coinitial flat multistep,
  then
  $\flatten{(\lam{\var}{\redseq})} \projrm \flatten{(\lam{\var}{\mstep})} = \flatten{(\lam{\var}{(\redseq\projrm\mstep))}}$.
\item
  If $\redseq,\redseqtwo$ are coinitial flat rewrites,
  then
  $\flatten{(\lam{\var}{\redseq})} \projrr \flatten{(\lam{\var}{\redseqtwo})} = \flatten{(\lam{\var}{(\redseq\projrr\redseqtwo))}}$.
\end{enumerate}
\end{lem}
\begin{proof}
The first item is by  induction on $\redseq$ using \rprop{body:properties_of_projection} and ~\rprop{flat_confluent}. The second is similar and also uses the first item. The third is by induction on $\redseqtwo$ and uses  ~\rprop{flat_confluent} and the second item.
\end{proof}

\begin{lem}[Projection of application]
\llem{proja_application}
The following hold:
\begin{enumerate}
\item
  If $\mstep_1,\mstep_2$ are flat multisteps and
  $\redseq_1,\redseq_2$ are flat rewrites
  such that $\mstep_1$ and $\redseq_1$ are coinitial
  and $\mstep_2$ and $\redseq_2$ are coinitial,
  then
  $\flatten{(\mstep_1\,\mstep_2)} \projmr \flatten{(\redseq_1\,\redseq_2)}
   = \flatten{((\mstep_1\projmr\redseq_1)\,(\mstep_2\projmr\redseq_2))}$.
\item
  If $\redseq_1,\redseq_2$ are flat rewrites and
  $\mstep_1,\mstep_2$ are flat multisteps
  such that $\redseq_1$ and $\mstep_1$ are coinitial
  and $\redseq_2$ and $\mstep_2$ are coinitial,
  then
  $\flatten{(\redseq_1\,\redseq_2)} \projrm \flatten{(\mstep_1\,\mstep_2)}
   = \flatten{((\redseq_1\projrm\mstep_1)\,(\redseq_2\projrm\mstep_2))}$.
\item
  If $\redseq_1,\redseq_2,\redseqtwo_1,\redseqtwo_2$ are flat rewrites
  such that $\redseq_1$ and $\redseqtwo_1$ are coinitial
  and $\redseq_2$ and $\redseqtwo_2$ are coinitial,
  then
  $\flatten{(\redseq_1\,\redseq_2)} \projrr \flatten{(\redseqtwo_1\,\redseqtwo_2)}
   = \flatten{((\redseq_1\projrr\redseqtwo_1)\,(\redseq_2\projrr\redseqtwo_2))}$.
\end{enumerate}
\end{lem}
\begin{proof}
All items use the fact that $\tof$ is strongly normalizing~(\rprop{flat_sn}),
and proceed by induction on the length of the longest reduction of their target. We present the proof of the first item as a sample.
% We prove each item separately:  
% \begin{enumerate}
% \item
%  Using the fact that $\tof$ is strongly normalizing~(\rprop{flat_sn}),
  We proceed by induction on the length of the longest reduction
  $\redseq_1\,\redseq_2 \tofs \flatten{(\redseq_1\,\redseq_2)}$,
  considering four cases, depending on whether
  each of $\redseq_1$ and $\redseq_2$ is a multistep or a composition:
  \begin{enumerate}
  \item
    If both are multisteps,
    \ie $\redseq_1 = \msteptwo_1$ and $\redseq_2 = \msteptwo_2$:
    \[
      \begin{array}{rcll}
        \flatten{(\mstep_1\,\mstep_2)}\projmr\flatten{(\msteptwo_1\,\msteptwo_2)}
      & = &
        \flatten{(\mstep_1\,\mstep_2)}/\flatten{(\msteptwo_1\,\msteptwo_2)}
        & \text{by definition}
      \\
      & = &
        \flatten{((\mstep_1\,\mstep_2)/(\msteptwo_1\,\msteptwo_2))}
        & \text{by \rprop{body:properties_of_projection}}
      \\
      & = &
        \flatten{((\mstep_1/\msteptwo_1)\,(\mstep_2/\msteptwo_2))}
        & \text{by \rprop{body:properties_of_projection}}
      \\
      & = &
        \flatten{((\mstep_1\projmr\msteptwo_1)\,(\mstep_2\projmr\msteptwo_2))}
        & \text{by definition}
      \end{array}
    \]
  \item
    If $\redseq_1 = \msteptwo_1$ is a multistep
    and $\redseq_2 = \redseq_{21}\seq\redseq_{22}$ is a sequence:
    \[
      \begin{array}{rcll}
      &&
        \flatten{(\mstep_1\,\mstep_2)}\projmr\flatten{(\msteptwo_1\,(\redseq_{21}\seq\redseq_{22}))}
      \\
      & = &
        \flatten{(\mstep_1\,\mstep_2)}\projmr(
          \flatten{(\msteptwo_1\,\redseq_{21})}
          \seq
          \flatten{(\rtgt{\msteptwo_1}\,\redseq_{22})}
        )
        & \text{by~\rprop{flat_confluent}}
      \\
      & = &
        (\flatten{(\mstep_1\,\mstep_2)}
           \projmr\flatten{(\msteptwo_1\,\redseq_{21})})
           \projmr\flatten{(\rtgt{\msteptwo_1}\,\redseq_{22})}
        & \text{by definition}
      \\
      & = &
        \flatten{((\mstep_1\projmr\msteptwo_1)\,(\mstep_2\projmr\redseq_{21}))}
          \projmr\flatten{(\rtgt{\msteptwo_1}\,\redseq_{22})}
        & \text{by \ih}
      \\
      & = &
        \flatten{(((\mstep_1\projmr\msteptwo_1)\projmr\rtgt{\msteptwo_1})\,((\mstep_2\projmr\redseq_{21})\projmr\redseq_{22}))}
        & \text{by \ih}
      \\
      & = &
        \flatten{((\mstep_1\projmr(\msteptwo_1\seq\rtgt{\msteptwo_1}))\,(\mstep_2\projmr(\redseq_{21}\seq\redseq_{22})))}
        & \text{by definition}
      \\
      & = &
        \flatten{((\mstep_1\projmr(\msteptwo_1\seq\ftgt{\msteptwo_1}))\,(\mstep_2\projmr(\redseq_{21}\seq\redseq_{22})))}
        & \text{by~\rprop{flat_confluent}}
      \\
      & = &
        \flatten{((\mstep_1\projmr\msteptwo_1))\,(\mstep_2\projmr(\redseq_{21}\seq\redseq_{22})))}
        & \text{by \rprop{congruence_flateq_projection}, since $\msteptwo_1\seq\ftgt{\msteptwo_1} \flateq \msteptwo_1$}
      \end{array}
    \]
    To justify that the \ih may be applied, note that
    $\msteptwo_1\,(\redseq_{21}\seq\redseq_{22})
     \tof
     (\msteptwo_1\,\redseq_{21})\seq(\rtgt{\msteptwo_1}\,\redseq_{22})$.
  \item
    If $\redseq_1 = \redseq_{11}\seq\redseq_{12}$ is a sequence
    and $\redseq_2 = \msteptwo_2$ is a multistep, the proof is similar
    as for the previous case.
  \item
    If both are sequences, \ie 
    $\redseq_1 = \redseq_{11}\seq\redseq_{12}$
    and $\redseq_2 = \redseq_{21}\seq\redseq_{22}$:
    \[
      \begin{array}{rcll}
      &&
        \flatten{(\mstep_1\,\mstep_2)}\projmr
          \flatten{((\redseq_{11}\seq\redseq_{12})\,(\redseq_{21}\seq\redseq_{22}))}
      \\
      & = &
        \flatten{(\mstep_1\,\mstep_2)}\projmr
          (\flatten{((\redseq_{11}\seq\redseq_{12})\,\rsrc{\redseq_{21}}))}
           \seq
           \flatten{(\rtgt{\redseq_{12}}\,(\redseq_{21}\seq\redseq_{22}))})
        & \text{by~\rprop{flat_confluent}}
      \\
      & = &
        (\flatten{(\mstep_1\,\mstep_2)}
          \projmr\flatten{((\redseq_{11}\seq\redseq_{12})\,\rsrc{\redseq_{21}}))})
          \projmr\flatten{(\rtgt{\redseq_{12}}\,(\redseq_{21}\seq\redseq_{22}))}
        & \text{by definition}
      \\
      & = &
        \flatten{(
          (\mstep_1\projmr(\redseq_{11}\seq\redseq_{12}))\,
          (\mstep_2\projmr\rsrc{\redseq_{21}})
        )}
          \projmr\flatten{(\rtgt{\redseq_{12}}\,(\redseq_{21}\seq\redseq_{22}))}
        & \text{by \ih}
      \\
      & = &
        \flatten{(
          ((\mstep_1\projmr(\redseq_{11}\seq\redseq_{12}))\projmr\rtgt{\redseq_{12}})\,
          ((\mstep_2\projmr\rsrc{\redseq_{21}})\projmr(\redseq_{21}\seq\redseq_{22}))
        )}
        & \text{by \ih}
      \\
      & = &
        \flatten{(
          (\mstep_1\projmr((\redseq_{11}\seq\redseq_{12})\seq\rtgt{\redseq_{12}}))\,
          (\mstep_2\projmr(\rsrc{\redseq_{21}}\seq(\redseq_{21}\seq\redseq_{22})))
        )}
        & \text{by definition}
      \\
      & = &
        \flatten{(
          (\mstep_1\projmr((\redseq_{11}\seq\redseq_{12})\seq\ftgt{\redseq_{12}}))\,
          (\mstep_2\projmr(\fsrc{\redseq_{21}}\seq(\redseq_{21}\seq\redseq_{22})))
        )}
        & \text{by~\rprop{flat_confluent}}
      \\
      & = &
        \flatten{(
          (\mstep_1\projmr(\redseq_{11}\seq\redseq_{12}))\,
          (\mstep_2\projmr(\redseq_{21}\seq\redseq_{22}))
        )}
        & \text{by \rprop{congruence_flateq_projection}}
      \end{array}
    \]
    To justify that the \ih may be applied, note that
    $(\redseq_{11}\seq\redseq_{12})\,(\redseq_{21}\seq\redseq_{22})
     \tof
     (\redseq_{11}\seq\redseq_{12})\,\rsrc{\redseq_{21}}\seq
     \rtgt{\redseq_{12}}\,(\redseq_{21}\,\redseq_{22})$.
    To justify the last equality, note that 
    $(\redseq_{11}\seq\redseq_{12})\seq\ftgt{\redseq_{12}} \flateq \redseq_{11}\seq\redseq_{12}$
    and
    $\fsrc{\redseq_{21}}\seq(\redseq_{21}\seq\redseq_{22}) \flateq \redseq_{21}\seq\redseq_{22}$.
  \end{enumerate}
\end{proof}

\begin{prop}[Properties of projection for arbitrary rewrites]
\lprop{body:properties_of_arbitrary_projection}
\quad
\begin{enumerate}
\item
  Projection of a rewrite over a sequence and of a sequence over a rewrite
  obey the expected equations
  $\redseq\proja(\redseqtwo_1\seq\redseqtwo_2) = (\redseq\proja\redseqtwo_1)\proja\redseqtwo_2$
  and
  $(\redseq_1\seq\redseq_2)\proja\redseqtwo = (\redseq_1\proja\redseqtwo)\seq(\redseq_2\proja(\redseqtwo\proja\redseq_1))$.
\item
  Projection commutes with abstraction and application, that is:
  \begin{enumerate}
  \item
    $(\lam{\var}{\redseq})\proja(\lam{\var}{\redseqtwo}) \permeq \lam{\var}{(\redseq\proja\redseqtwo)}$,
    and more precisely
    $(\lam{\var}{\redseq})\proja(\lam{\var}{\redseqtwo}) \tofsinv \lam{\var}{(\redseq\proja\redseqtwo)}$.
  \item
    If $\redseq_1,\redseqtwo_1$ are coinitial and $\redseq_2,\redseqtwo_2$ are coinitial,
    then
    $(\redseq_1\,\redseq_2)\proja(\redseqtwo_1\,\redseqtwo_2) \permeq
     (\redseq_1\proja\redseqtwo_1)\,(\redseq_2\proja\redseqtwo_2)$,
    and more precisely
    $(\redseq_1\,\redseq_2)\proja(\redseqtwo_1\,\redseqtwo_2) \tofsinv
     (\redseq_1\proja\redseqtwo_1)\,(\redseq_2\proja\redseqtwo_2)$.
  \end{enumerate}
\item
  The projection of a rewrite over itself is always empty,
  $\redseq\proja\redseq \permeq \rtgt{\redseq}$.
\item
  Permutation equivalence is a congruence with respect to projection,
  namely if $\redseq \permeq \redseqtwo$
  then $\redseqthree\proja\redseq = \redseqthree\proja\redseqtwo$
  and $\redseq\proja\redseqthree \permeq \redseqtwo\proja\redseqthree$.
\item
  The key equation
  $\redseq\seq(\redseqtwo\proja\redseq) \permeq \redseqtwo\seq(\redseq\proja\redseqtwo)$
  holds.
\end{enumerate}
\end{prop}
  \begin{proof}
We prove each item separately: 
\begin{enumerate}
  \item Projection of composition:
  \begin{itemize}
  \item On one hand:
    \[
      \begin{array}{rcll}
        \redseq\proja(\redseqtwo_1\seq\redseqtwo_2)
      & = &
        \flatten{\redseq}/(\flatten{\redseqtwo_1}\seq\flatten{\redseqtwo_2})
        & \text{by definition}
      \\
      & = &
        (\flatten{\redseq}/\flatten{\redseqtwo_1})/\flatten{\redseqtwo_2}
      \\
      & = &
        (\redseq\proja\redseqtwo_1)/\flatten{\redseqtwo_2}
        & \text{by definition}
      \\
      & = &
        \flatten{(\redseq\proja\redseqtwo_1)}/\flatten{\redseqtwo_2}
        & \text{as $\redseq\proja\redseqtwo_1$ is flat by construction}
      \\
      &= &
        (\redseq\proja\redseqtwo_1)\proja\redseqtwo_2
        & \text{by definition}
      \end{array}
    \]
  \item On the other hand:
    \[
      \begin{array}{rcll}
        (\redseq_1\seq\redseq_2)\proja\redseqtwo
      & = &
        \flatten{(\redseq_1\seq\redseq_2)}/\flatten{\redseqtwo}
        & \text{by definition}
      \\
      & = &
        (\flatten{\redseq_1}\seq\flatten{\redseq_2})/\flatten{\redseqtwo}
      \\
      & = &
        (\flatten{\redseq_1}/\flatten{\redseqtwo})\seq
        (\flatten{\redseq_2}/(\flatten{\redseqtwo}/\flatten{\redseq_1}))
      \\
      & = &
        (\redseq_1\proja\redseqtwo)\seq
        (\flatten{\redseq_2}/(\redseqtwo\proja\redseq_1))
        & \text{by definition}
      \\
      & = &
        (\redseq_1\proja\redseqtwo)\seq
        (\flatten{\redseq_2}/\flatten{(\redseqtwo\proja\redseq_1)})
        & \text{as $\redseqtwo\proja\redseq_1$ is flat by construction}
      \\
      & = &
        (\redseq_1\proja\redseqtwo)\seq(\redseq_2\proja(\redseqtwo\proja\redseq_1))
        & \text{by definition}
      \end{array}
    \]
  \end{itemize}
\item Projection commutes with
  \begin{enumerate}
    \item Abstraction:
  \[
    \begin{array}{rcll}
      (\lam{\var}{\redseq})\proja(\lam{\var}{\redseqtwo})
    & = &
      \flatten{(\lam{\var}{\redseq})}/\flatten{(\lam{\var}{\redseqtwo})}
      & \text{by definition}
    \\
    & = &
      \flatten{(\lam{\var}{\flatten{\redseq}})}/\flatten{(\lam{\var}{\flatten{\redseqtwo}})}
      & \text{by \rprop{flat_confluent}}
    \\
    & = &
      \flatten{(\lam{\var}{(\flatten{\redseq}/\flatten{\redseqtwo})})}
      & \text{by \rlem{proja_abstraction}}
    \\
    & \tofsinv &
      \lam{\var}{(\flatten{\redseq}/\flatten{\redseqtwo})}
    \\
    & = &
      \lam{\var}{(\redseq\proja\redseqtwo)}
      & \text{by definition}
    \end{array}
  \]
\item Application:
  \[
    \begin{array}{rcll}
      (\redseq_1\,\redseq_2)\proja(\redseqtwo_1\,\redseqtwo_2)
    & \eqdef &
      \flatten{(\redseq_1\,\redseq_2)}/\flatten{(\redseqtwo_1\,\redseqtwo_2)}
      & \text{by definition}
    \\
    & \eqdef &
      \flatten{(\flatten{\redseq_1}\,\flatten{\redseq_2})}/
                \flatten{(\flatten{\redseqtwo_1}\,\flatten{\redseqtwo_2})}
      & \text{by \rprop{flat_confluent}}
    \\
    & \eqdef &
      \flatten{(
       (\flatten{\redseq_1}/\flatten{\redseqtwo_1})\,
       (\flatten{\redseq_2}/\flatten{\redseqtwo_2}))}
      & \text{by \rlem{proja_application}}
    \\
    & \tofsinv &
       (\flatten{\redseq_1}/\flatten{\redseqtwo_1})\,
       (\flatten{\redseq_2}/\flatten{\redseqtwo_2})
    \\
    & = &
      (\redseq_1\proja\redseqtwo_1)\,(\redseq_2\proja\redseqtwo_2)
      & \text{by definition}
    \end{array}
  \]
\end{enumerate}

\item Self-erasure:
  \[
    \begin{array}{rcll}
      \redseq\proja\redseq
    & = &
      \flatten{\redseq}/\flatten{\redseq}
      & \text{by definition}
    \\
    & \flateq &
      \ftgtb{\redseq}
      & \text{by \rlem{self_erasure}}
    \\
    & \tofsinv &
      \rtgt{\redseq}
    \end{array}
  \]
  It suffices to recall that
  flat permutation equivalence ($\flateq$) and
  flattening ($\tof$)
  are both included in permutation equivalence ($\permeq$).
\item
  Congruence of projection:
  Let $\redseq \permeq \redseqtwo$.
  By \rthm{body:soundness_completeness_flateq} this means that
  $\flatten{\redseq} \flateq \flatten{\redseqtwo}$. Then:
  \begin{itemize}
  \item On one hand:
    \[
      \begin{array}{rcll}
        \redseqthree\proja\redseq
      & = &
        \flatten{\redseqthree}/\flatten{\redseq}
        & \text{by definition}
      \\
      & = &
        \flatten{\redseqthree}/\flatten{\redseqtwo}
        & \text{by \rprop{congruence_flateq_projection}}
      \\
      & = &
        \redseqthree\proja\redseqtwo
        & \text{by definition}
      \end{array}
    \]
  \item On the other hand:
    \[
      \begin{array}{rcll}
        \redseq\proja\redseqthree
      & = &
        \flatten{\redseq}/\flatten{\redseqthree}
        & \text{by definition}
      \\
      & \flateq &
        \flatten{\redseqtwo}/\flatten{\redseqthree}
        & \text{by \rprop{congruence_flateq_projection}}
      \\
      & = &
        \redseqtwo\proja\redseqthree
        & \text{by definition}
      \end{array}
    \]
  \end{itemize}
\item Permutation:
  \[
    \begin{array}{rcll}
      \redseq\seq(\redseqtwo\proja\redseq)
    & = &
      \redseq\seq(\flatten{\redseqtwo}/\flatten{\redseq})
      & \text{by definition}
    \\
    & \tofs &
      \flatten{\redseq}\seq(\flatten{\redseqtwo}/\flatten{\redseq})
    \\
    & \flateq &
      \flatten{\redseqtwo}\seq(\flatten{\redseq}/\flatten{\redseqtwo})
      & \text{by \rlem{rewrite_permutation}}
    \\
    & \tofsinv &
      \redseqtwo\seq(\flatten{\redseq}/\flatten{\redseqtwo})
    \\
    & = &
      \redseqtwo\seq(\redseq\proja\redseqtwo)
      & \text{by definition}
    \end{array}
  \]
  It suffices to recall that
  flat permutation equivalence ($\flateq$) and
  flattening ($\tof$)
  are both included in permutation equivalence ($\permeq$).
\end{enumerate}
\end{proof}

\subsection{Characterization of Permutation Equivalence in Terms of Projection.}
Finally, we are able to characterize permutation equivalence
$\redseq \permeq \redseqtwo$ as the condition that the projections
$\redseq\proja\redseqtwo$ and $\redseqtwo\proja\redseq$ are both empty.
Indeed:

\begin{thm}[Projection equivalence]
\lthm{body:projection_equivalence}
Let $\redseq,\redseqtwo$ be arbitrary coinitial rewrites.
Then $\redseq \permeq \redseqtwo$
if and only if
$\redseq\proja\redseqtwo \permeq \rtgt{\redseqtwo}$
and
$\redseqtwo\proja\redseq \permeq \rtgt{\redseq}$.
\end{thm}
\begin{proof}
$(\Rightarrow)$
  Suppose that $\redseq \permeq \redseqtwo$.
  Then, by~\rprop{body:properties_of_arbitrary_projection},
  $\redseq\proja\redseqtwo \permeq \redseqtwo\proja\redseqtwo \permeq \rtgt{\redseqtwo}$.
  Symmetrically, $\redseqtwo\proja\redseq \permeq \rtgt{\redseq}$.
$(\Leftarrow)$
  Let $\redseq\proja\redseqtwo \permeq \rtgt{\redseqtwo}$
  and $\redseqtwo\proja\redseq \permeq \rtgt{\redseq}$.
  Then, by~\rprop{body:properties_of_arbitrary_projection},
  $
    \redseq
    \,\,\,\permeq\,\,\,
    \redseq\seq\rtgt{\redseq}
    \,\,\,\permeq\,\,\,
    \redseq\seq(\redseqtwo\proja\redseq)
    \,\,\,\permeq\,\,\,
    \redseqtwo\seq(\redseq\proja\redseqtwo)
    \,\,\,\permeq\,\,\,
    \redseqtwo\seq\rtgt{\redseqtwo}
    \,\,\,\permeq\,\,\,
    \redseqtwo
  $.
\end{proof}

Since flattening and projection are computable,
\rthm{body:soundness_completeness_flateq} and \rthm{body:projection_equivalence}
together provide an {\bf effective method to decide permutation equivalence}
$\redseq \permeq \redseqtwo$ for arbitrary rewrites.
Indeed, to test whether 
$\redseq\proja\redseqtwo \permeq \rtgt{\redseqtwo}$,
note by~\rthm{body:soundness_completeness_flateq}
that this is equivalent to testing whether
$\redseq\proja\redseqtwo \flateq \ftgtb{\redseqtwo}$,
so it suffices to check that $\redseq\proja\redseqtwo$ is {\em empty},
\ie it contains no rule symbols.
This is justified by the fact that if $\mstep$ has no rule symbols and
$\mstep \flateq \redseq$, then $\redseq$ has no rule symbols.
% (\SeeAppendix{See \rlem{characterization_of_empty_multisteps} in~\cite{techReport}}).
%(\SeeAppendix{See Lem. 162 in~\cite{techReport}}).

\begin{lem}[Characterization of empty multisteps]
\llem{characterization_of_empty_multisteps}
Let $\mstep$ be a flat multistep.
Then the following are equivalent:
\begin{enumerate}
\item $\mstep$ is a term, \ie $\mstep = \refl{\tm}$.
\item There exists a term $\tm$ such that $\mstep \flateq \refl{\tm}$.
\item There exists a term $\tm$ such that if
      $\mstep \flateq \redseq$
      then there exists a composition context $\kctx$
      such that
      $\redseq = \kctxof{\refl{\tm},\hdots,\refl{\tm}}$.
\item There exists a term $\tm$ such that if
      $\mstep \flateq \msteptwo$ then $\msteptwo = \refl{\tm}$.
\end{enumerate}
\end{lem}
\begin{proof}
\quad
\begin{itemize}
\item $(1 \implies 2)$
  Let $\mstep = \refl{\tm}$.
  Then it is immediate as $\mstep \flateq \refl{\tm}$.
\item $(2 \implies 3)$
  Let $\mstep \flateq \refl{\tm}$
  and suppose that $\mstep \flateq \redseq$.
  We claim that there is a composition context $\kctx$ such that
  $\redseq = \kctxof{\refl{\tm},\hdots,\refl{\tm}}$.
  First, note that we have that $\refl{\tm} \flateq \redseq$. 
  This means that there is a sequence of applications of the axioms
  defining flat permutation equivalence such that
  $\refl{\tm} = \redseq_0
        \flateq \redseq_1
        \flateq \redseq_2
        \hdots
        \flateq \redseq_n = \redseq$.
  We proceed by induction on $n$.
  If $n = 0$, it is trivial taking $\kctx := \ctxhole$.
  For the inductive step, it suffices to show
  that if $\redseq_i$ is of the form
  $\redseq_i = \kctxof{\refl{\tm},\hdots,\refl{\tm}}$
  then $\redseq_{i+1}$ is of the form
  $\redseq_i = \kctx'\ctxof{\refl{\tm},\hdots,\refl{\tm}}$.
  The case for the \flateqRule{Assoc} rule is immediate.
  The interesting case is the \flateqRule{Perm} rule.
  There are two subcases, depending on whether the \flateqRule{Perm}
  rule is applied forwards or backwards:
  \begin{enumerate}
  \item {\bf Forwards application of \flateqRule{Perm}.}
    That is, $\redseq_i = \sctxof{\refl{\tm}}$
    and $\redseq_{i+1} = \sctxof{\flatten{\mstep_1}\seq\flatten{\mstep_2}}$
    where $\sctx$ is a composition context
    and $\judgSplit{\refl{\tm}}{\mstep_1}{\mstep_2}$.
    Given that $\tm$ has no rule symbols, it is easy to check by
    induction on $\tm$
    that $\mstep_1 = \refl{\tm}$ and $\mstep_2 = \refl{\tm}$.
    This concludes the proof of this case.
  \item {\bf Backwards application of \flateqRule{Perm}.}
    That is, $\redseq_i = \sctxof{\refl{\tm}\seq\refl{\tm}}$
    and $\redseq_{i+1} = \sctxof{\mstep_1}$
    where $\judgSplit{\mstep}{\mstep_2}{\mstep_3}$
    and $\flatten{\mstep_2} = \flatten{\mstep_3} = \refl{\tm}$.
    Note that $\mstep_1$ is a flat multistep.
    To finish the proof,
    it suffices to show that,
    $\mstep_1 = \mstep_2 = \mstep_3$.
    This is implied by the following claim.

    {\bf Claim.}
    Let $\judgSplit{\mstep_1}{\mstep_2}{\mstep_3}$
    where $\flatten{\mstep_2} = \flatten{\mstep_3} = \refl{\tm}$
    for some term $\tm$.
    Then $\mstep_1 = \mstep_2 = \mstep_3$. \\
    {\em Proof of the claim.}
    We proceed by induction on $\mstep_1$, following the characterization of
    $\tof$-normal multisteps, 
    There are three subcases, depending on the head of $\mstep_1$:
    \begin{enumerate}
    \item {\bf $\mstep_1$ headed by a variable.}
      Then
      $\mstep_1 = \lam{\vec{\var}}{\var\,\msteptwo_{11}\hdots\msteptwo_{1n}}$.
      The judgment $\judgSplit{\mstep_1}{\mstep_2}{\mstep_3}$
      must be derived by a number of
      applications of the \indrulename{SAbs} rule,
      followed by $n$ applications of the \indrulename{SApp} rule,
      followed by an application of the \indrulename{SVar} rule.
      Hence
      $\mstep_2 = \lam{\vec{\var}}{\var\,\msteptwo_{21}\hdots\msteptwo_{2n}}$
      and
      $\mstep_3 = \lam{\vec{\var}}{\var\,\msteptwo_{31}\hdots\msteptwo_{3n}}$,
      where moreover $\judgSplit{\msteptwo_{1i}}{\msteptwo_{2i}}{\msteptwo_{3i}}$
      for all $1 \leq i \leq n$.
      Since $\flatten{\mstep_2} = \flatten{\mstep_3}$
      then also $\flatten{\mstep_{2i}} = \flatten{\mstep_{3i}}$
      for all $1 \leq i \leq n$,
      and moreover $\flatten{\mstep_{2i}}$ must be a term
      (\ie without occurrences of rule symbols),
      for otherwise there would be a rule symbol in $\flatten{\mstep_2}$.
      Then by \ih we have that
      $\flatten{\mstep_{1i}} = \flatten{\mstep_{2i}} = \flatten{\mstep_{3i}}$
      for all $1 \leq i \leq n$.
      This concludes the proof.
    \item {\bf $\mstep_1$ headed by a constant.}
      Similar to the previous case, when $\mstep_1$ is headed by a variable.
    \item {\bf $\mstep_1$ headed by a rule symbol.}
      We argue that this case is impossible.
      Indeed, if
      $\mstep_1$ is of the form
      $\lam{\vec{\var}}{\rulewit\,\msteptwo_{11}\hdots\msteptwo_{1n}}$
      then the judgment $\judgSplit{\mstep_1}{\mstep_2}{\mstep_3}$
      must be derived by a number of
      applications of the \indrulename{SAbs} rule,
      followed by $n$ applications of the \indrulename{SApp} rule,
      followed by an application of either \indrulename{SRuleL}
      or \indrulename{SRuleR} at the head.
      If the judgment is derived
      by an application of \indrulename{SRuleL} at the head,
      then we have that
      $\mstep_2 = \lam{\vec{\var}}{\rulewit\,\msteptwo_{21}\hdots\msteptwo_{2n}}$
      and
      $\mstep_3 = \lam{\vec{\var}}{\rtgt{\rulewit}\,\msteptwo_{31}\hdots\msteptwo_{3n}}$
      and moreover that
      $\judgSplit{\mstep_{1i}}{\mstep_{2i}}{\mstep_{3i}}$
      for all $1 \leq i \leq n$.
      But this is impossible, given that $\flatten{\mstep_2}$ would contain
      a rule symbol, contradicting the fact that
      $\flatten{\mstep_2} = \refl{\tm}$.
      Using a similar argument, we note that the judgment
      cannot be derived by an application of the \indrulename{SRuleR} rule
      at the head.
    \end{enumerate}
  \end{enumerate}
\item $(3 \implies 4)$
  Suppose that there exists a term $\tm$
  such that if $\mstep \flateq \redseq$
  then $\redseq$ is of the form
  $\redseq = \kctxof{\refl{\tm},\hdots,\refl{\tm}}$.
  Moreover, suppose that $\mstep \flateq \msteptwo$.
  Then by hypothesis $\msteptwo$ must be of the form
  $\kctxof{\refl{\tm},\hdots,\refl{\tm}}$.
  Since $\msteptwo$ is a multistep, containing no compositions,
  then $\kctxof = \ctxhole$ and indeed $\msteptwo = \refl{\tm}$.
\item $(4 \implies 1)$
  Suppose that there exists a term $\tm$
  such that if $\mstep \flateq \msteptwo$
  then $\msteptwo = \refl{\tm}$.
  Then in particular, since $\mstep \flateq \mstep$,
  we have that $\mstep = \refl{\tm}$.
\end{itemize}
\end{proof}

%%% Local Variables:
%%% mode: latex
%%% TeX-master: "main"
%%% End:

\section{Standardization}
\lsec{standardization}

A rewrite is considered \emph{standard} if the steps of the computation
are ordered in such a way that they cannot be moved further upfront.
The process that reorders steps, converting rewrites into their standard
forms, is called \emph{standardization}.
Any two rewrites that convert to the same
standard rewrite are considered {\em standardization equivalent}.
In this section we introduce the
{\em standardization rewrite system}~(\rdef{standardization_rewrite_system}),
we show that it is
{\em strongly normalizing and confluent}~(\rprop{body:standardization_sn_cr}),
subject to a finiteness condition
and up to a notion of strong equivalence.
From this we derive
a {\em standardization result}~(\rthm{body:standardization_theorem}).

To formulate
standardization, it is convenient to deal with sequential
rewrites,
consisting of a sequence of composed flat multisteps
ending in a unit rewrite,
\ie rewrites of the form
$\mstep_1\seq\mstep_2\seq\hdots\seq\mstep_n\seq\refl{\tm}$ with $n \geq 0$.
Composition (``$\seq$'') is assumed to be right-associative.
The {\em length} of a sequential rewrite is
$\len{\mstep_1\seq\mstep_2\seq\hdots\seq\mstep_n\seq\refl{\tm}} \eqdef n$.
Each flat rewrite is equivalent to some (non-unique)
sequential rewrite, by associating compositions,
and using the equivalence $\mstep \flateq \mstep\seq\ftgtb{\mstep}$ to
ensure that the last multistep is empty.

\begin{defi}[Standardization rewrite system]
\ldef{standardization_rewrite_system}
The standardization rewrite system
is given by the set of all sequential rewrites,
with a reduction rule $\redseq \tostd \redseqtwo$ defined by:
\[
  \indrule{\stdrule{Del}}{
  }{
    \refl{\tm}\seq\redseq
    \tostd
    \redseq
  }
  \indrule{\stdrule{Pull}}{
    \mstep_{12} \flateq \mstep_1\seq\mstep_2
    \HS
    \mstep_{23} \flateq \mstep_2\seq\mstep_3
    \HS
    \mstep_2 \text{\ is non-empty}
  }{
    \mstep_1\seq\mstep_{23}\seq\redseq
    \tostd
    \mstep_{12}\seq\mstep_3\seq\redseq
  }
\]
and closed by a congruence rule \stdrule{Cong}
stating that $\redseq \tostd \redseqtwo$
implies $\mstep\seq\redseq \tostd \mstep\seq\redseqtwo$.
\end{defi}
Rule \stdrule{Del} removes empty steps from the rewrite
(other than the very last one). Rule \stdrule{Pull} ``pulls''
steps upfront.
\begin{exa}
Given typed constants $\fdup : (\iota \imp \iota) \to \iota$
and $\fun, \funtwo : \iota \to \iota$,
and rewriting rules
$\rewr{
  \rulewit
 }{
  \lam{\var}{\lam{\vartwo}{\fdup\,(\lam{\varthree}{\var\,\varthree})\,\vartwo}}
 }{
  \lam{\var}{\lam{\vartwo}{\var\,(\var\,\vartwo)}}
 }{
  (\iota\to\iota)\to\iota\to\iota
 }$
and
$\rewr{
  \rulewittwo
 }{
  \lam{\var}{\fun(\fun(\var))}
 }{
  \lam{\var}{\funtwo(\var)}
 }{
  \iota\to\iota
 }$,
note first that:
\[
  \begin{array}{l@{\HS\HS}l}
    \fdup(\rulewit\,\fun) :
    \fdup\,(\fdup\,\fun)
    \rewto
    \fdup\,(\lam{\var}{\fun\,(\fun\,\var)})
  &
    \fdup(\rulewittwo) :
    \fdup\,(\lam{\var}{\fun\,(\fun\,\var)})
    \rewto
    \fdup\,\funtwo
  \\
    \rulewit\,\funtwo :
    \fdup\,\funtwo
    \rewto
    \lam{\var}{\funtwo\,(\funtwo\,\var)}
  &
    \rulewit\,(\rulewit\,\fun) :
    \fdup\,(\fdup\,\fun)
    \rewto
    \lam{\var}{\fun\,(\fun\,(\fun\,(\fun\,\var)))}
  \\
    \rulewit\,\rulewittwo :
    \fdup\,(\lam{\var}{\fun\,(\fun\,\var)})
    \rewto
    \lam{\var}{\funtwo\,(\funtwo\,\var)}
  &
    \lam{\var}{\rulewittwo\,(\rulewittwo\,\var)} :
    \lam{\var}{\fun\,(\fun\,(\fun\,(\fun\,\var)))}
    \rewto
    \lam{\var}{\funtwo\,(\funtwo\,\var)}
  \\
  \end{array}
\]
Moreover:
\[
  \begin{array}{ll}
  &
    \fdup(\rulewit\,\fun)
    \seq \fdup(\rulewittwo)
    \seq \rulewit\,\funtwo
    \seq \refl{\lam{\var}{\funtwo\,(\funtwo\,\var)}}
     \\
    \tostdBy{Pull} &
    \fdup(\rulewit\,\fun)
    \seq \rulewit\,\rulewittwo
    \seq \refl{\lam{\var}{\funtwo\,(\funtwo\,\var)}}
    \seq \refl{\lam{\var}{\funtwo\,(\funtwo\,\var)}}
  \\
  \tostdBy{Del} &
    \fdup(\rulewit\,\fun)
    \seq \rulewit\,\rulewittwo
    \seq \refl{\lam{\var}{\funtwo\,(\funtwo\,\var)}}
    \\
    \tostdBy{Pull} &
    \rulewit\,(\rulewit\,\fun)
    \seq \lam{\var}{\rulewittwo\,(\rulewittwo\,\var)}
    \seq \refl{\lam{\var}{\funtwo\,(\funtwo\,\var)}}
  \end{array}
\]
\end{exa}

One important observation in the presence of erasing rewrite rules is
that  
 $\mstep \flateq \msteptwo$
does not necessarily imply that $\mstep = \msteptwo$.
For example, if
$\rulewit : \lam{\var}{\cons\,\var} \to \lam{\var}{\constwo}$
and
$\rulewittwo : \consthree \to \consthree'$
then the multistep
$\rulewit\,\rulewittwo : \cons\,\consthree \to \constwo$
and the multistep
$\rulewit\,\refl{\consthree} : \cons\,\consthree \to \constwo$
are permutation equivalent.
To justify this, note that, on one hand,
$\judgSplit{
   \rulewit\,\rulewittwo
 }{
   \rulewit\,\refl{\consthree}
 }{
   (\lam{\var}{\refl{\constwo}})\,\rulewittwo
 }$
and, on the other hand,
$\judgSplit{
   \rulewit\,\refl{\consthree}
 }{
   \rulewit\,\refl{\consthree}
 }{
   (\lam{\var}{\refl{\constwo}})\,\refl{\consthree}
 }$
so, by the \flatRule{Perm} rule,
$\rulewit\,\rulewittwo
 \flateq
 \rulewit\,\refl{\consthree} \seq \refl{\constwo}
 \flateq
 \rulewit\,\refl{\consthree}$.
Thus, in the standardization system that we propose,
uniqueness of standardization will be determined
modulo a notion of {\em strong equivalence}
rather than syntactical equality of the resulting sequential rewrites.

\begin{defi}[Strong equivalence]
Two sequential rewrites $\redseq,\redseqtwo$
are said to be {\em strongly equivalent},
written $\redseq \streq \redseqtwo$,
if they are of the form
$\redseq = \mstep_1\seq\hdots\seq\mstep_n\seq\refl{\tm}$
and
$\redseqtwo = \msteptwo_1\seq\hdots\seq\msteptwo_n\seq\refl{\tm}$
where $\mstep_i \flateq \msteptwo_i$ for all $i\in1..n$.
\end{defi}

It can then be shown that standardization is locally confluent,
up to strong equivalence:

\begin{lem}[Local confluence of $\tostd$ up to $\streq$]
\llem{body:streq_strong_bisimulation}
\llem{body:standardization_wcr}
\quad
\begin{enumerate}
\item
  \resultName{Strong Bisimulation.}
  If $\redseq \streq \redseq' \tostd \redseqtwo$
  then there exists $\redseqtwo'$ such that
  $\redseq \tostd \redseqtwo' \streq \redseqtwo$.
\item
  \resultName{Local Confluence.}
  If $\redseq_1 \tostd \redseq_2$ 
  and $\redseq_1 \tostd \redseq_3$ 
  then $\redseq_2 \tostd^*\streq \redseq_4$
  and $\redseq_3 \tostd^*\streq \redseq_4$
  for some $\redseq_4$.
\end{enumerate}
\end{lem}
\begin{proof}
The proof of the first item is by induction on the derivation of $\redseq' \tostd \redseqtwo$. It relies on the property that  $\mstep = \refl{\tm}$ iff $\mstep \flateq \refl{\tm}$. For the second item, it suffices to study the cases in which the rewriting rules
\stdrule{Del} and \stdrule{Pull} overlap. We present one sample case, the most interesting one, namely the \stdrule{Pull}/\stdrule{Pull} case.
  Consider a \stdrule{Pull} step
  $\mstep_1\seq\mstep_{23}\seq\redseq
   \tostd
   \mstep_{12}\seq\mstep_3\seq\redseq$
  where $\mstep_{12} \flateq \mstep_1\seq\mstep_2$
  and $\mstep_{23} \flateq \mstep_2\seq\mstep_3$
  with $\mstep_2$ a non-empty multistep.
  Moreover, suppose that there is
  another \stdrule{Pull} step at the root,
  of the form
  $\mstep_1\seq\mstep_{23}\seq\redseq
   \tostd
   \mstep'_{12}\seq\mstep'_3\seq\redseq$
  where $\mstep'_{12} \flateq \mstep_1\seq\mstep'_2$
  and $\mstep_{23} \flateq \mstep'_2\seq\mstep'_3$,
  with $\mstep'_2$ a non-empty multistep.

  Note that
  $\mstep_2\seq\mstep_3
   \flateq \mstep_{23}
   \flateq \mstep'_2\seq\mstep'_3$.
  Hence:
  \[
    \begin{array}{rlll}
      \mstep_3
    & \flateq &
      \ftgt{\mstep_2} \seq \mstep_3
      & \text{by \rthm{body:soundness_completeness_flateq} using \permeqRule{IdL}}
    \\
    & \flateq &
      \ftgt{\mstep_2} \seq (\mstep_3/\ftgt{\mstep_2})
      & \text{by \rprop{body:properties_of_projection}}
    \\
    & \flateq &
      (\mstep_2/\mstep_2)\seq(\mstep_3/(\mstep_2/\mstep_2))
      & \text{by \rlem{self_erasure} and \rprop{congruence_flateq_projection}}
    \\
    & = &
      (\mstep_2\seq\mstep_3)/\mstep_2
      & \text{by definition of projection}
    \\
    & \flateq &
      (\mstep'_2\seq\mstep'_3)/\mstep_2
      & \text{by \rprop{congruence_flateq_projection}}
    \\
    & = &
      (\mstep'_2/\mstep_2)\seq(\mstep'_3/(\mstep_2/\mstep'_2))
      & \text{by definition of projection}
    \end{array}
  \]
  Symmetrically,
  $\mstep'_3 \flateq
   (\mstep_2/\mstep'_2)\seq(\mstep_3/(\mstep'_2/\mstep_2))$.

  The proof proceeds by case analysis, depending on whether
  $\mstep'_2/\mstep_2$ and $\mstep_2/\mstep'_2$
  are empty or not. There are four subcases:
  \begin{enumerate} 
  \item {\bf $\mstep'_2/\mstep_2$ and $\mstep_2/\mstep'_2$ are both empty.}
    Then, note that $\mstep_{12} \flateq \mstep'_{12}$
    and $\mstep_3 \flateq \mstep'_3$. Indeed:
    \[
      \begin{array}{rlll}
        \mstep_{12}
      & \flateq &
        \mstep_1\seq\mstep_2
      \\
      & \flateq &
        \mstep_1\seq\mstep'_2
        & \text{by \rthm{body:projection_equivalence}}
      \\
      & \flateq &
        \mstep'_{12}
      \end{array}
    \]
    \[
      \begin{array}{rlll}
        \mstep_3
      & \flateq &
        (\mstep'_2/\mstep_2)\seq(\mstep'_3/(\mstep_2/\mstep'_2))
        & \text{as shown above}
      \\
      & = &
        \ftgt{\mstep_2}\seq(\mstep'_3/\ftgt{(\mstep'_2)})
        & \text{as $\mstep_2/\mstep'_2$ is empty}
      \\
      & \permeq &
        \mstep'_3
        & \text{by \rthm{body:soundness_completeness_flateq},
                   \rprop{body:properties_of_projection}}
      \end{array}
    \]
    Hence the diagram may be closed as follows:
    \[
      \xymatrix{
        \mstep_1\seq\mstep_{23}\seq\redseq
        \ar[r]^-{\stdrule{Pull}}
        \ar[d]_-{\stdrule{Pull}}
      &
        \mstep_{12}\seq\mstep_3\seq\redseq
      \\
        \mstep'_{12}\seq\mstep'_3\seq\redseq
        \ar@{.}[ru]_-{\streq}
      }
    \]
  \item {\bf $\mstep'_2/\mstep_2$ is empty and
             $\mstep_2/\mstep'_2$ is non-empty.}
    Let $\msteptwo := \mstep_{12}$.
    First note that:
    \[
      \begin{array}{rcll}
        \msteptwo
      & = &
        \mstep_{12}
      \\
      & \flateq &
        \mstep_1\seq\mstep_2
      \\
      & \flateq &
        \mstep_1\seq\mstep_2\seq\ftgt{\mstep_2}
        & \text{by \rprop{body:properties_of_projection}}
      \\
      & \flateq &
        \mstep_1\seq\mstep_2\seq(\mstep'_2/\mstep_2)
        & \text{as $\mstep'_2/\mstep_2$ is empty}
      \\
      & \flateq &
        \mstep_1\seq\mstep'_2\seq(\mstep_2/\mstep'_2)
        & \text{by \rlem{multistep_permutation}}
      \\
      & \flateq &
        \mstep'_{12}\seq(\mstep_2/\mstep'_2)
      \end{array}
    \]
    Moreover, note that:
    \[
      \begin{array}{rlll}
        \mstep'_3
      & \flateq &
        (\mstep_2/\mstep'_2)\seq(\mstep_3/(\mstep'_2/\mstep_2))
        & \text{as shown above}
      \\
      & \flateq &
        (\mstep_2/\mstep'_2)\seq(\mstep_3/\ftgt{\mstep_2})
        & \text{as $\mstep'_2/\mstep_2$ is empty}
      \\
      & \flateq &
        (\mstep_2/\mstep'_2)\seq\mstep_3
        & \text{by \rprop{body:properties_of_projection}}
      \end{array}
    \]
    Since $\mstep_2/\mstep'_2$ is non-empty,
    this means that there is a \stdrule{Pull} step
    $\mstep'_{12}\seq\mstep'_3\seq\redseq
     \tostd
     \msteptwo\seq\mstep_3\seq\redseq$.

    Then:
    \[
      \xymatrix{
        \mstep_1\seq\mstep_{23}\seq\redseq
        \ar[r]^-{\stdrule{Pull}}
        \ar[d]_-{\stdrule{Pull}}
      &
        \mstep_{12}\seq\mstep_3\seq\redseq
        \ar@{.}[d]^-{\streq}
      \\
        \mstep'_{12}\seq\mstep'_3\seq\redseq
        \ar@{.>}[r]_-{\stdrule{Pull}}
      &
        \msteptwo\seq\mstep_3\seq\redseq
      }
    \]
  \item {\bf $\mstep'_2/\mstep_2$ is non-empty and
             $\mstep_2/\mstep'_2$ is empty.}
    Symmetric to the previous case.
  \item {\bf $\mstep'_2/\mstep_2$ and $\mstep_2/\mstep'_2$ are non-empty.}
    Let $\msteptwo := \mstep_{12} \join \mstep'_{12}$
    and $\mstepthree := \mstep'_3/(\mstep_2/\mstep'_2)$.
    We claim that there are two \stdrule{Pull} steps,
    $\mstep_{12}\seq\mstep_3\seq\redseq
     \tostd \msteptwo\seq\mstepthree\seq\redseq$
    and
    $\mstep'_{12}\seq\mstep'_3\seq\redseq
     \tostd \msteptwo\seq\mstepthree\seq\redseq$.
    We prove these two facts:
    \begin{enumerate}
    \item
      First, note that:
      \[
        \begin{array}{rlll}
          \msteptwo
        & = &
          \mstep_{12} \join \mstep'_{12}
        \\
        & \flateq &
          \mstep_{12} \seq (\mstep'_{12}/\mstep_{12})
          & \text{by \rlem{rewrite_permutation}} %\rlem{body:basic_property_of_join}}
        \\
        & \flateq &
          \mstep_{12} \seq ((\mstep_1\seq\mstep'_2)/(\mstep_1\seq\mstep_2))
          & \text{by \rprop{congruence_flateq_projection}}
        \\
        & \flateq &
          \mstep_{12}\seq(\mstep'_2/\mstep_2)
          & \text{by \rlem{self_erasure} and \rprop{body:properties_of_projection}}
        \end{array}
      \]
      and recall that, as shown above
      $
        \mstep_3
        \flateq (\mstep'_2/\mstep_2)\seq(\mstep'_3/(\mstep_2/\mstep'_2))
        = (\mstep'_2/\mstep_2)\seq\mstepthree
      $.
      Hence, since $\mstep'_2/\mstep_2$ is non-empty,
      there is a \stdrule{Pull} step
      $\mstep_{12}\seq\mstep_3\seq\redseq
       \tostd \msteptwo\seq\mstepthree\seq\redseq$.
    \item
      Second, note that:
      \[
        \begin{array}{rlll}
          \msteptwo
        & \flateq &
          \mstep_{12}\seq(\mstep'_2/\mstep_2)
          & \text{as already shown}
        \\
        & \flateq &
          \mstep_1\seq\mstep_2\seq(\mstep'_2/\mstep_2)
        \\
        & \flateq &
          \mstep_1\seq\mstep'_2\seq(\mstep_2/\mstep'_2)
          & \text{by \rlem{rewrite_permutation}} %\rlem{body:multistep_permutation}}
        \\
        & \flateq &
          \mstep'_{12}\seq(\mstep_2/\mstep'_2)
        \end{array}
      \]
      Moreover, recall that, as shown above,
      $\mstep'_3 \flateq (\mstep_2/\mstep'_2)\seq(\mstep_3/(\mstep'_2/\mstep_2))$.
      Therefore:
      \[
        \begin{array}{rlll}
          (\mstep_2/\mstep'_2)/\mstep'_3
        & = &
          (\mstep_2/\mstep'_2)/((\mstep_2/\mstep'_2)\seq
                                (\mstep_3/(\mstep'_2/\mstep_2)))
          & \text{by \rprop{congruence_flateq_projection}}
        \\
        & = &
          ((\mstep_2/\mstep'_2)/(\mstep_2/\mstep'_2))/
          (\mstep_3/(\mstep'_2/\mstep_2))
          & \text{by definition of projection}
        \\
        & \flateq &
          \ftgt{(\mstep'_2)}/
          (\mstep_3/(\mstep'_2/\mstep_2))
          & \text{by \rlem{self_erasure} and \rprop{congruence_flateq_projection}}
        \\
        & \flateq &
          \ftgt{(\mstep_3/(\mstep'_2/\mstep_2))}
          & \text{by \rprop{body:properties_of_projection}}
        \end{array}
      \]
      This means that $(\mstep_2/\mstep'_2)/\mstep'_3$ is equivalent
      to an empty multistep.
      By \rlem{characterization_of_empty_multisteps}, this implies
      that $(\mstep_2/\mstep'_2)/\mstep'_3$ itself has to be empty.
      As a consequence, we have that:
      \[
        \begin{array}{rlll}
          \mstep'_3
        & \flateq &
          \mstep'_3\seq((\mstep_2/\mstep'_2)/\mstep'_3)
          & \text{as $(\mstep_2/\mstep'_2)/\mstep'_3$ is empty}
        \\
        & \flateq &
          (\mstep_2/\mstep'_2)\seq(\mstep'_3/(\mstep_2/\mstep'_2))
          & \text{by \rlem{rewrite_permutation}} %\rlem{body:multistep_permutation}}
        \\
        & = &
          (\mstep_2/\mstep'_2)\seq\mstepthree
        \end{array}
      \]
      Since $\mstep_2/\mstep'_2$ is non-empty,
      there is a \stdrule{Pull} step
      $\mstep'_{12}\seq\mstep'_3\seq\redseq
       \tostd
       \msteptwo\seq\mstepthree\seq\redseq$.
    \end{enumerate}
    With these two \stdrule{Pull} steps, we are able to close the diagram:
    \[
      \xymatrix{
        \mstep_1\seq\mstep_{23}\seq\redseq
        \ar[r]^-{\stdrule{Pull}}
        \ar[d]_-{\stdrule{Pull}}
      &
        \mstep_{12}\seq\mstep_3\seq\redseq
        \ar@{.>}[d]^-{\stdrule{Pull}}
      \\
        \mstep'_{12}\seq\mstep'_3\seq\redseq
        \ar@{.>}[r]_-{\stdrule{Pull}}
      &
        \msteptwo\seq\mstepthree\seq\redseq
      }
    \]
  \end{enumerate} 
\item \stdrule{Pull}/\stdrule{Pull} (2):
  Consider two overlapping \stdrule{Pull} steps,
  the first one of the form
  \[
    \mstep_1\seq\mstep_{23}\seq\mstep_{45}\seq\redseq
    \tostd
    \mstep_{12}\seq\mstep_3\seq\mstep_{45}\seq\redseq
  \]
  where
  $\mstep_{12} \flateq \mstep_1\seq\mstep_2$
  and
  $\mstep_{23} \flateq \mstep_2\seq\mstep_3$
  with $\mstep_2$ a non-empty multistep,
  and the second one of the form
  \[
    \mstep_1\seq\mstep_{23}\seq\mstep_{45}\seq\redseq
    \tostd
    \mstep_1\seq\mstep_{234}\seq\mstep_5\seq\redseq
  \]
  where $\mstep_{234} \flateq \mstep_{23}\seq\mstep_4$ 
  and $\mstep_{45} \flateq \mstep_4\seq\mstep_5$,
  with $\mstep_4$ a non-empty multistep.

  Let $\mstep_{34} := \mstep_{234}/\mstep_2$.
  We claim that there are two \stdrule{Pull} steps,
  $\mstep_{12}\seq\mstep_3\seq\mstep_{45} \tostd
   \mstep_{12}\seq\mstep_{34}\seq\mstep_5$
  and
  $\mstep_1\seq\mstep_{234}\seq\mstep_5 \tostd
   \mstep_{12}\seq\mstep_{34}\seq\mstep_5$.

  \begin{enumerate}
  \item
    First, note that:
    \[
      \begin{array}{rlll}
        \mstep_{34}
      & = &
        \mstep_{234}/\mstep_2
      \\
      & \flateq &
        ((\mstep_2 \seq \mstep_3)\seq\mstep_4)/\mstep_2
        & \text{by \rprop{congruence_flateq_projection}}
      \\
      & \flateq &
        \mstep_3\seq\mstep_4
        & \text{by \rlem{self_erasure} and \rprop{body:properties_of_projection}}
      \end{array}
    \]
    and recall that
    $\mstep_{45} \flateq \mstep_4\seq\mstep_5$.
    Since $\mstep_4$ is non-empty,
    this gives us a \stdrule{Pull} step
    $\mstep_{12}\seq\mstep_3\seq\mstep_{45} \tostd
     \mstep_{12}\seq\mstep_{34}\seq\mstep_5$.
  \item
    Second, recall that $\mstep_{12} \flateq \mstep_1\seq\mstep_2$
    and note that:
    \[
      \begin{array}{rlll}
        \mstep_{234}
      & \flateq &
        \mstep_{23}\seq\mstep_4
      \\
      & \flateq &
        \mstep_2\seq\mstep_3\seq\mstep_4
      \\
      & \flateq &
        \mstep_2\seq\mstep_{34}
        & \text{as $\mstep_{34} \flateq \mstep_3\seq\mstep_4$,
                as shown above}
      \end{array}
    \]
    Since $\mstep_2$ is non-empty,
    this gives us a \stdrule{Pull} step
    $\mstep_1\seq\mstep_{234}\seq\mstep_5 \tostd
     \mstep_{12}\seq\mstep_{34}\seq\mstep_5$.
  \end{enumerate}
  With these two \stdrule{Pull} steps, we are able to close the diagram:
  \[
    \xymatrix{
      \mstep_1\seq\mstep_{23}\seq\mstep_{45}\seq\redseq
      \ar[r]^-{\stdrule{Pull}}
      \ar[d]_-{\stdrule{Pull}}
    &
      \mstep_{12}\seq\mstep_3\seq\mstep_{45}\seq\redseq
      \ar@{.>}[d]^-{\stdrule{Pull}}
    \\
      \mstep_1\seq\mstep_{234}\seq\mstep_5\seq\redseq
      \ar@{.>}[r]_-{\stdrule{Pull}}
    &
      \mstep_{12}\seq\mstep_{34}\seq\mstep_5\seq\redseq
    }
  \]
\end{proof}

If $\mstep$ is a multistep, an {\em unfolding} of $\mstep$
is any sequential rewrite $\mstep_1\seq\hdots\seq\mstep_n\seq\refl{\tm}$
such that $\mstep \flateq \mstep_1\seq\hdots\seq\mstep_n\seq\refl{\tm}$
and such that the $\mstep_i$ are not empty.
The remainder of this section {\bf relies on the following assumption}
on the HRS $\hrsone$,
which allows to prove strong normalization and the standardization
theorem itself.
\begin{assumption}[Finiteness condition]
For any given multistep $\mstep$,
the length of the unfoldings
of $\mstep$ is bounded. We write $\depth{\mstep}$ for the length
of the longest unfolding of $\mstep$,
called its {\em depth}.
\end{assumption}

Any strongly normalizing HRS enjoys the finiteness condition.
Unfortunately, the untyped $\lambda$-calculus does not.
For instance, given terms
$\omega = \flam\,(\lam{\var}{\fapp\,\var\,\var})$
and
$\Omega = \fapp\,\omega\,\omega$
and the rewrite
$\redseq = \beta\,(\lam{\var}{\fapp\,\var\,\var})\,\omega : \Omega \rewto \Omega$,
then
$\beta\,(\lam{\var}{\vartwo})\,\Omega :
  \fapp\,(\flam\,(\lam{\var}{\vartwo}))\,\Omega \rewto \vartwo$
has arbitrarily long unfoldings:
$
\fapp\,(\flam\,(\lam{\var}{\vartwo}))\,\redseq
\seq \hdots \seq
\fapp\,(\flam\,(\lam{\var}{\vartwo}))\,\redseq
\seq
\beta\,(\lam{\var}{\vartwo})\,\Omega$.

\begin{prop}
\lprop{body:standardization_sn_cr}
Standardization is strongly normalizing,
and it is confluent up to $\streq$,
that is,
if $\redseq \tostd^* \redseq_1$
and $\redseq \tostd^* \redseq_2$
then there are $\redseq'_1,\redseq'_2$
such that $\redseq_1 \tostd^* \redseq'_1$
and $\redseq_2 \tostd^* \redseq'_2$
and $\redseq'_1 \streq \redseq'_2$.
\end{prop}
\begin{proof}
For the first item,
define the {\em Sekar--Ramakrishnan measure}~\cite{DBLP:journals/iandc/SekarR93}
of a sequential rewrite $\redseq$
as a tuple of non-negative integers written $\SR{\redseq}$ defined by:
\[
  \SR{\mstep_1\seq\mstep_2\seq\hdots\mstep_n\seq\refl{\tm}} =
  (n, \depth{\mstep_n}, \hdots, \depth{\mstep_2}, \depth{\mstep_1})
\]
These tuples are partially ordered by declaring
$(n,x_1,\hdots,x_n) \srgt (m,y_1,\hdots,y_m)$ to hold
if either $n > m$
or $n = m$ and $(x_1,\hdots,x_n) > (y_1,\hdots,y_n)$ holds
according to the usual lexicographic order.
This order is well-founded.
To see that $\tostd$ is strongly normalizing, it suffices to note that
it decreases the Sekar--Ramakrishnan measure.
More precisely,
$\redseq \tostd \redseqtwo$ implies $\SR{\redseq} \srgt \SR{\redseqtwo}$.
Moreover,
$\redseq \streq \redseqtwo$ implies $\SR{\redseq} = \SR{\redseqtwo}$,
which means that $\tostd$ can be lifted to $\streq$-equivalence classes.
The second item is a consequence of Newman's Lemma on
the rewriting relation $\tostd$ lifted to $\streq$-equivalence classes,
using local confluence~(\rlem{body:standardization_wcr}).
\end{proof}

\begin{thm}[Standardization]
\lthm{body:standardization_theorem}
Let $\redseq,\redseqtwo$ be sequential rewrites.
Then
$\redseq \flateq \redseqtwo$
if and only if
there exist rewrites $\redseq',\redseqtwo'$ in $\tostd$-normal form
such that $\redseq \tostd^* \redseq'$
and $\redseqtwo \tostd^* \redseqtwo'$ and $\redseq' \streq \redseqtwo'$.
\end{thm}
\begin{proof}
We prove that the following are equivalent:
\begin{enumerate}
\item $\redseq \flateq \redseqtwo$
\item $\redseq \mathrel{(\tostd \cup \tostd^{-1} \cup \streq)^*} \redseqtwo$
\item There exist rewrites $\redseq',\redseqtwo'$ in $\tostd$-normal form
      such that $\redseq \tostd^* \redseq'$
      and $\redseqtwo \tostd^* \redseqtwo'$
      and $\redseq' \streq \redseqtwo'$.
\end{enumerate}
\begin{itemize}
\item $(1 \implies 2)$
  Since $\redseq,\redseqtwo$ are sequential,
  we may work implicitly up to the \flateqRule{Assoc} rule.
  Doing this, it can be shown
  that the equivalence $\redseq \flateq \redseqtwo$
  must be derived from the reflexive, symmetric, and transitive closure
  of following variant of the \flateqRule{Perm} rule:
  \[
    \begin{array}{rllr}
        \mstepthree_1\seq\hdots\mstepthree_i\seq
        \mstep
        \seq\mstepthree_{i+1}\seq\hdots\mstepthree_n\seq\refl{\tm}
      & \flateq &
        \mstepthree_1\seq\hdots\mstepthree_i\seq
        \flatten{\mstep_1}\seq\flatten{\mstep_2}
        \seq\mstepthree_{i+1}\seq\hdots\mstepthree_n\seq\refl{\tm}
      &
        \flateqRule{Perm'}
      \\
      &&&\text{if $\judgSplit{\mstep}{\mstep_1}{\mstep_2}$}
    \end{array}
  \]
  It suffices to show that if $\redseq \flateq \redseqtwo$
  is derived using the \flateqRule{Perm'} axiom,
  then $\redseq \mathrel{(\tostd \cup \tostd^{-1} \cup \streq)^*} \redseqtwo$.
  Indeed, let $\judgSplit{\mstep}{\mstep_1}{\mstep_2}$.
  We consider two cases, depending on whether $\flatten{\mstep_2}$
  is empty or not:
  \begin{enumerate}
  \item {\bf If $\flatten{\mstep_2}$ is empty.}
    Then $\flatten{\mstep_2} = \ftgt{\mstep_2} = \fsrc{\mstep_1}$.
    Note that
    $\mstep \flateq \flatten{\mstep_1}\seq\flatten{\mstep_2}
            \flateq \flatten{\mstep_1}$
    holds by \rthm{body:soundness_completeness_flateq}, using \permeqRule{IdR}.
    Then:
    \[
      \begin{array}{rlll}
        \mstepthree_1\seq\hdots\mstepthree_i\seq
        \mstep
        \seq\mstepthree_{i+1}\seq\hdots\mstepthree_n\seq\refl{\tm}
      & \streq &
        \mstepthree_1\seq\hdots\mstepthree_i\seq
        \flatten{\mstep_1}
        \seq\mstepthree_{i+1}\seq\hdots\mstepthree_n\seq\refl{\tm}
      \\
      & \tostdinv &
        \mstepthree_1\seq\hdots\mstepthree_i\seq
        \flatten{\mstep_1}\seq\flatten{\mstep_2}
        \seq\mstepthree_{i+1}\seq\hdots\mstepthree_n\seq\refl{\tm}
        & \text{by \stdrule{Del}}
      \end{array}
    \]
  \item {\bf If $\flatten{\mstep_2}$ is non-empty.}
    On one hand,
    $\mstep \flateq \flatten{\mstep_1}\seq\flatten{\mstep_2}$.
    Moreover,
    $\flatten{\mstep_2} \flateq \flatten{\mstep_2}\seq\ftgt{\mstep_2}$
    holds by \rthm{body:soundness_completeness_flateq}, using \permeqRule{IdR}.
    Then:
    \[
      \begin{array}{rlll}
        \mstepthree_1\seq\hdots\mstepthree_i\seq
        \mstep
        \seq\mstepthree_{i+1}\seq\hdots\mstepthree_n\seq\refl{\tm}
      & \tostdinv &
        \mstepthree_1\seq\hdots\mstepthree_i\seq
        \mstep\seq\ftgt{\mstep_2}
        \seq\mstepthree_{i+1}\seq\hdots\mstepthree_n\seq\refl{\tm}
        & \text{by \stdrule{Del}}
      \\
      & \tostdinv &
        \mstepthree_1\seq\hdots\mstepthree_i\seq
        \flatten{\mstep_1}\seq\flatten{\mstep_2}
        \seq\mstepthree_{i+1}\seq\hdots\mstepthree_n\seq\refl{\tm}
        & \text{by \stdrule{Pull}}
      \end{array}
    \]
  \end{enumerate}
\item $(2 \implies 3)$
  This is an immediate consequence of~\rprop{body:standardization_sn_cr}.
\item $(3 \implies 1)$
  It suffices to show that $\tostd \ \subseteq\ \flateq$
  and $\streq\ \subseteq\ \flateq$:
  \begin{enumerate}
  \item $\tostd\ \subseteq\ \flateq$:
    Let $\redseq \tostd \redseqtwo$ and
    let us check that $\redseq \flateq \redseqtwo$.
    We consider two subcases, depending on whether
    $\redseq \tostd \redseqtwo$ is derived from
    \stdrule{Del} or \stdrule{Pull}:
    \begin{enumerate}
    \item \stdrule{Del}:
      The step is of the form:
      \[
        \mstepthree_1\seq\hdots\seq\mstepthree_i\seq
        \mstep
        \seq\mstepthree_{i+1}\seq\hdots\seq\mstepthree_n\seq\refl{\tm}
        \tostd
        \mstepthree_1\seq\hdots\seq\mstepthree_i
        \seq\mstepthree_{i+1}\seq\hdots\seq\mstepthree_n\seq\refl{\tm}
      \]
      where $\mstep$ is empty.
      Then note that $\mstep = \ftgt{\mstep} = \fsrc{\mstepthree_{i+1}}$,
      so by \rthm{body:soundness_completeness_flateq}, using \permeqRule{IdR}, 
      we have that:
      \[
        \mstepthree_1\seq\hdots\seq\mstepthree_i\seq
        \mstep
        \seq\mstepthree_{i+1}\seq\hdots\seq\mstepthree_n\seq\refl{\tm}
        \flateq
        \mstepthree_1\seq\hdots\seq\mstepthree_i
        \seq\mstepthree_{i+1}\seq\hdots\seq\mstepthree_n\seq\refl{\tm}
      \]
    \item \stdrule{Pull}:
      The step is of the form:
      \[
        \mstepthree_1\seq\hdots\seq\mstepthree_i\seq
        \mstep_1\seq\mstep_{23}
        \seq\mstepthree_{i+1}\seq\hdots\seq\mstepthree_n\seq\refl{\tm}
        \tostd
        \mstepthree_1\seq\hdots\seq\mstepthree_i\seq
        \mstep_{12}\seq\mstep_3
        \seq\mstepthree_{i+1}\seq\hdots\seq\mstepthree_n\seq\refl{\tm}
      \]
      where $\mstep_{12} \flateq \mstep_1\seq\mstep_2$
      and $\mstep_{23} \flateq \mstep_2\seq\mstep_3$
      and $\mstep_2$ is non-empty.
      Then:
      \[
        \begin{array}{rll}
        &&
          \mstepthree_1\seq\hdots\seq\mstepthree_i\seq
          \mstep_1\seq\mstep_{23}
          \seq\mstepthree_{i+1}\seq\hdots\seq\mstepthree_n\seq\refl{\tm}
        \\
        & \flateq &
          \mstepthree_1\seq\hdots\seq\mstepthree_i\seq
          \mstep_1\seq\mstep_2\seq\mstep_3
          \seq\mstepthree_{i+1}\seq\hdots\seq\mstepthree_n\seq\refl{\tm}
        \\
        & \flateq &
          \mstepthree_1\seq\hdots\seq\mstepthree_i\seq
          \mstep_{12}\seq\mstep_3
          \seq\mstepthree_{i+1}\seq\hdots\seq\mstepthree_n\seq\refl{\tm}
        \\
        \end{array}
      \]
    \end{enumerate}
  \item $\streq\ \subseteq\ \flateq$:
    Suppose that $\redseq \streq \redseqtwo$.
    This means that
    $\redseq = \mstep_1\seq\hdots\seq\mstep_n\seq\refl{\tm}$
    and $\redseqtwo = \msteptwo_1\seq\hdots\seq\msteptwo_n\seq\refl{\tm}$
    where $\mstep_i \flateq \msteptwo_i$ for all $1 \leq i \leq n$.
    Then it is immediate to conclude that
    $\mstep_1\seq\hdots\seq\mstep_n\seq\refl{\tm}
    \flateq \msteptwo_1\seq\hdots\seq\msteptwo_n\seq\refl{\tm}$.
  \end{enumerate}
\end{itemize}
\end{proof}

%%% Local Variables:
%%% mode: latex
%%% TeX-master: "main"
%%% End:

\section{Related Work and Conclusions}
\lsec{related_work_and_conclusions}

As mentioned in the introduction, proof terms were introduced by van Oostrom and de Vrijer for first-order left-linear rewrite systems to study equivalence of reductions in~\cite{DBLP:journals/entcs/OostromV02} and~\cite[Chapter 9]{Terese03}. They are inspired in Rewriting Logic~\cite{DBLP:journals/tcs/Meseguer92}. In the setting of HORs, Hilken~\cite{DBLP:journals/tcs/Hilken96} introduces rewrites for $\beta\eta$-reduction together with a notion of permutation equivalence for those rewrites. He does not study permutation equivalence for arbitrary HORs nor formulate notions of projection. Hilken does, however, justify his equations through a categorical semantics. We have already discussed Bruggink's work extensively~\cite{thesis:bruggink:08,DBLP:conf/rta/Bruggink03}. Another attempt at devising proof terms for HOR by the authors of the present paper is~\cite{DBLP:conf/ppdp/BarenbaumB20}. The latter uses a term assignment for a minimal modal logic called Logic of Proofs (LP), to model rewrites. LP is a refinement of S4 in which the modality $\Box A$ is refined to $[s]A$, where $s$ is said to be a witness to the proof of $A$. The intuition is that terms and rewrites may be seen to belong to different stages of discourse; rewrites verse about terms. Terms are typed with simple types and rewrites are typed with a modal type $[s]A$ where the term $s$ is the source term of the rewrite. However, the notion of substitution that is required for subject reduction is arguably ad-hoc. In particular, substitution of a rewrite $\rewr{\redseq}{\tm}{\tm'}{\typ}$ for $x$ in another rewrite  $\rewr{\redseqtwo}{\tmtwo}{\tmtwo'}{\typ}$ is \emph{defined} as the  composed rewrite $\redseq\subt{\var}{\tmtwo}\seq\tm'\subtr{\var}{\redseqtwo}$,  where $\redseq$ is substituted for $x$ in $\refl{\tmtwo}$ followed by $\redseqtwo$ where $\tm'$ is substituted for $x$.

\textbf{Future work.} It would be of interest to develop tools based
on the work presented here for reasoning about computations
in higher-order rewriting, as has recently been explored for 
first-order rewriting~\cite{DBLP:conf/rta/KohlM18,DBLP:conf/cade/KohlM19}.
One downside is that our
rewrites cannot be treated as terms in a higher-order rewrite system. Indeed,
rewrites are not defined modulo $\beta\eta$ (for good reason
since an expression such as $(\lam{\var}{\redseq})\,\redseqtwo$ should not be
subject to $\beta$ reduction).

One problem that should be addressed is that of formulating {\em standardization}~(see \eg~\cite[Section 8.5]{Terese03}) using rewrites.
This amounts to giving a procedure that reorders the steps of a rewrite $\redseq$,
yielding a rewrite $\redseq^*$ in which outermost steps are performed before
innermost ones.
Standardization finds canonical representatives of
$\permeq$-equivalence classes, in the sense that
$\redseq \permeq \redseqtwo$ if and only if $\redseq^* = \redseqtwo^*$.
The flattening rewrite system of~\rsec{flat_permutation_equivalence} is
a first approximation to standardization, since
$\redseq \permeq \redseqtwo$ if and only if
$\flatten{\redseq} \flateq \flatten{\redseqtwo}$.
\rsec{standardization} presents a procedure
to compute canonical representatives of $\permeq$-equivalence classes,
based on the idea of repeatedly converting
$\mstep\seq\msteptwo$ into $\mstep'\seq\msteptwo'$
whenever $\judgSplit{\msteptwo}{\mstepthree}{\msteptwo'}$
and $\judgSplit{\mstep'}{\mstep}{\mstepthree}$,
an idea reminiscent of {\em greedy decompositions}~\cite{dehornoy2015foundations}.
Unfortunately, this procedure does not always terminate,
due to the fact that rewrites may have infinitely long ``unfoldings'';
for instance, if $\rulewit : \cons \rewto \cons$
and $\rulewittwo : \consfour(\var) \rewto \constwo$
then $\rulewittwo(\cons) : \consfour(\cons) \rewto \constwo$
is equivalent to arbitrarily long rewrites of the form
$\consfour(\rulewit) \seq \hdots \seq \consfour(\rulewit) \seq \rulewittwo(\cons)$.
A terminating procedure should probably rely on a measure based on the
notion of {\em essential development}~\cite[Definition~11]{DBLP:conf/rta/Oostrom99}.

Another avenue to pursue is to characterize permutation equivalence via
{\em labelling}. The application of a rewrite step leaves a witness in the term
itself, manifested as a decoration (a label). These labels thus collect and
record the history of a computation. By comparing them one can determine
whether two computations are equivalent. Labelling equivalence for first-order
rewriting is studied by van Oostrom and de Vrijer
in~\cite{DBLP:journals/entcs/OostromV02} and~\cite[Chapter 9]{Terese03}.

We have given semantics to rewrites via Higher-Order Rewriting Logic. A categorical semantics for a similar notion of rewrite and permutation equivalence was presented by Hirshowitz~\cite{DBLP:journals/corr/Hirschowitz13} (projection equivalence and flattening are not studied though). Our $\tm\subtr{\var}{\redseq}$ is called \emph{left whiskering} and $\redseq\subt{\var}{\tm}$ \emph{right whiskering}, using the terminology of 2-category theory. These are then used to  define $\redseq\subrr{\var}{\redseqtwo}$. A precise relation between the two notions of rewrite should be investigated.

%%% Local Variables:
%%% mode: latex
%%% TeX-master: "main"
%%% End:

\bibliographystyle{alphaurl}
\bibliography{main}

% \newpage
% \appendix

% \section{Rewrites}
% \lsec{appendix:rewrites}
% \input{appendix/a01-terms_and_rewrites}

% \section{Permutation equivalence}
% \lsec{appendix:permutation_equivalence}
% \input{appendix/a02-permutation_equivalence}

% \section{Restricted $\eta$-expansion}
% \lsec{appendix:restricted_eta_expansion}
% \input{appendix/a03-restricted_eta_expansion}

% \section{Flattening}
% \lsec{appendix:flattening}
% \input{appendix/a04-flattening}

% \section{Projection for Flat Rewrites}
% \lsec{appendix:projection}
% \input{appendix/a05-projection}

% \section{Properties of Projection for Flat Rewrites}
% \input{appendix/a06-properties_of_flat_projection}

% \section{Projection for Arbitrary Rewrites}
% \lsec{appendix:generalized_projection}
% \input{appendix/a07-generalized_projection}

% \section{Standardization}
% \lsec{appendix:standardization}
% \input{appendix/a08-standardization}

\end{document}